\documentclass[a4paper,onecolumn,11pt]{scrartcl}

\usepackage{amsmath}
\usepackage{amssymb}  
\usepackage{amsthm}
\usepackage{amsfonts}
\usepackage{graphicx}
\usepackage{cancel}
\usepackage{bbm}
\usepackage{url}
\usepackage{authblk}
\usepackage{enumerate} 
\usepackage{ upgreek }
\usepackage{dsfont}
\usepackage{color}
\usepackage{subcaption}
\usepackage{makecell}
\usepackage{diagbox}
\usepackage{wrapfig}
\usepackage{rotating}
\usepackage{tabularx}

\usepackage{hyperref}
\hypersetup{
    colorlinks=true,
    linkcolor=blue,
    citecolor=red,
    filecolor=magenta,      
    urlcolor=cyan,
}

\newcommand{\abs}[1]{\ensuremath{|#1|}}

\newcommand{\ket}[1]{| #1 \rangle}

\newcommand{\bra}[1]{\langle #1 |}

\newcommand{\braket}[2]{\langle #1 | #2 \rangle}

\newcommand{\proj}[2]{| #1 \rangle\!\langle #2 |}

\newcommand{\id}{\ensuremath{\mathds{1}}}

\def\beq{\begin{equation}}
\def\eeq{\end{equation}}
\def\bq{\begin{quote}}
\def\eq{\end{quote}}
\def\ben{\begin{enumerate}}
\def\een{\end{enumerate}}
\def\bit{\begin{itemize}}
\def\eit{\end{itemize}}

\def\ra{\rightarrow}

\def\lb{\left(}
\def\rb{\right)}
\def\lset{\lbrace}
\def\rset{\rbrace}

\def\r|{\right|}
\def\lbr{\left[}
\def\rbr{\right]}
\def\ident{\textnormal{id}}
\def\one{\id}

\newcommand\C{\mathbbm{C}}

\newcommand\R{\mathbbm{R}}
\newcommand\N{\mathbbm{N}}
\newcommand\M{\mathcal{M}}

\newcommand{\Pmap}{\Pi}
\newcommand{\dm}{D}
\newcommand{\CP}{\text{CP}}
\newcommand{\coCP}{\text{coCP}}
\newcommand{\Dec}{\text{Dec}}

\theoremstyle{plain}
\newtheorem{thm}{Theorem}[section]
\newtheorem{lem}[thm]{Lemma}
\newtheorem{cor}[thm]{Corollary}

\newtheorem{defn}[thm]{Definition}

\theoremstyle{definition}

\begin{document}
\title{{Decomposable Pauli diagonal maps and Tensor Squares of Qubit Maps}}
\author{Alexander M\"uller-Hermes\thanks{Email address: muellerh@posteo.net}~}
\affil{\small{Institut Camille Jordan, Universit\'{e} Claude Bernard Lyon 1,\\ 43 boulevard du 11 novembre 1918, 69622 Villeurbanne cedex, France}}

\maketitle
\date{\today}
\begin{abstract}
It is a well-known result due to E.~St\o rmer that every positive qubit map is decomposable into a sum of a completely positive map and a completely copositive map. Here, we generalize this result to tensor squares of qubit maps. Specifically, we show that any positive tensor product of a qubit map with itself is decomposable. This solves a recent conjecture by S.~Fillipov and K.~Magadov. We contrast this result with examples of non-decomposable positive maps arising as the tensor product of two distinct qubit maps or as the tensor square of a decomposable map from a qubit to a ququart. To show our main result, we reduce the problem to Pauli diagonal maps. We then characterize the cone of decomposable ququart Pauli diagonal maps by determining all $252$ extremal rays of ququart Pauli diagonal maps that are both completely positive and completely copositive. These extremal rays split into three disjoint orbits under a natural symmetry group, and two of these orbits contain only entanglement breaking maps. Finally, we develop a general combinatorial method to determine the extremal rays of Pauli diagonal maps that are both completely positive and completely copositive between multi-qubit systems using the ordered spectra of their Choi matrices. Classifying these extremal rays beyond ququarts is left as an open problem.

\end{abstract}

\newpage 
\tableofcontents

\section{Introduction}

For $d\in\N$ we denote by $\M_d$ the set of complex $d\times d$ matrices and by $\M^+_d$ the cone of positive semidefinite matrices, simply called ``positive matrices'' in the following. A linear map $P:\M_{d_1}\ra\M_{d_2}$ is called \emph{positive} if $P(\M^+_{d_1})\subset\M^+_{d_2}$, and it is called \emph{completely positive} if $\ident_k\otimes P:\M_k\otimes \M_{d_1}\ra \M_k\otimes \M_{d_2}$ is positive for every $k\in\N$. It is called \emph{decomposable} if $P=S+\vartheta_{d_2}\circ T$, for completely positive maps $S,T:\M_{d_1}\ra\M_{d_2}$ and a matrix transpose $\vartheta_{d_2}:\M_{d_2}\ra\M_{d_2}$. Linear maps of the form $\vartheta_{d_2}\circ T$ for completely positive $T:\M_{d_1}\ra\M_{d_2}$ are also called \emph{completely copositive}. In general, the set of decomposable maps is a strict subset of the set of positive maps, but in certain low dimensions these two sets coincide. In 1963, E.~St{\o}rmer~\cite{stormer1963positive} showed that every positive qubit\footnote{A \emph{qubit} is a quantum bit, i.e.~a quantum system with a state space of poitive matrices in $\M^+_2$ with trace equal to one. Motivated by this terminology, we will sometimes use the terms \emph{qubit}, \emph{qutrit} and \emph{ququart} to refer to the matrix algebras $\M_2$, $\M_3$ and $\M_4$ respectively.} map $P:\M_{2}\ra\M_{2}$ is decomposable. In 1976, S.L.~Woronowicz~\cite{woronowicz1976positive} showed that for any $(d_1,d_2)\in\lset (2,3),(3,2)\rset$ every positive map $P:\M_{d_1}\ra\M_{d_2}$ is decomposable, and he proved nonconstructively that a non-decomposable positive map $P:\M_2\ra\M_4$ exists. The first explicit example of such a map was constructed by W.-S.~Tang in~\cite{tang1986positive}.  

Despite being only valid in small dimensions, these results became important in quantum information theory implying that a bipartite quantum state $\rho\in(\M_{d_1}\otimes \M_{d_2})^+$ with $(d_1,d_2)\in\lset (2,2),(2,3),(3,2)\rset$ is separable if and only if its partial transpose $(\ident_{d_1}\otimes \vartheta_{d_2})(\rho)$ is positive~\cite{HORODECKI19961}. Moreover, by a duality first observed in~\cite{stormer1982decomposable} examples of non-decomposable positive maps give rise to entangled quantum states with positive partial transpose~\cite{horodecki1997separability}. Such quantum states are the only known examples of bound entanglement~\cite{horodecki1998mixed}, and they have been studied in different quantum communication scenarios~\cite{horodecki2005secure,smith2008quantum,bauml2015limitations}. 

\subsection{Motivation and summary of main results}

Our main motivation is to extend St{\o}rmer's theorem on the decomposability of positive qubit maps to positive tensor squares of such maps. Our strategy is inspired by an elegant proof of St{\o}rmer's result given by G.~Aubrun and S.~Szarek in \cite{aubrun2015two} exploiting the structure of positive maps that are diagonal in the Pauli basis of $\M_2$, i.e.~the set
\begin{equation}
\sigma_1 = \begin{pmatrix} 1 & 0 \\ 0 & 1\end{pmatrix},\quad \sigma_2 = \begin{pmatrix} 0 & 1 \\ 1 & 0\end{pmatrix},\quad \sigma_3 = \begin{pmatrix} 1 & 0 \\ 0 & -1\end{pmatrix},\quad \sigma_4 = \begin{pmatrix} 0 & -i \\ i & 0\end{pmatrix}.
\label{equ:Pauli}
\end{equation}
Using a Sinkhorn-type scaling operation (cf.~\cite{gurvits2004classical} or \cite[Proposition 2.32]{aubrun2017alice}) every positive map $P:\M_{2}\ra\M_{2}$ can be transformed into a positive map that is unital, trace-preserving, and diagonal in the Pauli basis. Such a map is easily seen to be decomposable, and reversing the scaling operation shows that the original map is decomposable as well.  

\newpage

Recently, S.~Fillipov and K.~Magadov~\cite{filippov2017positive} analyzed when tensor squares of qubit maps are positive. Specifically, they showed that for the Pauli diagonal map
\[
\Pmap_{\mu}(X) = \frac{1}{2}\text{Tr}\lbr \sigma_{1}X\rbr \sigma_{1} + \frac{x}{2}\text{Tr}\lbr \sigma_{2}X\rbr \sigma_{2} + \frac{y}{2}\text{Tr}\lbr \sigma_{3}X\rbr \sigma_{3} + \frac{z}{2}\text{Tr}\lbr \sigma_{4}X\rbr \sigma_{4}
\]
with $\ket{\mu}=(1,x,y,z)^T\in\R^4$ its tensor square $\Pmap_\mu\otimes \Pmap_\mu$ is positive if and only if the following inequalities hold: 
\begin{equation}
\begin{aligned}
1+x^2 &\geq y^2 + z^2 \\
1+y^2 &\geq x^2 + z^2 \\
1+z^2 &\geq x^2 + y^2
\end{aligned} 
\label{equ:ParamCond}
\end{equation}
Figure \ref{fig:starrySurf} shows the beautiful set of parameters $(x,y,z)^T\in\R^3$ satisfying these inequalities. Note that the set contains the two tetrahedra corresponding to completely positive and completely copositive Pauli diagonal maps respectively (cf.~\cite{ruskai2002analysis}). Using the Sinkhorn-type scaling technique as in \cite{aubrun2015two}, the inequalities \eqref{equ:ParamCond} give a characterization of the linear maps $P:\M_2\ra\M_2$ for which the tensor square $P\otimes P$ is positive. 

\begin{figure*}[t!]
        \center
        \includegraphics[scale=0.7]{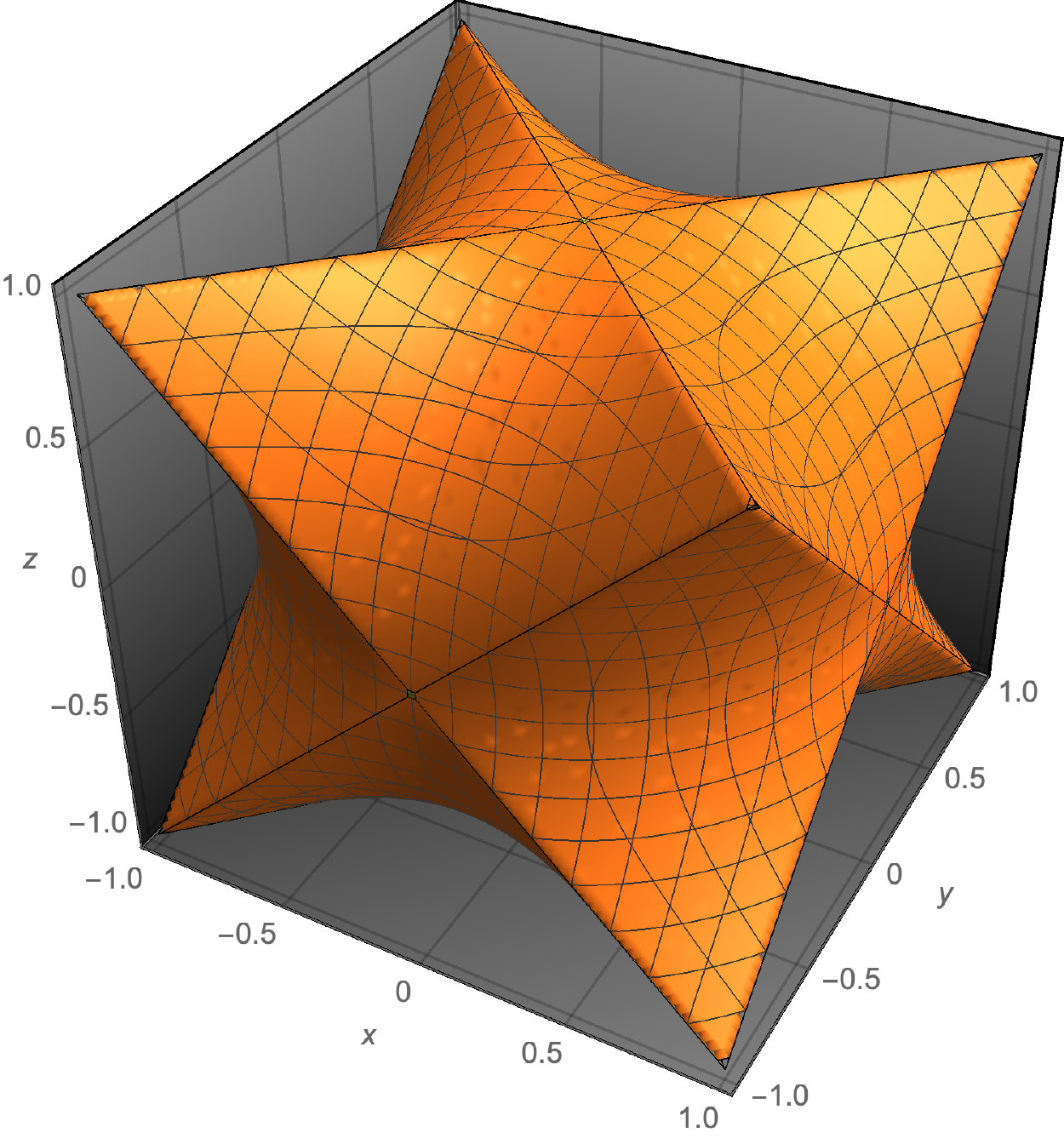}
        \caption{The parameters $(x,y,z)^T\in\R^3$ for which the tensor square $\Pi_\mu\otimes \Pi_\mu$ of the Pauli diagonal map $\Pi_\mu:\M_2\ra\M_2$ with $\ket{\mu}=(1,x,y,z)^T$ is positive. }
        \label{fig:starrySurf}
\end{figure*}

In their article~\cite{filippov2017positive}, S.~Fillipov and K.~Magadov continued to study the decomposability of tensor squares $\Pmap_\mu\otimes \Pmap_\mu$ of Pauli diagonal maps. They conjectured that in fact every positive tensor square $P\otimes P$ of a positive map $P:\M_2\ra\M_2$ is decomposable. We solve this conjecture in the affirmative. Specifically, we show the following theorem.

\begin{thm}
For a linear map $P:\M_2\ra\M_2$ the following are equivalent:
\begin{enumerate}
\item $P\otimes P$ is positive.
\item $P\otimes P$ is decomposable.
\end{enumerate}
\label{thm:Main1}
\end{thm}   

This result should be contrasted with several counterexamples to natural stronger conjectures: We give examples of decomposable maps $P:\M_2\ra\M_4$ such that the tensor square $P\otimes P$ is positive but not decomposable, and examples of positive maps $P,Q:\M_2\ra\M_2$ such that the tensor product $P\otimes Q$ is positive but not decomposable. However, we leave open the question whether there are positive maps $P:\M_2\ra\M_2$ for which a tensor power $P^{\otimes n}$ with $n\geq 3$ is positive but not decomposable.  

To prove Theorem \ref{thm:Main1} via the Sinkhorn-type scaling technique described before, but also for future research in this direction, we study the decomposability of generalized Pauli diagonal maps $\Pmap^{(2)}_\mu:\M^{\otimes 2}_2\ra\M^{\otimes 2}_2$ given by 
\[
\Pmap^{(N)}_\mu(X) = \sum_{i_1,\ldots i_N \in \lset 1,2,3,4\rset^N} \frac{\mu_{i_1\ldots i_N}}{2^N}\text{Tr}\lbr (\sigma_{i_1}\otimes \ldots \otimes\sigma_{i_N})X\rbr \sigma_{i_1}\otimes \ldots \otimes\sigma_{i_N}
\]
for parameters $\ket{\mu}\in(\R^4)^{\otimes N}$. By recasting the well-known duality from~\cite{stormer1982decomposable}, characterizing the polyhedral cone $\Dec_N$ of decomposable Pauli diagonal maps is equivalent to characterizing the polyhedral cone $\CP_N\cap \coCP_N$ of Pauli diagonal maps that are both completely positive and completely copositive. By exploiting the one-to-one correspondence between such maps and the spectra of their Choi matrices (closely related to the Fujiwara-Algoet criterion~\cite{fujiwara1999one} for qubit maps), we give a combinatorial characterization of extremal rays in terms of their zero patterns (i.e.~$(0,1)$-tensors indicating the positions of zero entries) that are maximal in a certain partial order. We could fully describe the extremal rays of the polyhedral cone $\CP_N\cap \coCP_N$ for $N=1$ (which was known before) and $N=2$ (which is new to our knowledge). This leads to a characterization of decomposable Pauli diagonal maps in terms of the spectra of their Choi matrices.

\subsection{Outline} 

Our article is structured as follows: 

\begin{itemize}
\item In Section \ref{sec:ClassPauliMult} we introduce some classes of Pauli diagonal maps.
\begin{itemize}
\item In Section \ref{subsec:CPandcoCP} we introduce the cones of completely positive and completely copositive Pauli diagonal maps and characterize them in terms of the spectra of their Choi matrices. 
\item In Section \ref{sec:DecPPTPauliMult} we introduce the cone of decomposable Pauli diagonal maps $\Dec_N$ and show that its dual is the cone $\CP_N\cap \coCP_N$ of Pauli diagonal maps that are both completely positive and completely copositive. 
\item In Section \ref{subsec:Realignment} we review the realignment criterion. Later we will use this criterion to show that certain Pauli diagonal map are not entanglement breaking.
\end{itemize}
\item In Section \ref{sec:SpectraPPTPauli} we study the cone $\CP_N\cap \coCP_N$ of Pauli diagonal maps that are both completely positive and completely copositive. The main goal is to determine the extremal rays of this cone. 
\begin{itemize}
\item 
In Section \ref{subsec:ConePPTSpec} we introduce the cone $\mathcal{S}\lb\CP_N\cap \coCP_N\rb$ of Pauli PPT spectra, i.e.~the cone of spectra of the Choi matrices of Pauli diagonal maps that are both completely positive and completely copositive. We show that the two polyhedral cones $\mathcal{S}\lb\CP_N\cap \coCP_N\rb$ and $\CP_N\cap \coCP_N$ are isomorphic.   
\item In Section \ref{subsec:ZeroPatt} we introduce a partial ordering $\leq_Z$ on zero patterns that is useful to characterize the extremal rays of the polyhedral cone $\mathcal{S}\lb\CP_N\cap \coCP_N\rb$.
\item In Section \ref{subsec:ExtremalRays} we characterize the extremal rays of $\mathcal{S}\lb\CP_N\cap \coCP_N\rb$ through zero patterns that are maximal in the partial order $\leq_Z$. Moreover, we provide bounds on the rank of the Choi matrices of Pauli diagonal maps that generate extremal rays of $\CP_N\cap \coCP_N$. Finally, we show that tensor products of generators of extremal rays for values $N_1$ and $N_2$ generate extremal rays for values $N_1N_2$. 
\item In Section \ref{subsec:OrbitsSymm} we study a natural symmetry of the cone $\mathcal{S}\lb\CP_N\cap \coCP_N\rb$, which leads to a decomposition of the set of extremal rays into a union of disjoint orbits. 
\end{itemize}
\item In Section \ref{sec:extremalRaysConcrete1} we apply the theory developed in the preceeding sections to classify the extremal rays of $\mathcal{S}\lb\CP_1\cap \coCP_1\rb$ (and thereby of $\CP_1\cap \coCP_1$). This result is well-known and we recover the $6$ extremal rays generating the cone $\CP_1\cap \coCP_1$ with octahedral base. 

\item In Section \ref{sec:extremalRaysConcrete2} we characterize all extremal rays of the cone $\mathcal{S}\lb\CP_2\cap \coCP_2\rb$ (and thereby of $\CP_2\cap \coCP_2$). This cone has $252$ extremal rays distributed into three orbits. The largest orbit corresponds to Pauli diagonal maps $\Pmap^{(2)}_\mu:\M_4\ra\M_4$ for which the -Jamiolkowski matrices $C_{\Pmap^{(2)}_\mu}$ are multiples of entangled PPT projectors on $6$-dimensional subspaces of $\C^4\otimes \C^4$. We verify the classification of the extremal rays of $\mathcal{S}\lb\CP_2\cap \coCP_2\rb$ using standard software for analyzing convex polytopes, but we also present a human-readable proof. This proof is quite lengthy and presented in Appendix \ref{app:Therest}. Finally, we fully characterize the decomposable ququart Pauli diagonal maps in $\Dec_2$ through the spectra of their Choi matrices. 
\item In Section \ref{sec:DecTensSquares} we show that every positive tensor square of a qubit map is decomposable. The proof is split over three sections:
\begin{itemize}
\item In Section \ref{subsec:RedPauliMult} we show how to reduce membership questions of tensor products of qubit maps in mapping cones to the corresponding membership question of tensor products of Pauli diagonal maps. 
\item In Section \ref{subsec:SymmPauliMult} we review some useful symmetries of Pauli diagonal maps. 
\item In Section \ref{subsec:DecTensSquares} we prove that every positive tensor square of a qubit Pauli diagonal map is decomposable. Applying the reduction technique from Section \ref{subsec:RedPauliMult} we then show that every positive tensor square of a qubit map is decomposable.
\end{itemize} 
\item In Section \ref{Sec:NDPosFromTens} we provide several constructions of non-decomposable positive maps arising as tensor products of decomposable maps. 
\begin{itemize}
\item In Section \ref{subsec:NDTensProd} we show that tensor products of qubit maps can be non-decomposable and positive. 
\item In Section \ref{subsec:NDTensSquares} we give an example of a decomposable map $P:\M_2\ra\M_4$ for which the tensor square $P\otimes P$ is positive but not decomposable. 
\end{itemize}
\item In Section \ref{sec:Conclusion} we conclude with some open problems. 
\end{itemize}

\section{Notation}

We will write $\M_d$ to denote the set of complex $d\times d$ matrices, and more generally we will write $\M_d(\mathcal{S})$ to denote the set of $d\times d$ matrices with entries in a set $\mathcal{S}$. We denote by $\M^+_d\subset\M_d$ the cone of positive semidefinite matrices, which we simply call ``positive'' in the following. Furthermore, we say that a matrix $M\in \M_d$ is entrywise positive if $M_{ij}\geq 0$ for any $i,j\in\lset 1,2,\ldots ,d\rset$, and we will write ``ew-positive'' as an abbreviation.       

We will denote by $\one_d\in\M_d$ the $d\times d$ identity matrix, by $\omega_d\in (M_d\otimes\M_d)^+$ the (unnormalized) maximally entangled state, i.e.~$\omega_d=\proj{\Omega_d}{\Omega_d}$ for $\Omega_d=\sum^d_{i=1}\ket{i}\otimes \ket{i}$, and by $\mathbbm{F}_d\in\M_d\otimes \M_d$ the flip operator defined as $\mathbbm{F}_d\lb \ket{i}\otimes \ket{j}\rb = \ket{j}\otimes \ket{i}$. Here, we denote by $\lset\ket{i}\rset^d_{i=1}\subset \C^d$ the computational basis, i.e.~the vector $\ket{i}$ has a single $1$ in the $i$th position and zeros in the remaining entries. We denote by $\sigma_1, \sigma_2 ,\sigma_3 ,\sigma_4\in\M_2$ the Pauli matrices in the (slightly unusual) order stated in \eqref{equ:Pauli}.

We denote by $\ident_d:\M_d\ra\M_d$ the identity map $\ident_d(X)=X$ and by $\vartheta_d:\M_d\ra\M_d$ the matrix transpose $\vartheta_d(X)=X^T$ in the computational basis (our results will not depend on this choice of basis). Given a linear map $L:\M_{d_1}\ra\M_{d_2}$ we denote its Choi matrix~\cite{choi1975completely} by 
\[
C_L = (\ident_{d_1}\otimes L)(\omega_{d_1})\in \M_{d_1}\otimes \M_{d_2}.
\]   
It is well-known~\cite{choi1975completely} that $L$ is completely positive if and only if $C_L$ is positive, and that $L$ is completely copositive if and only if the partial transpose $C^{\Gamma}_{L}:=C_{\vartheta_{d_2}\circ L}$ is positive.

A positive matrix $C\in \lb \M_{d_1}\otimes \M_{d_2}\rb^+$ is called \emph{separable} if it can be written as $C=\sum^k_{i=1} A_i\otimes B_i$ for some $k\in\N$ and positive matrices $\lset A_i\rset^{k}_{i=1}\subset \M^+_{d_1}$ and $\lset B_i\rset^{k}_{i=1}\subset \M^+_{d_2}$. A linear map $L:\M_{d_1}\ra\M_{d_2}$ is called \emph{entanglement breaking}~\cite{horodecki2003entanglement} if its Choi matrix $C_L$ is separable. Recall that a linear map $L:\M_{d_1}\ra\M_{d_2}$ is entanglement breaking if and only if $\text{Tr}\lb C_L C_P\rb\geq 0$ for every positive map $P:\M_{d_1}\ra\M_{d_2}$ (see~\cite{HORODECKI19961}). Similarly, it follows from~\cite{stormer1982decomposable} that a linear map $P:\M_{d_1}\ra\M_{d_2}$ is decomposable if and only if $\text{Tr}\lb C_T C_P\rb\geq 0$ for every linear map $T:\M_{d_1}\ra\M_{d_2}$ that is both completely positive and completely copositive.

\section{Classes of Pauli diagonal maps}
\label{sec:ClassPauliMult}

We will begin by introducing the classes of Pauli diagonal maps needed in our article. 

\subsection{Completely positive and completely copositive Pauli diagonal maps}
\label{subsec:CPandcoCP}
Let us recall the following definition from the introduction. 

\begin{defn}[Pauli diagonal maps]
For every $\ket{\mu}\in(\R^{4})^{\otimes N}$ we define the $N$th order Pauli diagonal map $\Pmap^{(N)}_\mu:\M^{\otimes N}_2\ra\M^{\otimes N}_2$ by 
\[
\Pmap^{(N)}_\mu(X) = \sum_{i_1,\ldots i_N \in \lset 1,2,3,4\rset^N} \frac{\mu_{i_1\ldots i_N}}{2^N}\text{Tr}\lbr (\sigma_{i_1}\otimes \ldots \otimes\sigma_{i_N})X\rbr \sigma_{i_1}\otimes \ldots \otimes\sigma_{i_N}.
\]
\label{defn:PauliMult}
\end{defn}

Let $\Pmap^{(N)}_{\mu}:\M^{\otimes N}_2\ra \M^{\otimes N}_2$ denote a Pauli diagonal map with parameters $\ket{\mu}\in(\R^{4})^{\otimes N}$. After reshuffling the tensor factors, its Choi matrix is given by 
\[
C_{\Pmap^{(N)}_\mu} = \sum_{i_1\cdots i_N}\frac{\mu_{i_1\cdots i_N}}{2^N} ~\sigma_{i_1}\otimes \sigma^T_{i_1}\otimes \cdots \otimes \sigma_{i_N}\otimes \sigma^T_{i_N}.
\]
By the elementary commutation relations
\[
\lbr \sigma_i\otimes \sigma_i, \sigma_j\otimes \sigma_j\rbr = 0
\]
for $i,j\in\lset 1,2,3,4\rset$ the terms in the previous sum (and in its partial transpose) commute and they can be simultaneously diagonalized. Consequently, complete positivity or complete copositivity of $\Pmap^{(N)}_\mu$, which is equivalent (see \cite{choi1975completely}) to positivity of $C_{\Pmap^{(N)}_\mu}$ or $C_{\vartheta^{\otimes N}_2\circ \Pmap^{(N)}_\mu}$ respectively, can be easily checked by transforming the coefficient vector $\ket{\mu}$ to the corresponding spectral vector. Specifically, we introduce matrices $\dm,\tilde{\dm}\in \M_4$ as  
\begin{equation}
\dm = \frac{1}{2}\begin{pmatrix} 1 & 1 & 1 & 1 \\ 1 & 1 & -1 & -1 \\ 1 & -1 & 1 & -1 \\ 1 & -1 & -1 & 1\end{pmatrix}\hspace{0.3cm}\text{ and } \hspace{0.3cm}\tilde{\dm} = \frac{1}{2}\begin{pmatrix} 1 & -1 & -1 & -1 \\ 1 & -1 & 1 & 1 \\ 1 & 1 & -1 & 1 \\ 1 & 1 & 1 & -1\end{pmatrix}, 
\label{equ:DM}
\end{equation}
such that for $i\in\lset 1,2,3,4\rset$ we have
\begin{align*}
\dm\ket{i} &= \frac{1}{2}~\text{spec}\lb \sigma_i\otimes \sigma^T_i\rb,\\
\tilde{\dm}\ket{i} &= \frac{1}{2}~\widetilde{\text{spec}}\lb \sigma_i\otimes \sigma_i\rb .
\end{align*}
Here, $\text{spec}\lb\cdot\rb$ and $\widetilde{\text{spec}}\lb\cdot\rb$ denote the orderings of the spectra as in the diagonal matrices $U^{\dagger}_1(\sigma_i\otimes \sigma^T_i) U_1$ and $U^{\dagger}_2(\sigma_i\otimes \sigma_i) U_2$ for $i\in \lset 1,2,3,4\rset$ and with unitaries 
\[
U_1=\frac{1}{\sqrt{2}}\begin{pmatrix} 1 & 0 & 1 & 0 \\ 0 & 1 & 0 & 1 \\ 0 & 1 & 0 & -1 \\ 1 & 0 & -1 & 0\end{pmatrix} \quad \text{ and }\quad U_2=\frac{1}{\sqrt{2}}\begin{pmatrix} 0 & 1 & 0 & 1 \\ 1 & 0 & 1 & 0 \\ -1 & 0 & 1 & 0 \\ 0 & -1 & 0 & 1\end{pmatrix},
\]
containing the Bell states in their columns. The orderings in $\text{spec}\lb\cdot\rb$ and $\widetilde{\text{spec}}\lb\cdot\rb$ are fixed throughout our article. The following theorem is a generalization of the well-known Fujiwara-Algoet criterion for complete positivity of unital qubit maps~\cite{fujiwara1999one}:

\begin{thm}
For $\ket{\mu}\in(\R^{4})^{\otimes N}$ we have
\[ \text{spec}\lb C_{\Pmap^{(N)}_\mu}\rb = \dm^{\otimes N}\ket{\mu} \quad\text{ and }\quad \widetilde{\text{spec}}\lb C_{\vartheta^{\otimes N}_2\circ\Pmap^{(N)}_\mu}\rb = \tilde{\dm}^{\otimes N}\ket{\mu},\]
with the ordering of the spectrum described above. In particular we have:
\begin{enumerate}
\item The map $\Pmap^{(N)}_\mu:\M^{\otimes N}_2\ra\M^{\otimes N}_2$ is completely positive if and only if the vector $\dm^{\otimes N}\ket{\mu}$ is entrywise positive. 
\item The map $\Pmap^{(N)}_\mu:\M^{\otimes N}_2\ra\M^{\otimes N}_2$ is completely copositive if and only if the vector $\tilde{\dm}^{\otimes N}\ket{\mu}$ is entrywise positive.
\end{enumerate}
\label{thm:CPCond}
\end{thm}

In the following, we denote the set of parameters corresponding to completely positive Pauli diagonal maps by 
\begin{align*}
\CP_N &:= \Big\lset \ket{\mu}\in(\R^{4})^{\otimes N}~\Big{|}~\sum_{i_1\cdots i_N}\frac{\mu_{i_1\cdots i_N}}{2^N} ~\sigma_{i_1}\otimes \sigma^T_{i_1}\otimes \cdots \otimes \sigma_{i_N}\otimes \sigma^T_{i_N}\geq 0\Big\rset \\
&= \text{cone}\Big\lset (\dm^T)^{\otimes N}\ket{k_1\cdots k_N} ~\Big{|}~k_1,\ldots ,k_N\in \lset1,2,3,4\rset\Big\rset,
\end{align*}
where the second equality follows from Theorem \ref{thm:CPCond}. Similarly, we denote the set of parameters corresponding to completely copositive Pauli diagonal maps by 
\begin{align*}
\coCP_N &:= \Big\lset \ket{\mu}\in(\R^{4})^{\otimes N}~\Big{|}~\sum_{i_1\cdots i_N}\frac{\mu_{i_1\cdots i_N}}{2^N} ~\sigma_{i_1}\otimes \sigma_{i_1}\otimes \cdots \otimes \sigma_{i_N}\otimes \sigma_{i_N}\geq 0\Big\rset \\
&= \text{cone}\Big\lset (\tilde{\dm}^T)^{\otimes N}\ket{k_1\cdots k_N} ~\Big{|}~k_1,\ldots ,k_N\in \lset1,2,3,4\rset\Big\rset.
\end{align*}

\subsection{Decomposable and PPT Pauli diagonal maps}
\label{sec:DecPPTPauliMult}

To describe the set of decomposable Pauli diagonal maps we will first consider the intersection of the sets $\CP_N$ and $\coCP_N$, i.e.~the Pauli diagonal maps that are both completely positive and completely copositive. 

\begin{defn}[PPT Pauli diagonal maps]
We will call a vector $\ket{\mu}\in(\R^{4})^{\otimes N}$ PPT (abbreviating Positive Partial Transpose), if the corresponding Pauli diagonal map $\Pmap^{(N)}_\mu$ is both completely positive and completely copositive, or equivalently if 
\[
\ket{\mu}\in \CP_N\cap\coCP_N.
\]
   
\end{defn}

Next, we define the notion of decomposability: 

\begin{defn}[Decomposable Pauli diagonal maps]
We will call a vector $\ket{\mu}\in(\R^{4})^{\otimes N}$ decomposable, if the corresponding Pauli diagonal map $\Pmap^{(N)}_\mu$ is decomposable, i.e.~if it can be written as a sum $\Pmap^{(N)}_\mu = T + \vartheta^{\otimes N}_2\circ S$ with $T,S:\M^{\otimes N}_2\ra\M^{\otimes N}_2$ completely positive. Furthermore, we set
\[
\Dec_N := \lset\ket{\mu}\in(\R^{4})^{\otimes N} ~:~ \ket{\mu}\text{ is decomposable} \rset .
\]

\end{defn}

To simplify the previous definition, we will need the following lemma extending the well-known duality relation between decomposable maps and maps that are both completely positive and completely copositive~\cite{stormer1982decomposable} to the case of Pauli diagonal maps:

\begin{lem}[Duality]
We have 
\[
(\CP_N\cap\coCP_N)^* := \lset \ket{\mu}\in(\R^{4})^{\otimes N}~\Big{|}~\braket{\mu}{\kappa}\geq 0\text{ for any }\ket{\kappa}\in\CP_N\cap\coCP_N\rset = \CP_N+\coCP_N,
\]
where $\braket{\cdot}{\cdot}$ denotes the Euclidean inner product on $(\R^4)^{\otimes N}$.
\label{lem:Duality}
\end{lem}

\begin{proof}
We will first show that 
\[
(\CP_N)^* = \lset \ket{\mu}\in(\R^{4})^{\otimes N}~\Big{|}~\braket{\mu}{\kappa}\geq 0\text{ for any }\ket{\kappa}\in\CP_N\rset = \CP_N.
\]
and
\[
(\coCP_N)^* = \lset \ket{\mu}\in(\R^{4})^{\otimes N}~\Big{|}~\braket{\mu}{\kappa}\geq 0\text{ for any }\ket{\kappa}\in\coCP_N\rset = \coCP_N.
\]
Consider vectors $\ket{\mu},\ket{\kappa}\in \CP_N$. By Theorem \ref{thm:CPCond} we know that $D^{\otimes N}\ket{\mu}, D^{\otimes N}\ket{\kappa}$ are ew-positive, where $D\in\M_4$ is the orthogonal matrix defined in \eqref{equ:DM}. Therefore, we have that 
\[
\braket{\mu}{\kappa} = \bra{\mu}(D^TD)^{\otimes N}\ket{\kappa}\geq 0.
\]
Being true for any $\ket{\kappa}\in \CP_N$, this shows that $(\CP_N)^* \supseteq \CP_N$. To show equality, consider a vector $\ket{\mu}\in (\CP_N)^*$. Since $\ket{\kappa_{k_1,\ldots ,k_N}}:=(D^T)^{\otimes N}\ket{k_1 k_2,\ldots k_N}\in \CP_N$ we find that 
\[
0\leq \braket{\mu}{\kappa_{k_1,\ldots ,k_N}} = \bra{\mu}(D^T)^{\otimes N}\ket{k_1 k_2,\ldots k_N}
\]
for any $k_1,\ldots ,k_N\in\lset 1,2,3,4\rset$. Therefore, $D^{\otimes N}\ket{\mu}$ is ew-positive showing that $\ket{\mu}\in \CP_N$. The proof of $(\coCP_N)^* = \coCP_N$ follows in the same way. Now, by elementary convex analysis we find
\[
(\CP_N\cap \coCP_N)^* = \overline{\CP_N+ \coCP_N}.
\]
Finally, note that the sum $\CP_N+ \coCP_N$ is closed (see e.g.~\cite[Corollary 9.1.2]{rockafellar1970convex}), since both cones $\CP_N$ and $\coCP_N$ are contained in the pointed cone of parameters corresponding to positive Pauli diagonal maps.  

\end{proof}

Now we can state the final characterization of decomposable Pauli diagonal maps: 

\begin{thm}[Decomposable Pauli diagonal maps]
We have
\[
\Dec_N = \CP_N + \coCP_N = (\CP_N\cap\coCP_N)^*.
\]
\end{thm}

\begin{proof}
For any $\ket{\mu}\in \Dec_N$ and $\ket{\kappa}\in \CP_N\cap \coCP_N$ we have
\[
\braket{\mu}{\kappa} = \text{Tr}\lb C_{\Pi^{(N)}_{\mu}}C_{\Pi^{(N)}_{\kappa}}\rb \geq 0,
\]  
since $\Pi^{(N)}_{\mu}$ is a decomposable map and $\Pi^{(N)}_{\kappa}$ is both completely positive and completely copositive (see \cite{stormer1982decomposable}). This shows that $\Dec_N\subseteq (\CP_N\cap \coCP_N)^*=\CP_N + \coCP_N$ (see Lemma \ref{lem:Duality}). The other inclusion is clear, since every $\ket{\mu}\in \CP_N + \coCP_N$ is decomposable.

\end{proof}

The previous theorem shows that a Pauli diagonal map is decomposable if and only if it can be written as a sum of a completely positive Pauli diagonal map and a completely copositive Pauli diagonal map.

\subsection{Realignment of Pauli diagonal maps}
\label{subsec:Realignment}

To determine Pauli diagonal maps $\Pmap^{(N)}_\mu:\M^{\otimes N}_2\ra\M^{\otimes N}_2$ that are not entanglement breaking, it will be useful to consider the realignement criterion~\cite{rudolph2000separability,rudolph2005further,chen2002matrix}. To make our presentation self-contained, we will first introduce this entanglement criterion in the special case of multi-qubit quantum states: 

\begin{thm}[Realignment criterion~\cite{rudolph2000separability,rudolph2005further,chen2002matrix}]\hfill\\
The qubit realignment map $R:\M_2\otimes \M_2\ra \M_2\otimes \M_2$ is given by
\[
R(\proj{i}{j}\otimes \proj{k}{l}) = \proj{i}{k}\otimes \proj{j}{l}
\]
for any $i,j,k,l\in\lset 1,2\rset$. If a quantum state 
\[
\rho_{A_1B_1A_2B_2\cdots A_NB_N}\in (\M_{2}\otimes \M_2)^{\otimes N}
\]
is separable with respect to the bipartition $(A_1A_2\cdots A_N):(B_1B_2\cdots B_N)$, then we have
\[
\| R^{\otimes N}\lb \rho_{A_1B_1A_2B_2\cdots A_NB_N}\rb\|_1\leq 1.
\]
\label{thm:RealignmentBasic}
\end{thm}

The realignment criterion takes a very simple form, when it is applied to the Choi matrix of a Pauli diagonal map. This is a special case of a more general result~\cite{lupo2008bipartite} expressing the realignment criterion in terms of the operator Schmidt coefficients of a bipartite quantum state. To make our presentation self-contained, we will present a proof.    

\begin{thm}[Realignment of Pauli diagonal maps]
Let $\Pmap^{(N)}_\mu:\M^{\otimes N}_2\ra\M^{\otimes N}_2$ denote a unital and trace-preserving Pauli diagonal map with parameters $\ket{\mu}\in (\R^4)^{\otimes N}$. If $\Pmap^{(N)}_\mu$ is entanglement breaking, then we have 
\[
\sum_{i_1\cdots i_N} |\mu_{i_1\cdots i_N}|\leq 2^N.
\]
\label{thm:RealignementPauliMultipliers}
\end{thm}

\begin{proof}
It is easy to compute the action of the realignment map on the tensor products $\sigma_i\otimes \sigma^T_i$ for $i\in\lset 1,2,3,4\rset$. We have
\[
\frac{1}{2} R\lb \sigma_i\otimes \sigma^T_i\rb = \proj{\phi_i}{\phi_i} \in (\M_2\otimes \M_2)^+
\]
where we introduced the Bell states $\phi_i = \proj{\phi_i}{\phi_i}$ given by
\begin{align*}
\ket{\phi_1} &= \ket{\Omega_2}, \\
\ket{\phi_2} &= (\one_2\otimes \sigma_2)\ket{\Omega_2},\\
\ket{\phi_3} &= (\one_2\otimes \sigma_3)\ket{\Omega_2}, \\
\ket{\phi_4} &= (\one_2\otimes i\sigma_4)\ket{\Omega_2} .
\end{align*}
Therefore, we find that 
\[
R^{\otimes N}\lb C_{\Pmap^{(N)}_\mu}\rb = \sum_{i_1\cdots i_N} \mu_{i_1\cdots i_N} \proj{\phi_{i_1}}{\phi_{i_1}}\otimes \proj{\phi_{i_2}}{\phi_{i_2}}\otimes \cdots \proj{\phi_{i_N}}{\phi_{i_N}},
\]
for any $\ket{\mu}\in (\R^4)^{\otimes N}$. Since $\ket{\phi_i}\perp \ket{\phi_j}$ for any $i\neq j$, we have that 
\[
\| R^{\otimes N}\lb C_{\Pmap^{(N)}_\mu}\rb\|_1 = \sum_{i_1\cdots i_N} |\mu_{i_1\cdots i_N}|.
\]
Now, the statement of the theorem follows from Theorem \ref{thm:RealignmentBasic} and by noting that $\text{Tr}\lb C_{\Pmap^{(N)}_\mu}\rb=2^N$ since $\Pmap^{(N)}_\mu$ is unital and trace-preserving.

\end{proof}

\section{Spectral characterization of PPT Pauli diagonal maps}
\label{sec:SpectraPPTPauli}

\subsection{The cone of Pauli PPT spectra}
\label{subsec:ConePPTSpec}

Given $\ket{\mu}\in \CP_N\cap\coCP_N$ we have that both the Choi matrices $C_{\Pmap^{(N)}_\mu}$ and $C_{\vartheta^{\otimes N}_2\circ\Pmap^{(N)}_\mu}$ are positive, and we denote their respective spectral vectors (ordered as above) by 
\[
\ket{p}=\text{spec}\lb C_{\Pmap^{(N)}_\mu}\rb \quad\text{ and }\quad \ket{q}=\widetilde{\text{spec}}\lb C_{\vartheta^{\otimes N}_2\circ\Pmap^{(N)}_\mu}\rb.
\]
Both $\ket{p}$ and $\ket{q}$ are ew-positive, and by Theorem \ref{thm:CPCond} they satisfy
\begin{equation}
\ket{\mu} = (\dm^T)^{\otimes N}\ket{p} = (\tilde{\dm}^T)^{\otimes N}\ket{q}.
\label{equ:pqrelation}
\end{equation}
We can now consider the unitary matrix
\begin{equation}
K := \tilde{\dm}\dm^T = \frac{1}{2}\begin{pmatrix}-1 & 1 & 1 & 1 \\ 1 & -1 & 1 & 1 \\ 1 & 1 & -1 & 1 \\ 1 & 1 & 1 & -1\end{pmatrix}.
\label{equ:MatK}
\end{equation}
By \eqref{equ:pqrelation} and using that $\tilde{\dm}$ is an orthogonal matrix we find that 
\[
\ket{q} = K^{\otimes N} \ket{p}.
\]
Conversely, for any pair of ew-positive vectors $\ket{p},\ket{q}\in(\R^4)^{\otimes N}$ satisfying the previous equation we have that $\ket{\mu}=(\dm^T)^{\otimes N}\ket{p}\in \CP_N\cap\coCP_N$. Moreover, in this case $\ket{p}$ and $\ket{q}$ contain the spectra of $C_{\Pmap^{(N)}_\mu}$ and $C_{\vartheta^{\otimes N}_2\circ\Pmap^{(N)}_\mu}$ respectively. This motivates the following definition: 

\begin{defn}[Pauli PPT spectra]
For $N\in\N$ we define the set 
\begin{align*}
\mathcal{S}&\lb \CP_N\cap\coCP_N\rb \\
&:=\Big\lset \lb\text{spec}\lb C_{\Pmap^{(N)}_\mu}\rb,\widetilde{\text{spec}}\lb C_{\vartheta^{\otimes N}_2\circ\Pmap^{(N)}_\mu}\rb\rb\in(\R^{4})^{\otimes N}\oplus (\R^{4})^{\otimes N}~\Big{|}~\ket{\mu}\in \CP_N\cap \coCP_N\Big\rset \\
&=\Big\lset \lb\ket{p},\ket{q}\rb\in(\R^{4})^{\otimes N}\oplus (\R^{4})^{\otimes N}~\Big{|}~\ket{p},\ket{q}\text{ ew-positive and } K^{\otimes N}\ket{p}=\ket{q}\Big\rset
\end{align*}
called the set of Pauli PPT spectra. 
\label{defn:SPPPoly}
\end{defn}

By the previous discussion we have 
\[
\CP_N\cap\coCP_N = \Big\lset (\dm^T)^{\otimes N}\ket{p}~\Big{|}~\lb\ket{p}, K^{\otimes N}\ket{p}\rb\in \mathcal{S}\lb \CP_N\cap\coCP_N\rb\Big\rset,
\]
and since $K$ is a unitary the extremal rays of the polyhedral cone $\mathcal{S}\lb \CP_N\cap\coCP_N\rb$ correspond to the extremal rays of $\CP_N\cap\coCP_N$. We conclude this section with a lemma expressing the set of Pauli PPT spectra as an intersection of a subspace and the positive orthant. Its proof is immediate from the previous definition.  

\begin{lem}[Pauli PPT spectra as subspace intersecting positive orthant]
We have 
\begin{align*}
\mathcal{S}&\lb \CP_N\cap\coCP_N\rb \\
&= E_1(M_N)^T\cap\text{cone}\Big\lset \lb \ket{k_1\cdots k_N},\ket{k_{N+1}\cdots k_{2N}}\rb ~\Big{|}~k_1,\ldots ,k_{2N}\in \lset1,2,3,4\rset\Big\rset,
\end{align*}
where $E_1(M_N)$ denotes the eigenspace corresponding to the eigenvalue $1$ of the matrix
\[
M_N:=\begin{pmatrix} 0 & K^{\otimes N} \\ K^{\otimes N} & 0\end{pmatrix}.
\]
\label{lem:SubspaceCharact}
\end{lem}

By Lemma \ref{lem:Duality} the cone of decomposable Pauli diagonal maps $\Dec_N = \CP_N \vee \coCP_N$ is dual to the cone of PPT Pauli diagonal maps $\CP_N \cap \coCP_N$. Therefore, we can state the following spectral characterization of decomposability of Pauli diagonal maps, which follows immediately from the previous discussion. 

\begin{lem}[Spectral conditions of decomposability of Pauli diagonal maps]
The Pauli diagonal map $\Pmap^{(N)}_{\mu}$ with parameter $\ket{\mu}\in\lb\R^{4}\rb^{\otimes N}$ is decomposable if and only if the spectral vector 
\[
\text{spec}\lb C_{\Pmap^{(N)}_{\mu}}\rb = \dm^{\otimes N}\ket{\mu} =: \ket{s} \in(\R^4)^{}
\]
satisfies $\braket{s}{p}\geq 0$ for any $\ket{p}\in(\R^4)^{\otimes N}$ such that 
\[
\lb \ket{p},K^{\otimes N}\ket{p}\rb\in \mathcal{S}\lb \CP_N\cap\coCP_N\rb.
\]
In particular, one can restrict to the $\ket{p}$ arising from extremal rays of $\mathcal{S}\lb \CP_N\cap\coCP_N\rb$.
\label{lem:SpectrDecCondGen}
\end{lem}

In the following sections we will study the extremal rays of the cone $\mathcal{S}\lb \CP_N\cap\coCP_N\rb$ and develop some tools to characterize them. In Section \ref{sec:extremalRaysConcrete1} and Section \ref{sec:extremalRaysConcrete2} we will apply these tools to find all extremal rays of $\mathcal{S}\lb \CP_N\cap\coCP_N\rb$ in the cases $N=1$ and $N=2$, which by Lemma \ref{lem:SpectrDecCondGen} give a complete characterization of decomposable Pauli diagonal maps in these cases.

\subsection{Zero patterns and their ordering}
\label{subsec:ZeroPatt}

To characterize the extremal rays of the set of Pauli PPT spectra $\mathcal{S}\lb \CP_N\cap\coCP_N\rb$, see Definition \ref{defn:SPPPoly}, we will need the following definition. 

\begin{defn}[Zero pattern]
The \emph{zero pattern} of a vector $\ket{p}\in(\R^{4})^{\otimes N}$ is defined as
\[
\mathcal{Z}\lb \ket{p}\rb := \lset \lb k_1,\ldots ,k_N\rb\in \lset 1,2,3,4\rset^{N} ~|~ \braket{k_1\cdots k_N}{p}=0 \rset
\] 
i.e.~the set of all indices where the vector $\ket{p}$ has zeros in the computational basis. We extend this definition to pairs $\lb\ket{p},\ket{q}\rb\in(\R^{4})^{\otimes N}\oplus(\R^{4})^{\otimes N}$ by 
\[
\mathcal{Z}\lb \ket{p},\ket{q}\rb := \Big(\mathcal{Z}\lb \ket{p}\rb,\mathcal{Z}\lb \ket{q}\rb \Big).
\]
\label{defn:ZPattern}

\end{defn}

Sets are partially ordered by inclusion and we can extend this partial ordering to pairs of vectors via their zero pattern:

\begin{defn}[Zero partial ordering]
Given $\lb \ket{p_1},\ket{q_1}\rb,\lb \ket{p_2},\ket{q_2}\rb\in(\R^{4})^{\otimes N}\oplus (\R^{4})^{\otimes N}$ we write $\lb \ket{p_1},\ket{q_1}\rb\leq_Z \lb \ket{p_2},\ket{q_2}\rb$ if $\mathcal{Z}\lb \ket{p_1},\ket{q_1}\rb\subseteq \mathcal{Z}\lb \ket{p_2},\ket{q_2}\rb$, and $\lb \ket{p_1},\ket{q_1}\rb<_Z \lb \ket{p_2},\ket{q_2}\rb$ if $\mathcal{Z}\lb \ket{p_1},\ket{q_1}\rb\subsetneq \mathcal{Z}\lb \ket{p_2},\ket{q_2}\rb$. Here we say $(A,B)\subset (C,D)$ if and only if $A\subset B$ and $C\subset D$. 
\end{defn}

It will be convenient to sometimes express the zero partial ordering in terms of subspaces. For this we need the following definition: 

\begin{defn}[Subspace associated to zero pattern] 
Given $\lb \ket{p},\ket{q}\rb\in(\R^{4})^{\otimes N}\oplus (\R^{4})^{\otimes N}$, we define the subspace
\[
V^N_{p,q} := \Big\lbrace \lb\ket{p'},\ket{q'}\rb\in(\R^{4})^{\otimes N}\oplus (\R^{4})^{\otimes N} ~:~ \lb\ket{p},\ket{q}\rb\leq_Z \lb \ket{p'},\ket{q'}\rb\Big\rbrace, 
\]
of all vectors with zero patterns containing $\mathcal{Z}\lb \ket{p},\ket{q}\rb$.
\label{defn:Subspacepq}
\end{defn}

The following lemma will be useful to simplify the search for extremal rays.

\begin{lem}[Orthogonality]
For $(\ket{p_1},\ket{q_1}),(\ket{p_2},\ket{q_2})\in \mathcal{S}\lb \CP_N\cap\coCP_N\rb$ the following are equivalent:
\begin{enumerate}
\item $\braket{p_1}{p_2} = 0$.
\item $\braket{q_1}{q_2} = 0$.
\item $\mathcal{Z}\lb \ket{p_2},\ket{q_2}\rb \supset \mathcal{Z}\lb \ket{p_1},\ket{q_1}\rb^c$ using the convention $(A,B)^c:=(A^c,B^c)$ for sets $A$ and $B$.
\end{enumerate} 
\label{lem:Orthogonality}
\end{lem}
\begin{proof}
Note that the matrix $K\in\M_4$ introduced in \eqref{equ:MatK} is symmetric and unitary. By the definition of $\mathcal{S}\lb \CP_N\cap\coCP_N\rb$ we have  
\[
\braket{p_1}{p_2} = \bra{p_1}K^{\otimes N}K^{\otimes N}\ket{p_2} = \braket{q_1}{q_2}.
\] 
Since $\ket{p_1},\ket{q_1},\ket{p_2}$ and $\ket{q_2}$ are ew-positive, we have $\braket{p_1}{p_2} = \braket{q_1}{q_2} = 0$ if and only if $\braket{i_1,\ldots ,i_N}{p_2}=\braket{j_1,\ldots ,j_N}{q_2}=0$ for any $i_1,\ldots i_N,j_1,\ldots j_N\in\lset 1,2,3,4\rset$ for which 
\[
\braket{i_1,\ldots ,i_N}{p_1}\neq 0\neq \braket{j_1,\ldots ,j_N}{q_1}.
\]
\end{proof}

\subsection{Characterizing extremal rays}
\label{subsec:ExtremalRays}

The following lemma gives a characterization of the extreme rays of $\mathcal{S}\lb \CP_N\cap\coCP_N\rb$. 

\begin{lem}[Extremal rays]
For $\lb\ket{p},\ket{q}\rb\in \mathcal{S}\lb \CP_N\cap\coCP_N\rb\setminus\lset (0,0)\rset$ the following are equivalent: 
\begin{enumerate}
\item The pair $\lb\ket{p},\ket{q}\rb$ generates an extremal ray in $\mathcal{S}\lb \CP_N\cap\coCP_N\rb$.
\item If $\lb\ket{p},\ket{q}\rb\leq_Z\lb\ket{p'},\ket{q'}\rb$ for any $\lb\ket{p'},\ket{q'}\rb\in \mathcal{S}\lb \CP_N\cap\coCP_N\rb$, then $\lb\ket{p'},\ket{q'}\rb=\alpha\lb\ket{p},\ket{q}\rb$ for some $\alpha\geq 0$. 
\item If $\lb\ket{p},\ket{q}\rb<_Z\lb\ket{p'},\ket{q'}\rb$ for any $\lb\ket{p'},\ket{q'}\rb\in \mathcal{S}\lb \CP_N\cap\coCP_N\rb$, then $\lb\ket{p'},\ket{q'}\rb=(0,0)$. 
\item We have $\text{dim}\lb V^N_{p,q}\cap E_1(M_N)\rb=1$, with $E_1(M_N)$ as in Lemma \ref{lem:SubspaceCharact}.
\end{enumerate}
\label{lem:extremesZP}
\end{lem}

\begin{proof}
In the following let $\lb\ket{p},\ket{q}\rb\in \mathcal{S}\lb \CP_N\cap\coCP_N\rb$ be non-zero.

Assume that there exists a pair $\lb\ket{p'},\ket{q'}\rb\in \mathcal{S}\lb \CP_N\cap\coCP_N\rb$ such that $\lb\ket{p},\ket{q}\rb\leq_Z\lb\ket{p'},\ket{q'}\rb$ and $\lb\ket{p'},\ket{q'}\rb\neq \alpha\lb\ket{p},\ket{q}\rb$ for any $\alpha\geq 0$. Since $\mathcal{Z}\lb \ket{p},\ket{q}\rb\subseteq \mathcal{Z}\lb \ket{p'},\ket{q'}\rb$, there exists a $\beta>0$ such that $\lb\ket{p},\ket{q}\rb-\beta\lb\ket{p'},\ket{q'}\rb$ is ew-positive and by Lemma \ref{lem:SubspaceCharact} we have that $\lb\ket{p},\ket{q}\rb-\beta\lb\ket{p'},\ket{q'}\rb\in \mathcal{S}\lb \CP_N\cap\coCP_N\rb$. This implies
\[
\lb\ket{p},\ket{q}\rb = \beta\lb\ket{p'},\ket{q'}\rb + (\lb\ket{p},\ket{q}\rb-\beta\lb\ket{p'},\ket{q'}\rb),
\]  
showing that $\lb\ket{p},\ket{q}\rb$ does not generate an extremal ray.

It is obvious that $(2)$ implies $(3)$. To show that $(3)$ implies $(4)$ we assume that $\text{dim}\lb V^N_{p,q}\cap E_1(M_N)\rb>1$. This implies the existence of $\lb \ket{\tilde{p}},\ket{\tilde{q}}\rb\in V^N_{p,q}\cap E_1(M_N)$ such that $\lb \ket{\tilde{p}},\ket{\tilde{q}}\rb\neq \alpha \lb \ket{p},\ket{q}\rb$ for any $\alpha\in\R$, and without loss of generality we can assume that at least one entry of $\lb \ket{\tilde{p}},\ket{\tilde{q}}\rb$ is positive. By Definition \ref{defn:Subspacepq} of the subspace $V^N_{p,q}$ we have $\lb \ket{p},\ket{q}\rb\leq_Z \lb \ket{\tilde{p}},\ket{\tilde{q}}\rb$. Then, we have that 
\[
\alpha_{\max}:=\max\lset\alpha\geq 0~|~\lb\ket{p},\ket{q}\rb - \alpha \lb\ket{\tilde{p}},\ket{\tilde{q}}\rb\in \mathcal{S}\lb \CP_N\cap\coCP_N\rb\rset \in (0,\infty),
\]
and we can define 
\[
\lb\ket{p'},\ket{q'}\rb := \lb\ket{p},\ket{q}\rb - \alpha_{\max} \lb\ket{\tilde{p}},\ket{\tilde{q}}\rb\in \mathcal{S}\lb \CP_N\cap\coCP_N\rb .
\]
Finally, note that $\lb\ket{p'},\ket{q'}\rb\neq (0,0)$ and by maximality of $\alpha_{\max}$ we have $\lb\ket{p},\ket{q}\rb <_Z\lb\ket{p'},\ket{q'}\rb$. 

To show that $(4)$ implies $(1)$ assume that $\text{dim}\lb V^N_{p,q}\cap E_1(M_N)\rb=1$ for $\lb\ket{p},\ket{q}\rb\in \mathcal{S}\lb \CP_N\cap\coCP_N\rb$. Now, consider a non-zero pair $\lb\ket{p'},\ket{q'}\rb\in \mathcal{S}\lb \CP_N\cap\coCP_N\rb$ and assume that
\[
\lb\ket{p},\ket{q}\rb - \lb\ket{p'},\ket{q'}\rb\in \mathcal{S}\lb \CP_N\cap\coCP_N\rb .
\]
Since $\lb\ket{p'},\ket{q'}\rb$ is ew-positive we find that $\mathcal{Z}\lb \ket{p},\ket{q}\rb\subseteq \mathcal{Z}\lb \ket{p'},\ket{q'}\rb$, which implies that $\lb\ket{p'},\ket{q'}\rb\in  V^N_{p,q}\cap E_1(M_N)$. Using the assumption we conclude that $\lb\ket{p'},\ket{q'}\rb = \alpha \lb\ket{p},\ket{q}\rb$ for some $\alpha>0$. Therefore, $\lb\ket{p},\ket{q}\rb$ generates an extremal ray.  

\end{proof}

The previous lemma shows that the extremal rays of the cone $\mathcal{S}\lb \CP_N\cap\coCP_N\rb$ correspond to maximal pairs in the partial ordering $\leq_Z$. An important consequence is the following bound:

\begin{cor}[Rank bound]
If $\lb\ket{p},\ket{q}\rb\in \mathcal{S}\lb \CP_N\cap\coCP_N\rb\setminus\lset (0,0)\rset$ generates an extremal ray, then we have
\[
|\mathcal{Z}\lb \ket{p},\ket{q}\rb|\geq 4^N-1 ,
\]
where we set $|(A,B)| = |A|+|B|$ for sets $A$ and $B$.
\label{cor:RankBound}
\end{cor}
\begin{proof}
If $\lb\ket{p},\ket{q}\rb\in \mathcal{S}\lb \CP_N\cap\coCP_N\rb\setminus\lset (0,0)\rset$ generates an extremal ray, then we have $\text{dim}\lb V^N_{p,q}\cap E_1(M_N)\rb=1$ by Lemma \ref{lem:extremesZP}. Note that by Definition \ref{defn:Subspacepq} and Lemma \ref{lem:SubspaceCharact} respectively, we have $\text{dim}(E_1(M_N))=4^N$ and $\text{dim}(V_{p,q})=2\cdot 4^{N}-|\mathcal{Z}\lb \ket{p},\ket{q}\rb|$. The corollary then follows from the elementary inequality
\[
\text{dim}(V^N_{p,q}) + \text{dim}(E_1(M_N)) - \text{dim}(V^N_{p,q}\cap E_1(M_N)) \leq 2\cdot 4^N .
\] 
\end{proof}

The previous lemma and corollary characterize the extremal rays of $\mathcal{S}\lb \CP_N\cap\coCP_N\rb$. Now we will apply these results to show that tensor products of extremal rays stay extremal. 

\begin{lem}[Tensor products of extremal rays]
Consider $N_1,N_2\in\N$ and assume that $(\ket{p_1},\ket{q_1})\in \mathcal{S}\lb \CP_{N_1}\cap\coCP_{N_1}\rb$ and $(\ket{p_2},\ket{q_2})\in \mathcal{S}\lb \CP_{N_2}\cap\coCP_{N_2}\rb$ generate extremal rays. Then, the tensor product 
\[
(\ket{p_1}\otimes \ket{p_2},\ket{q_1}\otimes \ket{q_2})\in \mathcal{S}\lb \CP_{N_1+N_2}\cap\coCP_{N_1+N_2}\rb
\]
generates an extremal ray.
\label{lem:TensorProdExtremal}
\end{lem} 

\begin{proof}
Note that $(K^{\otimes N_1}\otimes K^{\otimes N_2})(\ket{p_1}\otimes \ket{p_2}) = \ket{q_1}\otimes \ket{q_2}$ and that tensor products of ew-positive vectors are ew-positive. Therefore, we have 
\[
(\ket{p_1}\otimes \ket{p_2},\ket{q_1}\otimes \ket{q_2})\in \mathcal{S}\lb \CP_{N_1+N_2}\cap\coCP_{N_1+N_2}\rb .
\]
Now, assume that $(\ket{p'},\ket{q'})\in \mathcal{S}\lb \CP_{N_1+N_2}\cap\coCP_{N_1+N_2}\rb$ satisfies 
\begin{equation}
(\ket{p_1}\otimes \ket{p_2},\ket{q_1}\otimes \ket{q_2})\leq_Z(\ket{p'},\ket{q'}).
\label{equ:tpex1}
\end{equation}
Consider the vectors
\[
\ket{v_1} = \begin{pmatrix} 1 \\ 1 \\ 0 \\ 0\end{pmatrix}, \quad \ket{v_2} = \begin{pmatrix} 1 \\ 0 \\ 1 \\ 0\end{pmatrix},\quad \ket{v_3} = \begin{pmatrix} 1 \\ 0 \\ 0 \\ 1\end{pmatrix},\quad \ket{v_4} = \begin{pmatrix} 0 \\ 0 \\ 1 \\ 1\end{pmatrix},
\]
and
\[
\ket{\tilde{v}_1} = \begin{pmatrix} 0 \\ 0 \\ 1 \\ 1\end{pmatrix}, \quad \ket{\tilde{v}_2} = \begin{pmatrix} 0 \\ 1 \\ 0 \\ 1\end{pmatrix},\quad \ket{\tilde{v}_3} = \begin{pmatrix} 0 \\ 1 \\ 1 \\ 0\end{pmatrix},\quad \ket{\tilde{v}_4} = \begin{pmatrix} 1 \\ 1 \\ 0 \\ 0\end{pmatrix}.
\]
For any $k_1,\ldots ,k_{N_2}\in\lset 1,2,3,4\rset$ we define
\[
\ket{p'_{k_1,\ldots ,k_{N_2}}} := \lb\one^{N_1}_4\otimes \bra{v_{k_1}}\otimes\cdots \otimes \bra{v_{k_{N_2}}}\rb \ket{p'},
\]
and 
\[
\ket{q'_{k_1,\ldots ,k_{N_2}}} := \lb\one^{N_1}_4\otimes \bra{\tilde{v}_{k_1}}\otimes\cdots \otimes \bra{\tilde{v}_{k_{N_2}}}\rb \ket{q'}.
\]
Using that $K\ket{v_i}=\ket{\tilde{v}_i}$ for any $i\in\lset 1,2,3,4\rset$ we find that 
\[
\lb\ket{p'_{k_1,\ldots ,k_{N_2}}},\ket{q'_{k_1,\ldots ,k_{N_2}}}\rb \in \mathcal{S}\lb \CP_{N_1}\cap\coCP_{N_1}\rb
\] 
and by \eqref{equ:tpex1} we have 
\[
\lb\ket{p_1},\ket{q_1}\rb \leq_Z (\ket{p'_{k_1,\ldots ,k_{N_2}}},\ket{q'_{k_1,\ldots ,k_{N_2}}})
\]
for any $k_1,\ldots ,k_{N_2}\in\lset 1,2,3,4\rset$. Using extremality of $\lb\ket{p_1},\ket{q_1}\rb$ and Lemma \ref{lem:extremesZP} we find $\alpha_{k_1,\ldots ,k_{N_2}}\geq 0$ for any $k_1,\ldots ,k_{N_2}\in\lset 1,2,3,4\rset$ such that 
\[
\ket{p'_{k_1,\ldots ,k_{N_2}}} = \alpha_{k_1,\ldots ,k_{N_2}}\ket{p_1} \quad\text{and}\quad \ket{q'_{k_1,\ldots ,k_{N_2}}} = \alpha_{k_1,\ldots ,k_{N_2}}\ket{q_1}.
\] 
Note that 
\[
\lset\ket{v_{k_1}}\otimes\cdots \otimes \ket{v_{k_{N_2}}} ~:~k_1,\ldots ,k_{N_2}\in\lset 1,2,3,4\rset\rset 
\]
is a basis of $(\R^4)^{\otimes N_2}$. Thus, we can define a linear functional $\alpha:(\R^4)^{\otimes N_2}\ra\R$ (by extending $\alpha_{k_1,\ldots ,k_{N_2}}$ linearly) such that 
\[
\lb\one^{N_1}_4\otimes \bra{v}\rb \ket{p'} = \alpha(v)\ket{p_1},
\]
for any $\ket{v}\in (\R^4)^{\otimes N_2}$. Writing $\alpha(v)=\braket{v}{a}$ for some $\ket{a}\in (\R^4)^{\otimes N_2}$ shows that 
\[
\ket{p'} = \ket{p_1}\otimes \ket{a}.
\]
Since the vectors $\ket{p'}$ and $\ket{p_1}$ are ew-positive, $\ket{a}$ is ew-positive as well, and since the vector
\[
\ket{q'} = (K^{\otimes N_1}\otimes K^{\otimes N_2}) \ket{p'} = \ket{q_1}\otimes K^{\otimes N_2}\ket{a},
\]
and the vector $\ket{q_1}$ are ew-positive, it follows that $K^{\otimes N_2}\ket{a}$ is ew-positive. We conclude that 
\[
\lb\ket{a},K^{\otimes N_2}\ket{a}\rb\in \mathcal{S}\lb \CP_{N_2}\cap\coCP_{N_2}\rb.
\]
By \eqref{equ:tpex1} we have that
\[
\lb\ket{p_2},\ket{q_2}\rb \leq_Z \lb\ket{a},K^{\otimes N_2}\ket{A}\rb,
\]
and by extremality (and Lemma \ref{lem:extremesZP}), we find $\beta\geq 0$ such that 
\[
\lb\ket{a},K^{\otimes N_2}\ket{a}\rb = \beta\lb\ket{p_2},\ket{q_2}\rb.
\]
Finally, we conclude that 
\[
\lb \ket{p'},\ket{q'}\rb = \lb\ket{p_1}\otimes \ket{a},\ket{q_1}\otimes K^{\otimes N_2}\ket{a}\rb = \beta\lb\ket{p_1}\otimes\ket{p_2},\ket{q_1}\otimes\ket{q_2}\rb.
\]
This finishes the proof by Lemma \ref{lem:extremesZP}.

\end{proof}

\subsection{Orbits of extremal rays under symmetry}
\label{subsec:OrbitsSymm}

To simplify the characterization of extremal rays we consider the symmetry group of the cone $\mathcal{S}\lb \CP_N\cap\coCP_N\rb$ (see Definition \ref{defn:SPPPoly}). Let $S_4\ni\sigma\mapsto U_{\sigma}\in\M_4$ denote the usual representation of the symmetric group on $\C^4$ (i.e.~by permutation matrices). As $\lbr K,U_{\sigma}\rbr=0$ for any $\sigma\in S_4$ we find that $\mathcal{S}\lb \CP_N\cap\coCP_N\rb$ is invariant under multiplication with $\bigotimes^N_{i=1} U_{\sigma_i}$ for any $\sigma_1,\ldots ,\sigma_N\in S_4$. Moreover, let $S_N\ni \tau \mapsto V_{\tau}\in\M^{\otimes N}_4$ denote the representation of the symmetric group $S_N$ acting by permuting the tensor factors, i.e.~such that $V_{\tau}\ket{i_1,\ldots ,i_N} = \ket{i_{\tau^{-1}(1)},\ldots ,i_{\tau^{-1}(N)}}$ for any $i_1,\ldots i_N\in\lset 1,2,3,4\rset$. Clearly, $\mathcal{S}\lb \CP_N\cap\coCP_N\rb$ is invariant under multiplication by $V_{\tau}$ for any $\tau\in S_N$. Finally, note that for every $(\ket{p},\ket{q})\in \mathcal{S}\lb \CP_N\cap\coCP_N\rb$ also $(\ket{q},\ket{p}) = (K^{\otimes N}\ket{p}, K^{\otimes N}\ket{q})\in \mathcal{S}\lb \CP_N\cap\coCP_N\rb$. This discussion shows: 

\begin{lem}[Symmetry]
For any pair $(\ket{p},\ket{q})\in \mathcal{S}\lb \CP_N\cap\coCP_N\rb$, permutations $\sigma_1,\ldots ,\sigma_N\in S_4$, any permutation $\tau\in S_N$, and $x\in\lset 0,1\rset$ we have that
\[
\lb V_{\tau}\lb\bigotimes^N_{i=1} U_{\sigma_i} \rb\lb K^{\otimes N}\rb^x\ket{p},V_{\tau}\lb\bigotimes^N_{i=1} U_{\sigma_i}\rb\lb K^{\otimes N}\rb^x \ket{q}\rb\in \mathcal{S}\lb \CP_N\cap\coCP_N\rb.
\]
Moreover, if $(\ket{p},\ket{q})\in \mathcal{S}\lb \CP_N\cap\coCP_N\rb$ generates an extremal ray, then the above pair generates an extremal ray as well.
\label{lem:Permutations}
\end{lem}    

The extremal rays of $\mathcal{S}\lb \CP_N\cap\coCP_N\rb$ form orbits under the symmetry group described in the previous lemma. To classify the extremal rays it is therefore sufficient to find a representant of each orbit. In the following, we will denote these orbits by
\begin{align*}
{\scriptstyle\text{orb}_N\lbr (\ket{p},\ket{q}) \rbr := \bigg\lset\lb V_{\tau}\lb\bigotimes^N_{i=1} U_{\sigma_i}K^x \rb\ket{p},V_{\tau}\lb\bigotimes^N_{i=1} U_{\sigma_i}K^x \rb\ket{q}\rb ~:~ \sigma_1,\ldots ,\sigma_N\in S_4 ,~\tau\in S_N,~x\in\lset 0,1\rset\bigg\rset}
\end{align*}
for any $(\ket{p},\ket{q})\in(\R^4)^{\otimes N}\oplus (\R^4)^{\otimes N}$. The following theorem summarizes the previous discussion:

\begin{thm}[Orbits of extremal rays]
For every $N\in\N$, there exists a finite set 
\[
\lset(\ket{p_i},\ket{q_i})\rset^M_{i=1}\subset \mathcal{S}\lb \CP_N\cap\coCP_N\rb
\] 
such that the following holds:
\begin{enumerate}
\item The zero patterns $\mathcal{Z}\lb \ket{p_i},\ket{q_i}\rb$ are maximal in the partial ordering $\leq_Z$ for $i\in\lset 1,\ldots ,M\rset$.
\item The scalings $\R^+\text{orb}_N\lb (\ket{p_i},\ket{q_i}) \rb$ for $i\in\lset 1,\ldots ,M\rset$ are pairwise disjoint. 
\item We have 
\[
\mathcal{S}\lb \CP_N\cap\coCP_N\rb = \text{cone}\lb \bigcup^M_{i=1} \text{orb}_N\lbr (\ket{p_i},\ket{q_i}) \rbr\rb.
\]

\end{enumerate}
\label{thm:OrbitStruc}
\end{thm} 

The previous theorem shows that to classify the extremal rays of $\mathcal{S}\lb \CP_N\cap\coCP_N\rb$ it is enough to classify the orbits under the symmetry group described above. In the next section we will apply this method to find all extremal rays of $\mathcal{S}\lb \CP_N\cap\coCP_N\rb$ for $N=1$ and $N=2$.

\section{The extremal rays of $\mathcal{S}\lb \CP_1\cap\coCP_1\rb$}
\label{sec:extremalRaysConcrete1}

We will now apply the formalism developed in the previous sections to characterize the extremal rays of $\mathcal{S}\lb \CP_1\cap\coCP_1\rb$. It should be emphasized that this result is not new: It is well-known that the parameters $(x,y,z)^T\in\R^3$ for which the Pauli diagonal map $\Pmap_{\mu}:\M_2\ra\M_2$ with $\ket{\mu}=(1,x,y,z)^T$ is both completely positive and completely copositive form a regular octahedron arising as the intersection of two tetrahedra corresponding to the parameters of $\CP_1$ and $\coCP_1$ respectively (see Figure \ref{fig:stellatedOct}). However, we believe that the characterization of the spectra $\mathcal{S}\lb \CP_1\cap\coCP_1\rb$ has not appeared in the literature before, and that the characterization (implied by this result) of positive Pauli diagonal maps by their spectra is new as well. We will begin with a definition.

\begin{figure*}[t!]
        \center
        \includegraphics[scale=0.7]{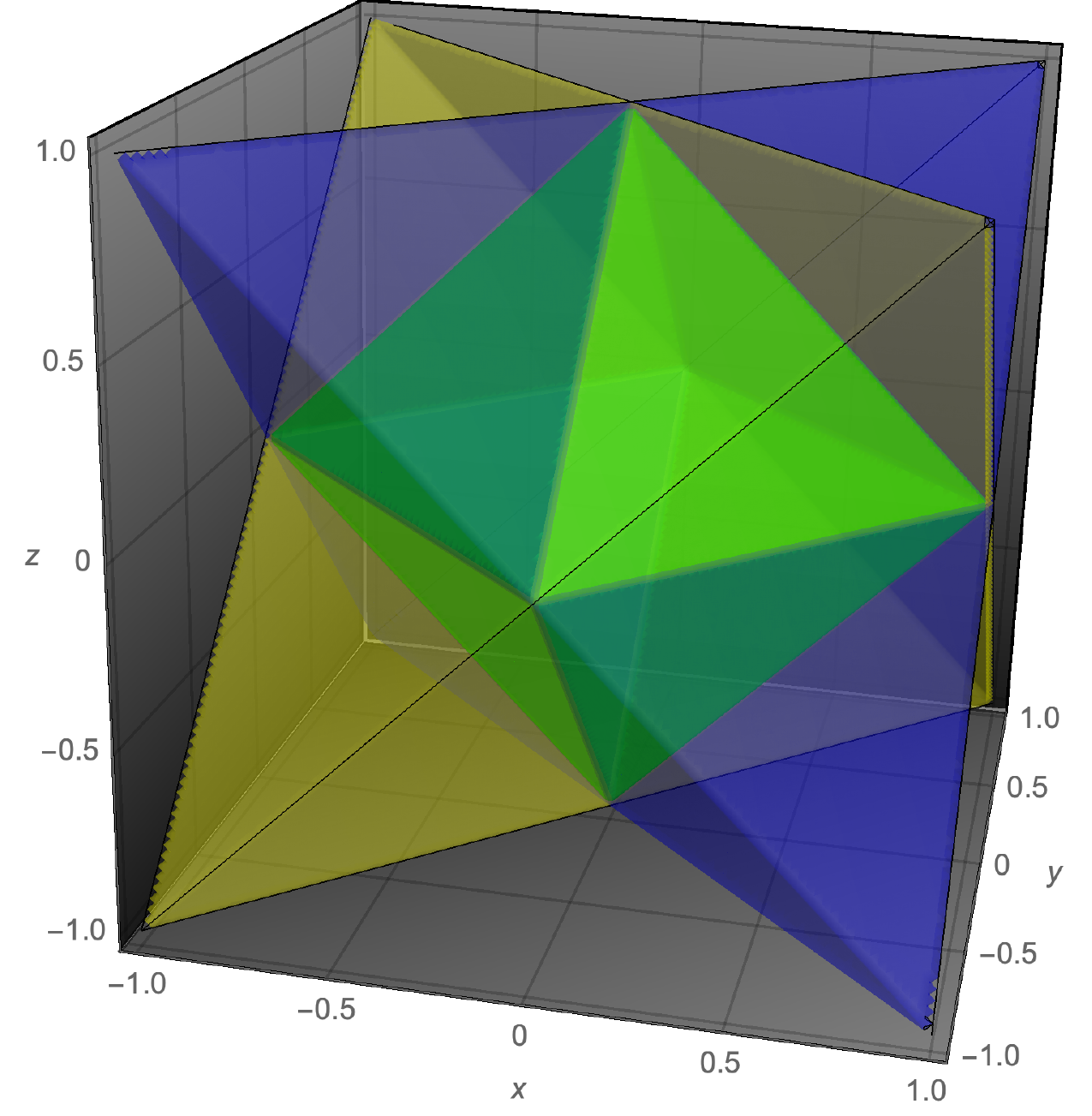}
        \caption{The parameters $(x,y,z)^T\in\R^3$ for which the Pauli diagonal map $\Pi_\mu:\M_2\ra\M_2$ with $\ket{\mu}=(1,x,y,z)^T$ is positive (and decomposable). The intersection of the two tetrahedra forms the octahedral base of the cone $\CP_1\cap \coCP_1$. }
        \label{fig:stellatedOct}
\end{figure*}

\begin{defn}
Given a $2$-element set of indices $\lset i,j\rset\subset \lset 1,2,3,4\rset$ with $i<j$ we define the \emph{complementary} set of indices as $\lset i,j\rset^c:=\lset k,l\rset$ such that $k<l$ and $i\neq k$ and $j\neq l$. Moreover, we introduce the vectors
\begin{equation}
\ket{\lset i,j\rset} = \ket{i}+\ket{j}\in \R^4.
\label{equ:boxVec}
\end{equation}
\end{defn}

Then, we have the following theorem. 

\begin{thm}[Extremal rays for $N=1$]
For every $2$-element set of indices $\lset i,j\rset\subset \lset 1,2,3,4\rset$ with $i<j$ the element $\lb \ket{\lset i,j\rset},\ket{\lset i,j\rset^c}\rb\in \mathcal{S}\lb \CP_1\cap\coCP_1\rb$ generates an extremal ray, and we have 
\[
\mathcal{S}\lb \CP_1\cap\coCP_1\rb = \text{cone}\Big( \text{orb}_1\lbr (\ket{\lset 1,2\rset},\ket{\lset 3,4\rset}) \rbr\Big) . 
\]
In particular, the polyhedral cone $\mathcal{S}\lb \CP_1\cap\coCP_1\rb$ has $6$ extremal rays.
\label{thm:ExtremalRaysN1}
\end{thm}

\begin{proof}
Note that 
\[
K \ket{\lset i,j\rset} = \ket{\lset i,j\rset^{c}}
\]
for any $2$-element set of indices $\lset i,j\rset\subset \lset 1,2,3,4\rset$ with $i<j$. This shows that $\lb \ket{\lset i,j\rset},\ket{\lset i,j\rset^c}\rb\in \mathcal{S}\lb \CP_1\cap\coCP_1\rb$. Next, consider $\lb \ket{p},\ket{q}\rb\in \mathcal{S}\lb \CP_1\cap\coCP_1\rb$ such that $\lb \ket{p},\ket{q}\rb >_Z \lb \ket{\lset i,j\rset},\ket{\lset i,j\rset^c}\rb$ for some $i<j$. This implies that  
\[
\mathcal{Z}\lb \ket{p}\rb\supseteq \mathcal{Z}\lb \ket{\lset i,j\rset}\rb = \lset i,j\rset^c \quad \text{and}\quad\mathcal{Z}\lb \ket{q}\rb\supseteq \mathcal{Z}\lb \ket{\lset i,j\rset^c}\rb = \lset i,j\rset,
\]
but such that at least one of the previous inclusions is strict. Without loss of generality, we can assume that the first inclusion is strict, which implies that $\ket{p}\in\R^4$ has at most one non-zero element. But since $K\ket{p}=\ket{q}$ is ew-positive, we conclude that $\ket{p}=\ket{q}=0$, and by Lemma \ref{lem:extremesZP} the element $\lb \ket{\lset i,j\rset},\ket{\lset i,j\rset^c}\rb\in \mathcal{S}\lb \CP_1\cap\coCP_1\rb$ generates an extremal ray.

It is easy to check that 
\[
\lset \lb \ket{\lset i,j\rset},\ket{\lset i,j\rset^c}\rb~:~i,j\in \lset 1,2,3,4\rset, i<j\rset = \text{orb}_1\lbr (\ket{\lset 1,2\rset},\ket{\lset 3,4\rset})\rbr.
\]
Finally, assume that there exists an element $\lb \ket{p},\ket{q}\rb\in \mathcal{S}\lb \CP_1\cap\coCP_1\rb$ generating an extremal ray and such that
\[
\lb \ket{p},\ket{q}\rb\notin \text{cone}\lb\text{orb}_1\lbr (\ket{\lset 1,2\rset},\ket{\lset 3,4\rset})\rbr\rb .
\]
By Corollary \ref{cor:RankBound} we have $|\mathcal{Z}\lb \ket{p},\ket{q}\rb|\geq 3$, and without loss of generality we can assume that $|\mathcal{Z}\lb \ket{p}\rb|\geq 2$. By Lemma \ref{lem:Permutations} we may apply a suitable permutation, and we can assume that $\ket{p}=\lb a,b,0,0\rb^T$ for some $a,b\geq 0$. Since $\braket{\lset 3,4\rset}{p}=0$, we can apply Lemma \ref{lem:Orthogonality} to conclude that $\ket{q}=\lb 0,0,c,d\rb^T$ for some $c,d\geq 0$. Because $\lb \ket{p},\ket{q}\rb\neq \alpha \lb\ket{\lset 3,4\rset},\ket{\lset 1,2\rset}\rb$ for any $\alpha\geq 0$ and because it generates an extremal ray in $\mathcal{S}\lb \CP_1\cap\coCP_1\rb$ we conclude that $0\in\lset a,b,c,d\rset$. But since both $\ket{p}$ and $\ket{q}=K\ket{p}$ are ew-positive, this implies that $\ket{p}=\ket{q}=0$. This finishes the proof. 
\end{proof}

The previous theorem classifies the extremal rays of $\mathcal{S}\lb \CP_1\cap\coCP_1\rb$. As expected there are $6$ extremal rays corresponding to the vertices of the octahedron in Figure \ref{fig:stellatedOct}. As an application we can give a characterization of positive Pauli diagonal maps in terms of their spectrum. Our result can be compared to a characterization of the possible spectra of the Choi matrices of general positive qubit maps obtained in~\cite{johnston2018inverse}.

\begin{thm}[Positivity of qubit Pauli diagonal maps]
Let $\Pi_\mu:\M_2\ra\M_2$ be a Pauli diagonal map with parameters $\ket{\mu}\in \R^4$. The following are equivalent: 
\begin{enumerate}
\item $\Pi_\mu$ is decomposable. 
\item $\Pi_\mu$ is positive.
\item The spectrum $(s_1,s_2,s_3,s_4)\in \R^4$ of the Choi matrix $C_{\Pi_\mu}$ ordered such that $s_1\geq s_2\geq s_3\geq s_4$ satisfies: 
\begin{enumerate}
	\item $s_1\geq s_2\geq s_3\geq 0$.
	\item $s_4+s_3\geq 0$.
\end{enumerate}
\end{enumerate}
\label{thm:DecompQubitPauli}
\end{thm}

\begin{proof}
It is clear that $1.$ implies $2.$. Consider a Pauli diagonal map $\Pmap_\mu:\M_2\ra\M_2$ with parameter vector $\ket{\mu}\in\R^4$ and such that 
\[
\ket{s}:=\text{spec}\lb C_{\Pmap_\mu}\rb = \dm\ket{\mu}.
\]
For any $\lset i,j\rset\subset \lset 1,2,3,4\rset$ with $i< j$ we have
\begin{equation}
\text{Tr}\lb C_{\Pi_\mu}C_{\Pi_{\dm^T\ket{\lset i,j\rset}}}\rb=\braket{s}{\lset i,j\rset} = s_i+s_j .
\label{equ:sisjPos}
\end{equation}
The Pauli diagonal map $\Pi_{\dm^T\ket{\lset i,j\rset}}:\M_2\ra\M_2$ corresponding to the extremal ray $\ket{\lset i,j\rset}$ of $\mathcal{S}\lb \CP_1\cap\coCP_1\rb$ satisfies 
\[
\text{rk}\lb C_{\Pi_{\dm^T\ket{\lset i,j\rset}}}\rb=2,
\]
since $\text{spec}\lb C_{\Pi_{\dm^T\ket{\lset i,j\rset}}}\rb = \ket{\lset i,j\rset}$. By \cite[Theorem 1]{horodecki2000operational} it follows that $C_{\Pi_{\dm^T\ket{\lset i,j\rset}}}\in \M_2\otimes \M_2$ is separable. Hence, the expression in \eqref{equ:sisjPos} is positive whenever $\Pi_\mu:\M_2\ra\M_2$ is a positive map. Simplifying these inequalities for any $\lset i,j\rset\subset \lset 1,2,3,4\rset$ with $i<j$ shows that $2.$ implies $3.$. 

The conditions in $3.$ imply that the expressions in \eqref{equ:sisjPos} are positive for $\lset i,j\rset\subset \lset 1,2,3,4\rset$ with $i<j$. By Theorem \ref{thm:ExtremalRaysN1} this shows that $\braket{s}{p}\geq 0$ for any $\ket{p}\in\R^4$ such that $\lb\ket{p},K\ket{p}\rb\in \mathcal{S}\lb \CP_2\cap\coCP_2\rb$. By Lemma \ref{lem:SpectrDecCondGen} we find that $\Pmap_\mu:\M_2\ra\M_2$ is decomposable. This shows that $3.$ implies $2.$ and finishes the proof. 

\end{proof}

Using the Sinkhorn-type scaling argument from~\cite{aubrun2015two}, the previous theorem also implies St\o rmer's theorem~\cite{stormer1963positive}, that every positive qubit map is decomposable. However, it is of course much easier to directly decompose a positive qubit Pauli diagonal map as a sum of a completely positive and a completely copositive map (see~\cite{aubrun2015two} or look at Figure \ref{fig:stellatedOct}).

\section{The structure of $\mathcal{S}\lb \CP_2\cap\coCP_2\rb$ and characterizing $\Dec_2$} 
\label{sec:extremalRaysConcrete2}

To study $\mathcal{S}\lb \CP_2\cap\coCP_2\rb$ it is convenient to identify $(\R^4)^{\otimes 2}\simeq M_4(\R)$, i.e.~the space of $4\times 4$ matrices with real entries. With this identification we have 
\[
\mathcal{S}\lb \CP_2\cap\coCP_2\rb \simeq \Big\lset \lb P,Q\rb\in M_4(\R)\oplus M_4(\R)~:~P,Q\text{ ew-positive and } KPK=Q\Big\rset .
\]

\subsection{Classification of extremal rays}

In the next theorem, we identify the orbits of extremal rays of the cone $\mathcal{S}\lb \CP_2\cap\coCP_2\rb$ under the symmetry group introduced in Theorem \ref{thm:OrbitStruc}. One such orbit arises from the tensor products (see Lemma \ref{lem:TensorProdExtremal}) of the extremal rays of $\mathcal{S}\lb \CP_1\cap\coCP_1\rb$ identified in Theorem \ref{thm:ExtremalRaysN1}. Surprisingly, there are only two other orbits of extremal rays. To abbreviate any further discussion about these extremal rays, we will call them \emph{boxes}, \emph{diagonals}, and \emph{crosses} motivated by the shape of their zero patterns:

\begin{thm}[Extremal rays for $N=2$] 
The polyhedral cone $\mathcal{S}\lb \CP_2\cap\coCP_2\rb$ has $252$ extremal rays divided into three orbits such that
\[
\mathcal{S}\lb \CP_2\cap\coCP_2\rb = \text{cone}\lb \bigcup^3_{i=1} \text{orb}_2\lbr (P_i,Q_i) \rbr\rb.
\]
The orbits are generated by the following pairs:

\begin{enumerate}
\item 
\[
(P_1,Q_1) = \lb\begin{pmatrix} 1 & 1 & 0 & 0 \\ 1 & 1 & 0 & 0 \\ 0 & 0 & 0 & 0 \\ 0 & 0 & 0 & 0\end{pmatrix},\begin{pmatrix} 0 & 0 & 0 & 0 \\ 0 & 0 & 0 & 0 \\ 0 & 0 & 1 & 1 \\ 0 & 0 & 1 & 1\end{pmatrix}\rb.
\]
The elements of $\alpha\text{orb}_2\lbr \lb P_1,Q_1\rb\rbr$ for any $\alpha >0$ are called \emph{boxes} and they generate $36$ extremal rays.
\item 
\[
(P_2,Q_2) = \lb\begin{pmatrix} 1 & 0 & 0 & 0 \\ 0 & 1 & 0 & 0 \\ 0 & 0 & 1 & 0 \\ 0 & 0 & 0 & 1\end{pmatrix},\begin{pmatrix} 1 & 0 & 0 & 0 \\ 0 & 1 & 0 & 0 \\ 0 & 0 & 1 & 0 \\ 0 & 0 & 0 & 1\end{pmatrix}\rb.
\]
The elements of $\alpha\text{orb}_2\lbr \lb P_2,Q_2\rb\rbr$ for any $\alpha >0$ are called \emph{diagonals} and they generate $24$ extremal rays.
\item 
\[
(P_3,Q_3) = \lb\begin{pmatrix} 1 & 0 & 0 & 0 \\ 0 & 1 & 0 & 0 \\ 0 & 0 & 1 & 0 \\ 1 & 1 & 1 & 0\end{pmatrix},\begin{pmatrix} 1 & 0 & 0 & 1 \\ 0 & 1 & 0 & 1 \\ 0 & 0 & 1 & 1 \\ 0 & 0 & 0 & 0\end{pmatrix}\rb.
\]
The elements of $\alpha\text{orb}_2\lbr \lb P_3,Q_3\rb\rbr$ for any $\alpha >0$ are called \emph{crosses} and they generate $192$ extremal rays.
\end{enumerate}

\label{thm:ExtremeRaysN2}
\end{thm}

It should be noted that Theorem \ref{thm:ExtremeRaysN2} can be easily verified using standard software for analyzing convex polytopes (e.g.~the Multi-Parametric Toolbox~\cite{MPT3} in Matlab, or polymake~\cite{polymake1,polymake2}). We will also present a human-readable proof in Appendix \ref{app:Therest}, which is unfortunately quite tedious. It would be nice to have a shorter proof for this result.

\subsection{Properties of ququart Pauli PPT maps}

We need to make a few comments about the extremal rays of the cone $\mathcal{S}\lb \CP_2\cap\coCP_2\rb$ identified in Theorem \ref{thm:ExtremeRaysN2}. Recall that the set $\mathcal{S}\lb \CP_2\cap\coCP_2\rb$ consists of pairs of spectra of Choi matrices $C_{\Pmap^{(2)}_{\mu}}$ and $C_{\vartheta^{\otimes 2}_2\circ\Pmap^{(2)}_{\mu}}$ for PPT Pauli diagonal maps $\Pmap^{(2)}_{\mu}:\M_4\ra\M_4$ (cf.~Definition \ref{defn:SPPPoly}). By the particular form of the spectra identified in Theorem \ref{thm:ExtremeRaysN2} we have the following: 

\begin{cor}[Properties of extremal PPT Pauli diagonal maps]
Let $\Pmap^{(2)}_{\mu}:\M_4\ra\M_4$ be an extremal PPT Pauli diagonal map. Then we have the following: 
\begin{enumerate}
\item Both Choi matrices 
\[
C_{\Pmap^{(2)}_{\mu}},C_{\vartheta^{\otimes 2}_2\circ\Pmap^{(2)}_{\mu}}\in (\M_4\otimes \M_4)^+
\] 
are multiples of Hermitian projectors. 
\item The birank 
\[
\lb \text{rk}\lb C_{\Pmap^{(2)}_{\mu}}  \rb, \text{rk}\lb C^{\Gamma}_{\Pmap^{(2)}_{\mu}}\rb\rb
\]
of the Choi matrix $C_{\Pmap^{(2)}_{\mu}}\in (\M_4\otimes \M_4)^+$ is either $(4,4)$ or $(6,6)$. 
\item If $\Pmap^{(2)}_{\mu}$ is not entanglement breaking, then its spectral matrix $\dm\mu \dm^T$ is a cross. 
\end{enumerate}
\label{cor:Properties}
\end{cor}
\begin{proof}
The first statement follows immediately since Choi matrices of Pauli diagonal maps are always Hermitian and the extremal spectra in Theorem \ref{thm:ExtremeRaysN2} only contain the values $0$ and $1$. The second statement follows from counting the entries that are equal to $1$. For the third statement recall that a Choi matrix $C_{\Pmap^{(2)}_{\mu}}\in (\M_4\otimes \M_4)^+$ with positive partial transpose and $\text{rk}\lb C_{\Pmap^{(2)}_{\mu}}\rb =4$ is separable (see \cite[Theorem 1]{horodecki2000operational}) and therefore the corresponding Pauli diagonal map $\Pmap^{(2)}_{\mu}$ would be entanglement breaking. 

It remains to show that the Pauli diagonal maps corresponding to the crosses from Theorem \ref{thm:ExtremeRaysN2} are not entanglement breaking. For this consider the cross $(P_3,Q_3)\in \mathcal{S}\lb \CP_2\cap\coCP_2\rb$ introduced in Theorem \ref{thm:ExtremeRaysN2}, and let $\mu\in\M_4(\R)$ denote the parameter matrix $\mu = 2(\dm^T P_3 \dm)/3$ (cf.~Theorem \ref{thm:CPCond}) normalized such that $\Pmap^{(2)}_{\mu}$ is unital and trace-preserving. It is easy to compute 
\[
\mu = \frac{1}{3}\begin{pmatrix} 3 & 1 & 1 & -1 \\ -1 & 1 & -1 & 1 \\-1 & -1 & 1 & 1 \\1 & 1 & 1 & 1 \end{pmatrix}.
\] 
Finally, the Choi matrix $C_{\Pmap^{(2)}_{\mu}}$ is entangled due to the realignment criterion from Theorem \ref{thm:RealignementPauliMultipliers} as 
\[
\sum_{ij}|\mu_{ij}| = 6 >4.
\]
\end{proof}

We conclude this section with a brief side remark regarding the so-called PPT squared conjecture~\cite{christandl2012PPT}, whether for linear maps $T_1, T_2$ that are both completely positive and completely copositive the composition $T_1\circ T_2$ is entanglement breaking (see~\cite{muller2018PPT} for details). Recently, this conjecture has received much attention~\cite{kennedy2018composition,rahaman2018eventually,collins2018ppt,chen2019positive,hanson2020eventually} and Pauli diagonal maps might be a natural candidate for finding a counterexample. However, we can show that no such counterexample can be found among ququart Pauli diagonal maps. 

Consider two Pauli diagonal maps $\Pmap_{\mu_1},\Pmap_{\mu_2}:\M_4\ra\M_4$ with parameter matrices $\mu_1,\mu_2\in\M_4(\R)$. The composition $\Pmap_{\mu_1}\circ\Pmap_{\mu_2}=\Pmap_{\mu_1\circ \mu_2}$ (where $\circ$ denotes the Schur product) is again Pauli diagonal, and the spectrum of its Choi matrix (cf.~Theorem \ref{thm:CPCond}) is given by $S = \dm (\mu_1\circ \mu_2)\dm^T$. It can be verified that for all $\mu_1,\mu_2\in\M_4(\R)$ corresponding to crosses from Theorem \ref{thm:ExtremeRaysN2}, the spectral matrix $S$ is a convex combination of boxes and diagonals. Since boxes and diagonals correspond to entanglement breaking Pauli diagonal maps by Corollary \ref{cor:Properties}, we can use Theorem \ref{thm:ExtremeRaysN2} to conclude the following:

\begin{thm}[PPT squared conjecture for ququart Pauli diagonal maps]
For any pair $\Pmap_{\mu_1},\Pmap_{\mu_2}:\M_4\ra\M_4$ of Pauli diagonal maps that are both completely positive and completely copositive the composition $\Pmap_{\mu_1}\circ\Pmap_{\mu_2}$ is entanglement breaking.
\end{thm}

\subsection{Spectral criteria for decomposability}

We will now use the characterization of extremal rays of $\mathcal{S}\lb \CP_2\cap\coCP_2\rb$ to prove decomposability criteria for Pauli diagonal maps $\Pmap^{(2)}_{\mu}:\M_4\ra\M_4$. 

\begin{thm}[Spectral conditions for decomposability, $N=2$]
Let $\Pmap^{(2)}_{\mu}:\M_4\ra\M_4$ denote a Pauli diagonal map with parameter matrix $\mu\in\M_4(\R)$ and spectral matrix $S=\dm\mu \dm^T\in \M_4(\R)$ (cf.~Theorem \ref{thm:CPCond}). The map $\Pmap^{(2)}_{\mu}$ is decomposable if and only if 
\begin{align*}
S_{\sigma_1(1)\sigma_2(1)}+S_{\sigma_1(1)\sigma_2(2)}+S_{\sigma_1(2)\sigma_2(1)}+S_{\sigma_1(2)\sigma_2(2)} &\geq 0 ,\\
S_{\sigma_1(1)\sigma_2(1)}+S_{\sigma_1(2)\sigma_2(2)}+S_{\sigma_1(3)\sigma_2(3)}+S_{\sigma_1(4)\sigma_2(4)} &\geq 0 ,\\
S_{\sigma_1(1)\sigma_2(1)}+S_{\sigma_1(2)\sigma_2(2)}+S_{\sigma_1(3)\sigma_2(3)} + S_{\sigma_1(4)\sigma_2(1)} + S_{\sigma_1(4)\sigma_2(2)} + S_{\sigma_1(4)\sigma_2(3)} &\geq 0 ,\\
S_{\sigma_1(1)\sigma_2(1)}+S_{\sigma_1(2)\sigma_2(2)}+S_{\sigma_1(3)\sigma_2(3)} + S_{\sigma_1(1)\sigma_2(4)} + S_{\sigma_1(2)\sigma_2(4)} + S_{\sigma_1(3)\sigma_2(4)} &\geq 0.
\end{align*}
for all permutations $\sigma_1,\sigma_2\in S_4$. Moreover, if the linear map $\Pmap^{(2)}_{\mu}$ is positive, then the first two inequalities are always satisfied. 
\label{thm:SpectralDecompCritN2}
\end{thm} 

\begin{proof}
The theorem follows immediately by combining Lemma \ref{lem:SpectrDecCondGen} and Theorem \ref{thm:ExtremeRaysN2}. Note that for the crosses $(P_3,Q_3)\in \mathcal{S}\lb \CP_2\cap\coCP_2\rb$ as introduced in Theorem \ref{thm:ExtremeRaysN2} we need to check both 
\[
\text{Tr}\lb S^TU^T_{\sigma_1} P_3 U_{\sigma_2}\rb \geq 0 \quad\text{ and }\quad \text{Tr}\lb S^T U^T_{\sigma_1} Q_3 U_{\sigma_2}\rb \geq 0
\]
for all permutations $\sigma_1,\sigma_2\in S_4$ since $P_3$ and $Q_3$ do not lie on the same $S_4\times S_4$ orbit. This leads to the last two inequalities. Since the extremal rays corresponding to boxes and diagonals (see Theorem \ref{thm:ExtremeRaysN2}) lead to entanglement breaking Pauli diagonal maps by Corollary \ref{cor:Properties} the first two inequalities in the statement of the theorem are always satisfied when the Pauli diagonal map $\Pmap^{(2)}_{\mu}$ is positive.  
\end{proof}

For convenience, we state a corollary where the Pauli diagonal map is a power $\Pi_\mu^{\otimes 2}$ for $\Pi_\mu:\M_2\ra\M_2$ and some parameter vector $\ket{\mu}\in\R^4$. The proof follows immediately from the previous theorem by realizing that the spectral matrix of the Pauli diagonal map $\Pi_\mu^{\otimes 2}$ is the symmetric matrix $S=\dm\proj{\mu}{\mu} \dm^T\in\M_4(\R)$ so that the two final inequalities in Theorem \ref{thm:SpectralDecompCritN2} coincide.   

\begin{cor}
Let $\Pi_\mu:\M_2\ra\M_2$ denote a Pauli diagonal map with parameter vector $\ket{\mu}\in\R^4$ and spectral vector $\ket{s}=\dm\ket{\mu}$ such that $\Pi_\mu^{\otimes 2}$ is positive. Then, $\Pi_\mu^{\otimes 2}$ is decomposable if and only if 
\begin{align*}
\lbr s_{\sigma_1(1)} + s_{\sigma_1(4)}\rbr s_{\sigma_2(1)} + \lbr s_{\sigma_1(2)} + s_{\sigma_1(4)}\rbr s_{\sigma_2(2)} + \lbr s_{\sigma_1(3)} + s_{\sigma_1(4)}\rbr s_{\sigma_2(3)} \geq 0,
\end{align*}
for all permutations $\sigma_1,\sigma_2\in S_4$.
\label{cor:DecompTensorProdPauli}
\end{cor}

\section{Decomposability of tensor squares of qubit maps}
\label{sec:DecTensSquares}

We will now present the proof of our main result stated in Theorem \ref{thm:Main1}: The tensor square $P^{\otimes 2}:\M_4\ra\M_4$ of a linear map $P:\M_2\ra\M_2$ is positive if and only if it is decomposable. Our proof has two parts: First, we reduce the problem to Pauli diagonal maps using the Sinkhorn-type scaling technique from~\cite{aubrun2015two}. Then, we apply Theorem \ref{thm:ExtremeRaysN2} and certain symmetries of qubit Pauli diagonal maps to show that positive tensor squares of qubit Pauli diagonal maps are decomposable. 

Although not needed for the proof of Theorem \ref{thm:Main1} we will formulate a general theorem to reduce similar questions about membership of tensor products of positive qubit maps in mapping cones~\cite{stormer1986extension} to the membership of Pauli diagonal maps. This generalizes the aforementioned Sinkhorn-type scaling technique from~\cite{aubrun2015two} and we hope that these results can be applied in different context in the future.

\subsection{Reduction to Pauli diagonal maps}
\label{subsec:RedPauliMult}

Let $\mathcal{P}\lb n,m\rb$ denote the cone of positive maps $P:\M_n\ra\M_{m}$. The notion of mapping cones was introduced by E.~St{\o}rmer in \cite{stormer1986extension} (see also \cite{skowronek2009cones} for more details). The following is a slight modification of the original definition: 

\begin{defn}[Mapping cones]
We call a system $\mathcal{C} = \lset \mathcal{C}_{n,m}\rset_{n,m\in \N}$ of subcones $\mathcal{C}_{n,m}\subset \mathcal{P}\lb n,m\rb$ a \emph{mapping cone} if the following conditions are satisfied:
\begin{enumerate}
\item For any $n,m\in\N$ the subcone $\mathcal{C}_{n,m}$ is closed. 
\item For any $n,m,n',m'\in \N$, $P\in \mathcal{C}_{n,m}$ and completely positive maps $T:\M_{n'}\ra\M_n$ and $S:\M_{m}\ra\M_{m'}$ we have that $S\circ P\circ T\in \mathcal{C}_{n',m'}$.
\end{enumerate}
For $P:\M_{n}\ra\M_{m}$ we will simply write $P\in \mathcal{C}$ instead of $P\in \mathcal{C}_{n,m}$, and given two mapping cones $\mathcal{C}_1$ and $\mathcal{C}_2$ we will write $P\in \mathcal{C}_1\setminus \mathcal{C}_2$ instead of $P\in (\mathcal{C}_1)_{n,m}$ and $P\notin (\mathcal{C}_2)_{n,m}.$
\label{defn:MappingCone}
\end{defn}
In the following we will focus mostly on the cones of positive maps and of decomposable maps, and we refer to \cite{skowronek2009cones} for more examples of mapping cones. The following proof follows mostly the lines of a proof by Aubrun and Szarek for St{\o}rmer's theorem (see~\cite{aubrun2015two} and \cite{aubrun2017alice}).  

\begin{thm}[Reduction to Pauli-multipliers]
Let $\mathcal{C}_1$ and $\mathcal{C}_2$ be mapping cones, and $Q\in \mathcal{C}_1(d_1,d_2)$. There exists a positive map $P:\M_2\ra\M_2$ such that $(P\otimes Q)\in \mathcal{C}_1\setminus \mathcal{C}_2$ if and only if there exists a $(x,y,z)^T\in\R^3$ such that $(\Pmap_\mu\otimes Q)\in \mathcal{C}_1\setminus \mathcal{C}_2$ with $\ket{\mu}=(1,x,y,z)^T$.
\label{thm:ReductionPauli}
\end{thm}
\begin{proof}
One direction is obvious. For the other direction consider a positive map $P:\M_2\ra\M_2$ such that $(P\otimes Q)\in \mathcal{C}_1\setminus \mathcal{C}_2$. By assumption $(\mathcal{C}_2)_{2d_1,2d_2}$ is closed. Therefore, there exists an $\epsilon>0$ such that $P_{\epsilon}:\M_2\ra\M_2$ defined by
\[
P_{\epsilon}(X) = P(X) + \epsilon \text{Tr}\lbr X\rbr \one_2 
\]
for $X\in\M_2$ satisfies $(P_\epsilon\otimes Q)\notin \mathcal{C}_2$. Setting $T_1:\M_{2}\otimes \M_{d_1}\ra\M_{d_1}$ to $T_1=\text{Tr}\otimes \ident_{d_1}$ and $T_2:\M_{d_2}\ra\M_2\otimes \M_{d_2}$ to $T_2(X)=\one_2\otimes X$ we have 
\[
P_{\epsilon}\otimes Q = P\otimes Q + \epsilon T_2\circ Q\circ T_1 \in \mathcal{C}_1
\]
since $T_1$ and $T_2$ are completely positive and $\mathcal{C}_1$ is a mapping cone. Since $P_\epsilon$ is in the interior of $\mathcal{P}\lb 2,2\rb$ we can use Sinkhorn's normal form (see e.g.~\cite[Proposition 2.32]{aubrun2017alice}) to find positive definite operators $A,B\in\M_2$ such that 
\[
\tilde{P} = \text{Ad}_A\circ P_\epsilon\circ \text{Ad}_B
\]
is positive, unital and trace-preserving. By~\cite{ruskai2002analysis} there exist unitaries $U,V\in\mathcal{U}_2$ and $(x,y,z)^T\in\R^3$ such that 
\[
\Pmap_{\mu} = \text{Ad}_U\circ \tilde{P}\circ \text{Ad}_V = \text{Ad}_{UA}\circ P_{\epsilon}\circ \text{Ad}_{BV},
\] 
where $\ket{\mu}=(1,x,y,z)^T$. Using that $\mathcal{C}_1$ is a mapping cones, we conclude that 
\[
\Pmap_\mu\otimes Q = \lb\text{Ad}_{UA}\otimes \ident_{d_2}\rb\circ \lb P_\epsilon\otimes Q\rb\circ \lb\text{Ad}_{BV}\otimes \ident_{d_1}\rb \in\mathcal{C}_1. 
\]
Since the matrices $UA\in\M_2$ and $BV\in\M_2$ are invertible, we find that 
\[
P_\epsilon\otimes Q = \lb\text{Ad}_{(UA)^{-1}}\otimes \ident_{d_2}\rb\circ \lb\Pmap_\mu\otimes Q\rb\circ \lb\text{Ad}_{(BV)^{-1}}\otimes \ident_{d_1}\rb.
\]
Therefore, $(P_\epsilon\otimes Q)\notin \mathcal{C}_2$ implies $(\Pmap_\mu\otimes Q)\notin \mathcal{C}_2$ as $\mathcal{C}_2$ is a mapping cone.
\end{proof}

The previous theorem implies the following result on tensor powers of qubit maps. 

\begin{cor}
Let $\mathcal{C}_1$ and $\mathcal{C}_2$ be mapping cones. For $N\in\N$ there exists a positive map $P:\M_2\ra\M_2$ such that $P^{\otimes N}\in \mathcal{C}_1\setminus \mathcal{C}_2$ if and only if there exists $(x,y,z)^T\in\R^3$ such that $\Pmap^{\otimes N}_{\mu}\in \mathcal{C}_1\setminus \mathcal{C}_2$ for $\ket{\mu}=(1,x,y,z)^T$.
\end{cor}

When $\mathcal{C}_1$ is the cone of positive maps and $\mathcal{C}_2$ the cone of decomposable maps we find: 

\begin{cor}
For $N\in\N$ there exists a positive map $P:\M_2\ra\M_2$ such that $P^{\otimes N}$ is positive but not decomposable if and only if there exists $(x,y,z)^T\in\R^3$ such that $\Pmap^{\otimes N}_{\mu}$ for $\ket{\mu}=(1,x,y,z)^T$ is positive but not decomposable.
\label{cor:Main}
\end{cor}

\subsection{Symmetries of qubit Pauli diagonal maps}
\label{subsec:SymmPauliMult}

 The next lemma collects some well-known transformations of Pauli diagonal maps.

 \begin{lem}[Symmetries]
 For any $(x,y,z)^T\in\R^3$ and $\ket{\mu}=(1,x,y,z)^T$ the following hold true: 
 \begin{enumerate}
 \item We have $\text{Ad}_H\circ \Pmap_\mu \circ \text{Ad}_H = \Pmap_{\mu'}$ for $\ket{\mu'}=F\ket{\mu}$ where \[ H=\frac{1}{\sqrt{2}}\begin{pmatrix}1 & 1\\ 1 & -1 \end{pmatrix}\quad\text{ and }\quad F=\begin{pmatrix} 1 & 0 & 0 & 0 \\0 & 0 & 1 & 0\\ 0 & 1 & 0 & 0 \\ 0 & 0 & 0 & -1\end{pmatrix}. \]
 \item We have $\text{Ad}_{W}\circ \Pmap_\mu \circ \text{Ad}_{W^\dagger} = \Pmap_{\mu'}$ for $\ket{\mu'}=S\ket{\mu}$ where \[W=\frac{1}{\sqrt{2}}\begin{pmatrix}1 & 1\\ i & -i \end{pmatrix}\quad\text{ and }\quad S=\begin{pmatrix} 1 & 0 & 0 & 0\\0 & 0 & 0 & 1\\ 0 & 1 & 0 & 0 \\ 0 & 0 & 1 & 0\end{pmatrix}. \]
 \item For $j\in\lset 2,3,4\rset$ we have $\text{Ad}_{\sigma_j}\circ \Pmap_\mu = \Pmap_{\mu'}$ for $\ket{\mu'}=D_j\ket{\mu}$ where 
 \[
 (D_j)_{kl}=\begin{cases} 1,& \text{ if }k=l= j \text{ or } k=l= 1\\ -1,& \text{ if } k=l\neq j\\ 0,& \text{ else.}\end{cases}.
 \]
 \end{enumerate} 
 \label{lem:LocalSym}
 \end{lem}

By the previous lemma we have the following:

 \begin{lem}[Restricted parameters]
 For any $(x,y,z)^T\in\R^3$ there exists unitaries $U,V\in\mathcal{U}_2$ such that 
 \[
 \Pmap_\mu = \text{Ad}_U\circ \Pmap_{\tilde{\mu}}\circ \text{Ad}_V,
 \] 
where $\ket{\mu}=(1,x,y,z)^T$ and $\ket{\tilde{\mu}}=(1,\tilde{x},\tilde{y},\tilde{z})^T$ with $(\tilde{x},\tilde{y},\tilde{z})^T\in\R^3$ such that
\[
\tilde{x}\geq\tilde{y}\geq \abs{\tilde{z}}.
\] 
In particular, for a mapping cone $\mathcal{C}$ and a positive map $Q:\M_{d_1}\ra\M_{d_2}$ we have $\Pmap_{\mu}\otimes Q\in\mathcal{C}$ if and only if $\Pmap_{\tilde{\mu}}\otimes Q\in\mathcal{C}$.  
 \label{lem:RestrParam}
 \end{lem} 

 \begin{proof}
By applying the first two statements of Lemma \ref{lem:LocalSym} there exist unitaries $U_1,V_1\in\mathcal{U}_2$ and $(x',y',z')^T\in \R^3$ such that 
\[
 \Pmap_\mu = \text{Ad}_{U_1}\circ \Pmap_{\mu'}\circ \text{Ad}_{V_1},
 \] 
where $\ket{\mu'}=(1,x',y',z')^T$ and
\[
\abs{x'}\geq\abs{y'}\geq \abs{z'}.
\]  
If $x'\geq 0$ and $y'\geq 0$, then we are done. In the other cases we apply the third statement of Lemma \ref{lem:LocalSym} (changing the sign of either both $x'$ and $y'$ together or changing the sign of either $x'$ or $y'$ together with the sign of $z'$) to find unitaries $U_2,V_2\in \mathcal{U}_2$ such that 
\[
 \Pmap_{\mu'} = \text{Ad}_{U_2}\circ \Pmap_{\tilde{\mu}}\circ \text{Ad}_{V_2},
 \] 
where $\ket{\tilde{\mu}}=(1,\tilde{x},\tilde{y},\tilde{z})$ for $(\tilde{x},\tilde{y},\tilde{z})^T\in \R^3$ satisfies the desired conditions. Setting $U=U_1U_2$ and $V=V_2V_1$ finishes the proof.
 \end{proof}

\subsection{Proof of Theorem \ref{thm:Main1}}
\label{subsec:DecTensSquares}

Following the ideas outlined in the previous sections, we first show the statement of Theorem \ref{thm:Main1} for normalized Pauli diagonal maps. 

\begin{thm}
For $(x,y,z)^T\in\R^3$ and $\ket{\mu}=(1,x,y,z)^T$ the following are equivalent:
\begin{enumerate}
\item $\Pmap_{\mu}\otimes \Pmap_{\mu}:\M_4\ra\M_4$ is decomposable.
\item $\Pmap_{\mu}\otimes \Pmap_{\mu}:\M_4\ra\M_4$ is positive.
\end{enumerate}
\label{thm:MainForPauli}
\end{thm}

\begin{proof}
It is clear that $1.$ implies $2.$. To show that $2.$ implies $1.$ consider the Pauli diagonal map $\Pmap_{\mu}:\M_2\ra\M_2$ and assume that $\Pmap_{\mu}\otimes \Pmap_{\mu}$ is positive. Since $\Pmap_{\mu}$ is positive and using Lemma \ref{lem:RestrParam} we can assume that 
\begin{equation}
1\geq x\geq y\geq \abs{z}.
\label{equ:paramCond}
\end{equation}
Applying the positive map $\Pmap_{\mu}\otimes \Pmap_{\mu}$ to the maximally entangled state $\omega_2\in (\M_2\otimes \M_2)^+$ shows that 
\[
(\Pmap_{\mu}\otimes \Pmap_{\mu})(\omega_2) = (\ident_2\otimes \Pmap_{\mu}\circ \Pmap_{\mu})(\omega_2) = C_{\Pmap_{\mu\circ \mu}}\geq 0, 
\]
where $\ket{\mu\circ \mu}\in\R^4$ denotes the vector with entries $(\mu\circ \mu)_i=\mu_i\mu_i$. Therefore, we conclude that $\Pmap_{\mu\circ \mu}:\M_2\ra\M_2$ is completely positive and by using the Fujiwara-Algoet criterion~\cite{fujiwara1999one} we obtain 
\begin{equation}
\begin{aligned}
1+x^2 &\geq y^2 + z^2 \\
1+y^2 &\geq x^2 + z^2 \\
1+z^2 &\geq x^2 + y^2
\end{aligned} ,
\label{equ:ParamCond2}
\end{equation}
which are the same conditions as in \eqref{equ:ParamCond}. Let $\ket{s}\in\R^4$ denote the spectral vector of the Choi matrix $C_{\Pmap_{\mu}}$, and by Theorem \ref{thm:CPCond} we find that 
\begin{equation}
\begin{aligned}
s_1 = \frac{1}{2}(1+x+y+z) \\
s_2 = \frac{1}{2}(1+x-y-z)  \\
s_3 = \frac{1}{2}(1-x+y-z) \\
s_4 = \frac{1}{2}(1-x-y+z) . 
\end{aligned}
\label{equ:spectrumByMu}
\end{equation}
Note that by \eqref{equ:paramCond} we have
\begin{equation}
s_1\geq s_2\geq s_3\geq \abs{s_4}.
\label{equ:orderWithBound}
\end{equation}
By Corollary \ref{cor:DecompTensorProdPauli} the positive map $\Pmap_{\mu}\otimes \Pmap_{\mu}$ is decomposable if and only if 
\begin{align}
\lbr s_{\sigma_1(1)} + s_{\sigma_1(4)}\rbr s_{\sigma_2(1)} + \lbr s_{\sigma_1(2)} + s_{\sigma_1(4)}\rbr s_{\sigma_2(2)} + \lbr s_{\sigma_1(3)} + s_{\sigma_1(4)}\rbr s_{\sigma_2(3)}\geq 0,
\label{equ:crossoverlap}
\end{align}
for any $\sigma_1,\sigma_2\in S_4$. By Theorem \ref{thm:DecompQubitPauli} the terms in the brackets are all positive and the expression in \eqref{equ:crossoverlap} can only be negative if $4\in \lset \sigma_2(1),\sigma_2(2),\sigma_2(3)\rset$ and without loss of generality we choose $\sigma_2(1)=4$. By \eqref{equ:orderWithBound} the smallest value for \eqref{equ:crossoverlap} will be obtained for $\lset \sigma_2(2),\sigma_2(3)\rset=\lset 2,3\rset$ and without loss of generality we choose $\sigma_2(2)=2$ and $\sigma_2(3)=3$. By \eqref{equ:orderWithBound} we have that $s_4+s_j\geq 0$ for any $j\in\lset 2,3\rset$. Hence, for \eqref{equ:crossoverlap} to be negative we need to have $\sigma_1(1)=1$. By the previous discussion we conclude that the positive map $\Pmap_{\mu}\otimes \Pmap_{\mu}$ is decomposable if and only if  
\begin{equation}
s_1s_4+s_{\sigma_1(2)}s_2 + s_{\sigma_1(3)}s_3 + s_{\sigma_1(4)}(2-s_1)\geq 0,
\label{equ:crossoverlap2}
\end{equation}
for any permutation $\sigma_1\in S_4$. By \eqref{equ:orderWithBound} and the normalization $\sum^4_{i=1} s_i=2$ we have that $(2-s_1)\geq s_2\geq s_3$. Therefore, we conclude that the smallest value of \eqref{equ:crossoverlap} is achieved for $\sigma_1(4)=4$, $\sigma_1(2)=3$ and $\sigma_1(3)=2$, where we used the elementary inequality $s^2_2+s^2_3\geq 2s_2s_3$. It remains to show that 
\[
\lbr s_1+s_4\rbr s_4 +\lbr s_3+s_4\rbr s_2+\lbr s_2+s_4\rbr s_3\geq 0
\]
for any $(s_1,s_2,s_3,s_4)^T\in\R^4$ arising as in \eqref{equ:spectrumByMu} from parameters $(x,y,z)^T\in\R^3$ where conditions \eqref{equ:ParamCond2} are satisfied. We compute 
\begin{align*}
\lbr s_1+s_4\rbr s_4 +\lbr s_3+s_4\rbr s_2+\lbr s_2+s_4\rbr s_3 &= \frac{1}{2}\lb 3+2(xy - x -y) +(-x^2-y^2+z^2)\rb \\
&\geq 1+xy - x -y = (1-x)(1-y) \geq 0,
\end{align*}
where we used \eqref{equ:ParamCond2} and that $\max(x,y)\leq 1$.

\end{proof}

Finally, we can prove our main result:

\begin{proof}[Proof of Theorem \ref{thm:Main1}]

It is clear that $2.$ implies $1.$. To show that $1.$ implies $2.$ assume that there exists a positive map $P:\M_2\ra\M_2$ such that $P\otimes P$ is positive but not decomposable. By Corollary \ref{cor:Main} this implies the existence of a positive Pauli diagonal map $\Pmap_\mu:\M_2\ra\M_2$ with $\ket{\mu}=(1,x,y,z)^T$ for some $(x,y,z)^T\in\R^3$ such that $\Pmap_\mu\otimes \Pmap_\mu$ is positive and not decomposable. However, by Theorem \ref{thm:MainForPauli} there are no such $(x,y,z)^T\in\R^3$ finishing the proof.
 
\end{proof}

Curiously, the equivalence of Theorem \ref{thm:Main1} is false when tensor products instead of tensor squares are considered or when the local dimension exceeds $2$. We give counterexamples in the next section.

\section{Non-decomposable positive maps from tensor products}
\label{Sec:NDPosFromTens}

In~\cite{muller2018decomposability} we found examples of a completely positive map $T:\M_d\ra\M_d$ and a completely copositive map $S:\M_d\ra\M_d$ such that their tensor product $T\otimes S$ is positive but not decomposable for any $d\geq 3$. This showed that non-decomposable positive maps can arise as tensor products of decomposable maps. However, our construction did not give any such example for $d=2$. In the next subsection we will find such examples. Moreover, we will use these examples to construct a decomposable map $P:\M_2\ra\M_4$ for which the tensor square $P\otimes P$ is positive but not decomposable. 

\subsection{Tensor products of qubit maps}
\label{subsec:NDTensProd}

For any $t\in\lbr 0,1\rbr$ the qubit depolarizing channel $T_t:\M_2\ra\M_2$ is defined as
\begin{equation}
T_t(X) = (1-t)\text{Tr}(X)\frac{\one_2}{2} + tX,
\label{equ:depol}
\end{equation}
and for any $a\in\lbr 0,1\rbr$ we define a positive map $\theta_a:\M_2\ra\M_2$ as
\begin{equation}
\theta_a(X) = (1-a)X + aX^T.
\end{equation}
It is well known that $T_t$ is entanglement-breaking for $t\leq 1/3$. Consequently, the tensor product $T_t\otimes \theta_a$ is positive for $t\leq 1/3$ and any $a\in\lbr 0,1\rbr$. In the following theorem we characterize all the pairs $(t,a)\in\lbr 0,1\rbr^2$ such that $T_t\otimes \theta_a$ is positive. 

\begin{thm}
The tensor product $T_t\otimes \theta_a:\M_4\ra\M_4$ is a positive map if and only if 
\[
t\leq \frac{1}{2a+1}.
\]
\label{thm:posProdParam}
\end{thm}

\begin{proof}
It can be checked easily, that 
\[
(T_t\otimes \theta_a)(\omega_2)\ngeq 0
\]
whenever $t>1/(2a+1)$ showing one direction of the statement. Note that
\[
T_t\otimes \theta_a = (T_{(2a+1)t}\otimes \ident_2)\circ (T_{1/(2a+1)}\otimes \theta_a)
\]
for any $t\leq 1/(2a+1)$ and any $a\in\lbr 0,1\rbr$. Since $T_{(2a+1)t}:\M_2\ra\M_2$ is completely positive, the statement of the theorem follows by showing that $T_{1/(2a+1)}\otimes \theta_a$ is positive for any $a\in\lbr 0,1\rbr$. 

Recall that any pure state $\ket{\psi}\in\C^2\otimes \C^2$ can be written as $\ket{\psi} = (X\otimes \one_2)\ket{\Omega_2}$ for some $X\in\M_2$. We have to show that 
\begin{equation}
(T_{1/(2a+1)}\otimes \theta_a)\circ \lb \text{Ad}_{X}\otimes \ident_2\rb\lb\omega_2\rb\geq 0
\label{equ:posCondMain}
\end{equation}
for any $a\in\lbr 0,1\rbr$ and any $X\in\M_2$. Applying the polar decomposition, we can write $X=UP$ for some unitary matrix $U\in\mathcal{U}_2$ and some positive matrix $P\in\M^+_2$. Using that $T_{1/(2a+1)}\circ \text{Ad}_U = \text{Ad}_U\circ T_{1/(2a+1)}$ and that $\text{Ad}_{U^\dagger}$ is completely positive it suffices to show \eqref{equ:posCondMain} for any positive $X\in\M^+_2$. Finally, we can normalize $X\in\M^+_2$ and by using the well-known parametrization (Bloch ball) of qubit states given by 
\[
B_1(\R^3)\ni x \mapsto \rho_x = \frac{\one_2}{2} + \sum^3_{i=1} x_i \sigma_i ,
\]  
we have to show \eqref{equ:posCondMain} for $X=\rho_x$ and any $x\in B_1(\R^3)$.

To show positivity of the matrix in \eqref{equ:posCondMain} we note first that it is block positive (i.e.~it is the Choi matrix of a positive map). By \cite[Theorem 3]{johnston2018inverse} such a matrix can have at most one negative eigenvalue. Therefore, we can use the determinant to determine when it is positive. It remains to show that the function $f:B_1(\R^3)\ra\R$ given by 
\[
f(x) = \text{Det}\lbr (T_{1/(2a+1)}\otimes \theta_a)\circ \lb \text{Ad}_{\rho_x}\otimes \ident_2\rb\lb\omega_2\rb\rbr 
\] 
only attains positive values. Using polar coordinates 
\[
x=\begin{pmatrix} r\text{sin}(\phi_1)\text{cos}(\phi_2) \\ r\text{sin}(\phi_1)\text{sin}(\phi_2) \\ r\text{cos}(\phi_1)\end{pmatrix}
\]
it is slightly tedious but straightforward\footnote{e.g.~using a computer algebra system} to compute 
\[
f(x) = \frac{a^4(1-a^2)r^2\text{cos}^2(\phi_1)\lbr 8(1-a^2)r^2(1+\text{cos}(2\phi_1))+16r^4-8r^2+1) \rbr}{(2a+1)^4}.
\]
Since $16r^4 - 8r^2+1\geq 0$ for any $r\in\lbr 0,1\rbr$ we find that $f(x)\geq 0$ for any $x\in B_1(\R^3)$. This finishes the proof.

\end{proof}

Since the map $T_t\otimes \theta_a:\M_4\ra\M_4$ is a Pauli diagonal map we can apply Corollary \ref{cor:DecompTensorProdPauli} to check when it is decomposable. We find the following. 

\begin{thm}
The tensor product $T_t\otimes \theta_a:\M_4\ra\M_4$ is positive and not decomposable if and only if 
\[
\frac{1}{3}<t\leq \frac{1}{2a+1} \quad\text{ and }\quad a>0.
\] 

\end{thm}

\begin{proof}
It is clear that $T_t\otimes \theta_a$ is decomposable whenever $t\leq 1/3$ or $a=0$ since in these cases either $T_t$ is entanglement breaking or $\theta_a$ is completely positive (note that $T_t$ is always completely positive). If $t>1/(2a+1)$, then by Theorem \ref{thm:posProdParam} the map $T_t\otimes \theta_a$ is not positive. Finally, note that $T_t = \Pmap_{\mu}$ for $\ket{\mu} = (1,t,t,t)^T$, and $\theta_a = \Pmap_{\mu'}$ for $\ket{\mu'}=(1,1,1,1-2a)^T$. Using Theorem \ref{thm:CPCond} we can compute the spectral vectors 
\[
\ket{s(t)} := \text{spec}\lb C_{T_t}\rb = \frac{1}{2}\begin{pmatrix} 1+3t \\ 1-t \\ 1-t\\ 1-t\end{pmatrix} \quad\text{ and }\quad \ket{s'(a)}:= \text{spec}\lb C_{\theta_a}\rb = \begin{pmatrix} 2-a \\ a \\ a\\ -a\end{pmatrix} .
\] 
Finally, we compute
\[
s(t)_1s'(a)_4 + s(t)_2s'(a)_2 + s(t)_3s'(a)_3 + s(t)_4s'(a)_4 + s(t)_4s'(a)_2 + s(t)_4s'(a)_3 = 2a(1-3t) . 
\]
By the third inequality in Theorem \ref{thm:SpectralDecompCritN2} with permutations $\sigma_1=\ident$ and $\sigma_2 = (14)$ we conclude that the map $T_t\otimes \theta_a$ is not decomposable for any $t>1/3$ and any $a>0$. Together with Theorem \ref{thm:posProdParam} this finishes the proof.

\end{proof}

In Figure \ref{fig:NDRegion} we plotted the parameters $(a,t)$ where the map $T_t\otimes \theta_a$ is positive. There is a region of parameters where it is also not decomposable. In contrast to the tensor squares, tensor products of qubit maps can be positive and not decomposable. 

\begin{figure*}[t!]
        \centering
        \includegraphics[scale=0.7]{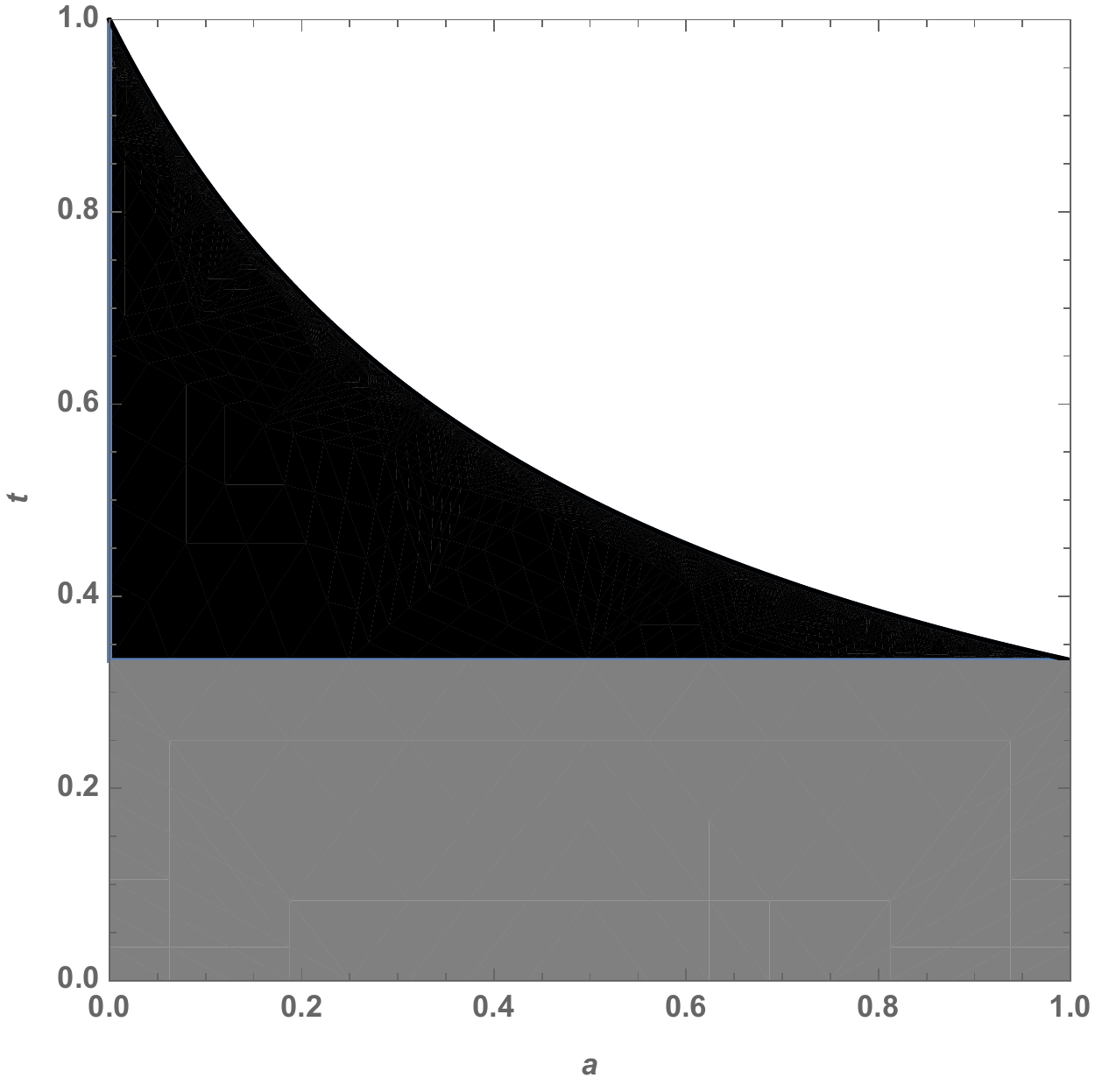}
        \caption{The parameters $(a,t)$ where $T_t\otimes \theta_a$ is positive but not decomposable are the interior of the black region together with the upper boundary given by the curve $(a,1/(2a+1))$ for $0<a<1$. The gray region contains the parameters where the map is decomposable.}
        \label{fig:NDRegion}
\end{figure*}

To conclude this section we will give another family of non-decomposable positive maps arising as tensor products of qubit maps. We will need this family in the next section to construct a positive map with a positive tensor square that is not decomposable. For $b\in \lbr 0,1\rbr$ consider the Pauli diagonal map $\Pi_{\lambda_b}:\M_2\ra\M_2$ with $\lambda_b=(1,b,0,b)^T$, i.e.~the linear map
\[
\Pi_{\lambda_b}\lb X\rb = \text{Tr}\lb X\rb\frac{\one_2}{2} + \frac{b}{2}\lb \text{Tr}\lb X\sigma_2\rb\sigma_2 +\text{Tr}\lb X\sigma_4\rb\sigma_4 \rb.
\] 
It can be checked that $\Pi_{\lambda_b}$ is positive for any $b\in\lbr 0,1\rbr$. With the depolarizing channel $T_t:\M_2\ra\M_2$ for $t\in\lbr 0,1\rbr$ as in \eqref{equ:depol} we have the following theorem:

\begin{thm}
For $t,b\in\lbr 0,1\rbr$ the tensor product $T_t\otimes \Pi_{\lambda_b}:\M_4\ra\M_4$ is positive if and only if $t\leq \frac{1}{2b}$.
\label{thm:posReg2}
\end{thm}

\begin{proof}
It can be checked easily, that 
\[
(T_t\otimes \Pi_{\lambda_b})(\omega_2)\ngeq 0
\]
whenever $t>1/2b$ showing one direction of the statement. For the other direction we use a similar strategy as for the proof of Theorem \ref{thm:posProdParam}. First, we note that $\Pi_{\lambda_b}$ is completely positive for any $b\in \lbr 0,1/2\rbr$. Therefore, for $b\in \lbr 0,1/2\rbr$ and any $t\in \lbr 0,1\rbr$ the linear map $T_t\otimes \Pi_{\lambda_b}$ is positive. Assume now that $b\geq 1/2$ and that $t\leq 1/2b$. Since
\[
T_t\otimes \Pi_{\lambda_b} = \lb T_{2bt}\otimes \ident_2\rb\circ \lb T_{(1/2b)}\otimes \Pi_{\lambda_b} \rb
\] 
and $T_{2tb}$ is completely positive, it is sufficient to show that $T_{1/2b}\otimes \Pi_{\lambda_b}$ is positive. As in the proof of Theorem \ref{thm:posProdParam} this follows from positivity of the function $f:B_1(\R^3)\ra\R$ on the Bloch ball $B_1(\R^3)$ given by 
\[
f(x) = \text{Det}\lbr (T_{1/(2b)}\otimes \Pi_{\lambda_b})\circ \lb \text{Ad}_{\rho_x}\otimes \ident_2\rb\lb\omega_2\rb\rbr ,
\] 
where 
\[
\rho_x = \frac{1}{2}\lb \one_2 + \sum^3_{i=1} x_i \sigma_i \rb.
\]
Using polar coordinates 
\[
x=\begin{pmatrix} r\text{sin}(\phi_1)\text{cos}(\phi_2) \\ r\text{sin}(\phi_1)\text{sin}(\phi_2) \\ r\text{cos}(\phi_1)\end{pmatrix}
\]
it is again slightly tedious but straightforward\footnote{e.g.~using a computer algebra system.} to compute 
\[
f(r,\phi_1,\phi_2) = \frac{(4b^2-1)}{16384b^2}r^2\lbr c_1(\phi_1,\phi_2)r^4 + \frac{c_2(\phi_1,\phi_2)}{16b^2}r^2 + c_1(\phi_1,\phi_2)\rbr
\]
with 
{\scriptsize
\begin{align*}
c_1(\phi_1,\phi_2) &= 7-6b^2-(2b^2-1)\cos(2\phi_1)-2(2b^2-1)\cos(2\phi_2)\sin(\phi_1)^2,\\
c_2(\phi_1,\phi_2) &=-64+128b^2-233b^4+164b^6+b^2(3b^2(4b^2-1)\cos(4\phi_1)+16(12b^4-11b^2)\cos(2\phi_2)\sin(\phi_1)^2 +\cdots\\
&\cdots + 8b^2(4b^2-1)\cos(4\phi_2)\sin(\phi_1)^4 + 4\cos(2\phi_1)(20b^4-21b^2+4b^2(4b^2-1)\cos(2\phi_2)\sin(\phi_1)^2))
\end{align*}}%
It is easy to check that $c_1(\phi_1,\phi_2)\geq 0$ for any $b\in \lbr 1/2,1\rbr$, any $\phi_1\in \lbr 0,\pi\rbr$ and any $\phi_2\in\lbr 0,2\pi\rbr$. We will now argue that for fixed $\phi_1\in \lbr 0,\pi\rbr$ and $\phi_2\in\lbr 0,2\pi\rbr$ the polynomial $r\mapsto f(r,\phi_1,\phi_2)$ is positive for any $r\in\lbr 0,1\rbr$. Clearly, this is true when $c_2(\phi_1,\phi_2)\geq 0$, and we will assume $c_2(\phi_1,\phi_2)<0$ in the following. Since any $x\in \partial B_1(\R^3)$ corresponds to a pure state $\rho_x=\proj{\psi}{\psi}$ we have 
\[
(T_{1/(2b)}\otimes \Pi_{\lambda_b})\circ \lb \text{Ad}_{\rho_x}\otimes \ident_2\rb\lb\omega_2\rb = T_{1/(2b)}(\proj{\psi}{\psi})\otimes \Pi_{\lambda_b}(\overline{\proj{\psi}{\psi}})\geq 0 
\] 
by positivity of $T_{1/(2b)}$ and $\Pi_{\lambda_b}$. Therefore, we have $f(1,\phi_1,\phi_2)\geq 0$ and equivalently
\[
2c_1(\phi_1,\phi_2) + \frac{c_2(\phi_1,\phi_2)}{16b^2}\geq 0.
\]
for any $b\in\lbr 1/2,1\rbr$, any $\phi_1\in \lbr 0,\pi\rbr$ and any $\phi_2\in\lbr 0,2\pi\rbr$. Note that $c_1(\phi_1,\phi_2)=0$ implies that $c_2(\phi_1,\phi_2)\geq 0$, and in the following we can assume that $c_1(\phi_1,\phi_2)>0$. The previous inequality implies that the zeros 
\[
r_{1,2} = -\frac{c_2(\phi_1,\phi_2)}{32b^2c_1(\phi_1,\phi_2)} \pm \sqrt{\frac{c_2(\phi_1,\phi_2)^2}{1024b^4c_1(\phi_1,\phi_2)^2}-1}
\]
of the polynomial 
\[
q(r) = c_1(\phi_1,\phi_2)r^2 + \frac{c_2(\phi_1,\phi_2)}{16b^2}r + c_1(\phi_1,\phi_2)
\]
either coincide, or are not real. As $q(0)>0$ and $q(1)\geq 0$ we conclude that the polynomial $r\mapsto p(r,\phi_1,\phi_2) = \frac{(4b^2-1)}{16384b^2}r^2 q(r^2)$ does not attain any negative values since otherwise the polynomial $q$ would need to have two positive zeros. This finishes the proof.

\end{proof}

Again, we can determine the parameters where the tensor product $T_t\otimes \Pi_{\lambda_b}$ is positive but not decomposable. 

\begin{thm}
The tensor product $T_t\otimes \Pi_{\lambda_b}:\M_4\ra\M_4$ is positive and not decomposable if and only if 
\[
t\leq \frac{1}{2b} \quad\text{ and }\quad 3<2b+t+2bt.
\] 
\label{thm:NDParam2}
\end{thm}

\begin{proof}

By Theorem \ref{thm:posReg2} the tensor product $T_t\otimes \Pi_{\lambda_b}$ is not positive whenever $t>1/2b$. Note that $T_t = \Pmap_{\mu}$ for $\mu = (1,t,t,t)^T$. Using Theorem \ref{thm:CPCond} we can compute the spectral vectors
\[
s(t) := \text{spec}\lb C_{T_t}\rb = \frac{1}{2}\begin{pmatrix} 1+3t \\ 1-t \\ 1-t\\ 1-t\end{pmatrix} \quad\text{ and }\quad s'(b):= \text{spec}\lb C_{\lambda_b}\rb = \frac{1}{2}\begin{pmatrix} 1+2b \\ 1 \\ 1-2b\\ 1\end{pmatrix} .
\] 
After reordering we find that 
\[
s(t)_1s'(b)_3 + s(t)_2s'(b)_2 + s(t)_3s'(b)_4 + s(t)_4s'(b)_3 + s(t)_4s'(b)_2 + s(t)_4s'(b)_4 = 2(3-2b-t-2tb) . 
\]
By the third inequality in Theorem \ref{thm:SpectralDecompCritN2} with permutations $\sigma_1=\ident$ and $\sigma_2 = (134)$ we conclude that the map $T_t\otimes \lambda_b$ is not decomposable whenever $3<2b+t+2tb$. It is easy to check that the inequalities from Theorem \ref{thm:SpectralDecompCritN2} are all satisfied whenever $3\geq 2b+t+2tb$ and $t\leq 1/2b$.   

\end{proof}

In Figure \ref{fig:NDRegion2} we have plotted the parameters $(b,t)$ where the tensor product $T_t\otimes \Pi_{\lambda_b}$ is positive. Again, there is a region of parameters for which the tensor product is positive but not decomposable. 

\begin{figure*}[t!]
        \centering
        \includegraphics[scale=0.7]{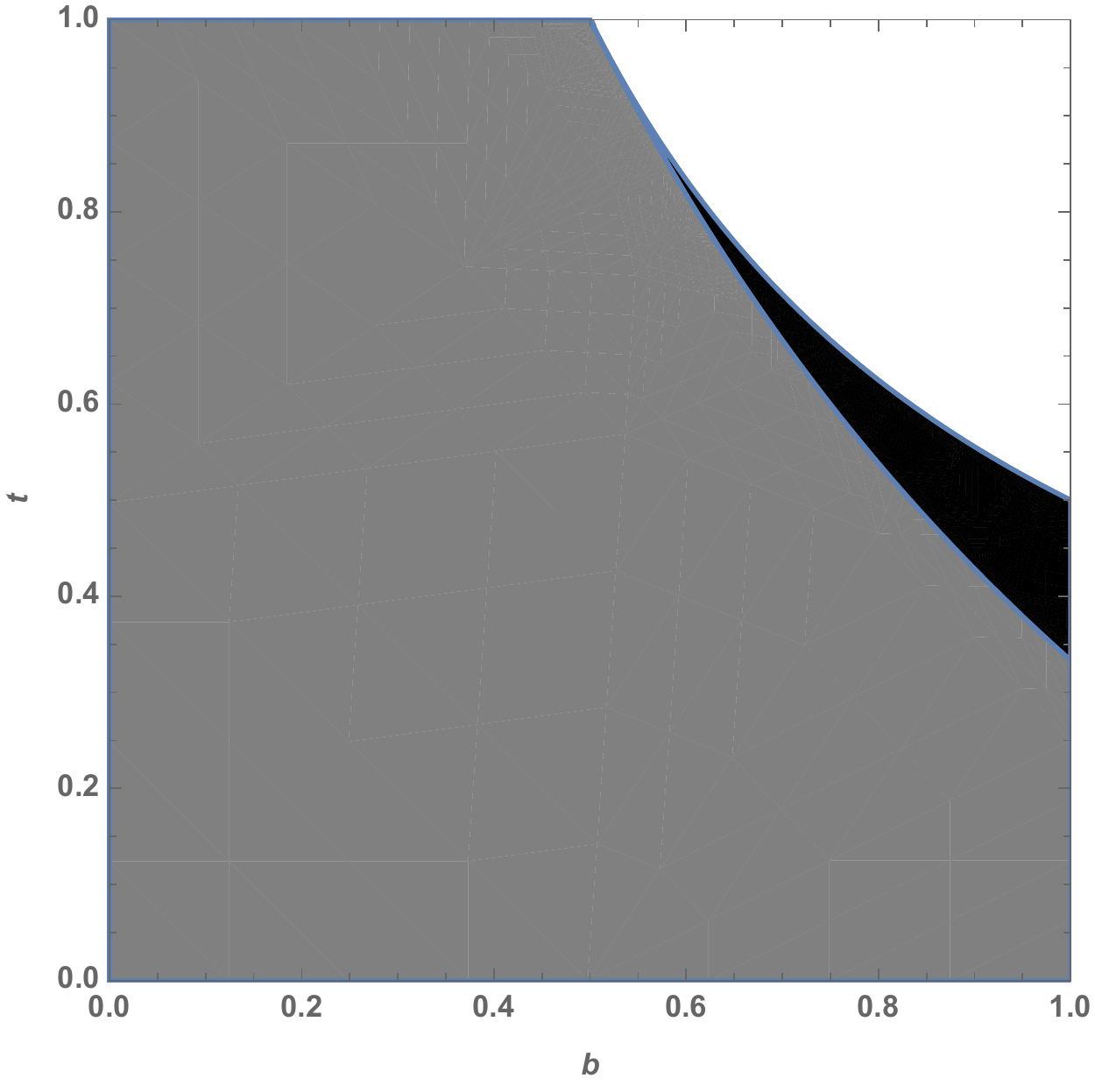}
        \caption{The parameters $(b,t)$ where $T_t\otimes \Pi_{\lambda_b}$ for $\lambda_b=(1,b,0,b)^T$ is positive but not decomposable are the interior of the black region together with the upper boundary given by the curve $(b,1/2b)$ for $1/2<b\leq 1$. The gray region contains the parameters where the map is decomposable.}
        \label{fig:NDRegion2}
\end{figure*}

\subsection{Positive tensor squares that are not decomposable}
\label{subsec:NDTensSquares}

In this section we will show that for any $d_1\geq 2$ and any $d_2\geq 4$ there exist decomposable maps $P:\M_{d_1}\ra\M_{d_2}$ such that the tensor square $P\otimes P$ is positive and not decomposable. We will need the following lemma:

\begin{lem}[Switch trick]
Let $P_1,P_2:\M_{d_1}\ra\M_{d_2}$ be positive maps such that $P_1\otimes P_1$, $P_2\otimes P_2$, and $P_1\otimes P_2$ are positive. If $P_1\otimes P_2$ is not decomposable, then for the map $P:\M_{d_1}\ra\M_{2d_2}$ given by 
\[
P(X) = \begin{pmatrix} P_1(X) & 0 \\ 0 & P_2(X)\end{pmatrix}
\]
the tensor square $P\otimes P$ is positive but not decomposable. 
\label{lem:switchy}
\end{lem}

\begin{proof}
With the computational basis $\lset \ket{1},\ket{2}\rset \in \C^2$ we can write
\[
P(X)= P_1(X)\otimes \proj{1}{1} + P_2(X)\otimes \proj{2}{2}.
\]
Then, we have 
\[
P\otimes P = \sum_{i,j} P_i\otimes\proj{i}{i}\otimes  P_j\otimes \proj{j}{j}. 
\]
Since the maps $P_1\otimes P_1$, $P_2\otimes P_2$ and $P_1\otimes P_2$ are all positive, we find that $P\otimes P$ is positive. Assume now for contradiction that $P\otimes P$ is decomposable. Then, there exist completely positive maps $T,S:\M_{d_1}\ra\M_{2d_2}$ such that 
\[
P\otimes P = T + \vartheta_{2d_2}\circ S.
\]
However, then the map 
\[
P_1\otimes P_2 = \lb\ident_{d_2}\otimes  \text{Ad}_{\bra{1}}\otimes\ident_{d_2}\otimes  \text{Ad}_{\bra{2}} \rb\circ \lb P\otimes P\rb
\]
would be decomposable as well, contradicting the assumption. 
\end{proof}

Finally, we can show the following theorem.

\begin{thm}
For any $d_1\geq 2$ and $d_2\geq 4$ there exists a decomposable map $P:\M_{d_1}\ra\M_{d_2}$ such that the tensor square $P\otimes P$ is positive and not decomposable. 
\end{thm} 
\begin{proof}
Consider the quantum channel $T_{3/4}:\M_2\ra\M_2$ as defined in \eqref{equ:depol} and the positive Pauli diagonal map $\Pi_{\lambda_{2/3}}:\M_2\ra\M_2$ for $\lambda_{2/3}=(1,2/3,0,2/3)^T$. Since $T_{3/4}$ is completely positive its tensor square $T_{3/4}\otimes T_{3/4}$ is positive. By \eqref{equ:ParamCond} the tensor square $\Pi_{\lambda_{2/3}}\otimes \Pi_{\lambda_{2/3}}$ is positive as well. Finally, by Theorem \ref{thm:NDParam2} the tensor product $T_{3/4}\otimes \Pi_{\lambda_{2/3}}$ is positive but not decomposable. We conclude that the map $P:\M_2\ra\M_4$ given by
\[
P(X) = \begin{pmatrix} T_{3/4}(X) & 0 \\ 0 & \Pi_{\lambda_{2/3}}(X)\end{pmatrix}
\]
is decomposable, and by Lemma \ref{lem:switchy} it has a positive tensor square $P\otimes P$ that is not decomposable. The statement of the theorem now follows from embedding this example in suitable higher dimensions.

\end{proof}

Another way to construct positive and non-decomposable tensor squares (or higher powers) can be found in the proof of~\cite[Theorem 1]{muller2016positivity}. Using unextendible product bases~\cite{divincenzo2003unextendible} this construction can be used to find for fixed $n\in\N$ a non-decomposable positive map $P:\M_{3}\ra\M_{3}$ such that $P^{\otimes n}$ is positive (and trivially non-decomposable). However, we have not been able to use this construction to find decomposable maps with non-decomposable but positive tensor powers. In particular, we do not know of an example of a decomposable map $P:\M_3\ra\M_3$ such that $P\otimes P$ is positive and non-decomposable. We expect that such an example exists.

\section{Conclusion and open questions}
\label{sec:Conclusion}
 
To characterize when positive Pauli diagonal maps are decomposable, we studied the polyhedral cone of Pauli diagonal maps that are both completely positive and completely copositive. Using the one-to-one correspondence between these maps and the spectra of their Choi matrices, we introduced the cone of Pauli PPT spectra $\mathcal{S}\lb \CP_N\cap\coCP_N\rb$ and analyzed its extremal rays. For $N=1$ and $N=2$ we found all extremal rays of this cone, and we used these to characterize the decomposable Pauli diagonal maps $\Pmap^{(N)}_\mu:\M^{\otimes N}_2\ra\M^{\otimes N}_{2}$ for $N=1$ and $N=2$. As an application of our results, we extended St\o rmer's theorem by showing that every positive tensor square $P\otimes P$ of linear maps $P:\M_2\ra\M_2$ is decomposable. Finally, we provided examples of linear maps $P,Q:\M_2\ra\M_2$ for which $P\otimes Q$ is positive but not decomposable, and an example of a decomposable map $P':\M_2\ra\M_4$ for which the tensor square $P'\otimes P'$ is positive but not decomposable. We finish with some open questions: 

\begin{itemize}
\item Is there a shorter (or more insightful) proof for Theorem \ref{thm:ExtremeRaysN2}?
\item What are the extremal rays of $\mathcal{S}\lb \CP_N\cap\coCP_N\rb$ for $N\geq 3$? We have some partial results in the case $N=3$ to be included in future work~\cite{fulvio20xx}, but even in this case the general structure of the extremal rays and their orbits under symmetry seems to be complicated. 
\item For every $n\in\N$ we can consider the symmetric orthogonal matrices 
\[
K_n = \frac{2}{n}E_n - \one_n.
\]
What are the ew-positive matrices $P\in\M_n(\R)$ such that $K_nPK_n$ is ew-positive as well? How many extremal rays does the corresponding polyhedral cone have? 
\item Is there a map $P:\M_2\ra\M_3$ such that the tensor square $P\otimes P$ is positive but not decomposable? If not, then what about higher powers? This would generalize Woronowicz's theorem~\cite{woronowicz1976positive}.  
\item Is there a positive map $P:\M_2\ra\M_2$ such that the tensor cube $P^{\otimes 3}$ is positive but not decomposable? If not, then what about higher powers?  
\end{itemize}

\section*{Acknowledgments}
We thank Guillaume Aubrun and Fulvio Gesmundo for interesting discussions and valuable comments that improved this article. Special thanks go to Linn Elki{\ae}r for diverting discussions about the content of Figure \ref{fig:starrySurf}. We acknowledge funding from the European Union’s Horizon 2020 research and innovation programme under the Marie Sk\l odowska-Curie Action TIPTOP (grant no. 843414).

\newpage

\appendix

\section{Proof of Theorem \ref{thm:ExtremeRaysN2}}
\label{app:Therest}

We will now present an analytic proof of Theorem \ref{thm:ExtremeRaysN2} characterizing the extremal rays of $\mathcal{S}\lb \CP_2\cap\coCP_2\rb$ (cf.~Definition \ref{defn:SPPPoly}). Recall that there are two statements to show: First, we will show that the following pairs generate extremal rays in $\mathcal{S}\lb \CP_2\cap\coCP_2\rb$.
\begin{align*}
(P_1,Q_1) &= \lb\begin{pmatrix} 1 & 1 & 0 & 0 \\ 1 & 1 & 0 & 0 \\ 0 & 0 & 0 & 0 \\ 0 & 0 & 0 & 0\end{pmatrix},\begin{pmatrix} 0 & 0 & 0 & 0 \\ 0 & 0 & 0 & 0 \\ 0 & 0 & 1 & 1 \\ 0 & 0 & 1 & 1\end{pmatrix}\rb \\
(P_2,Q_2) &= \lb\begin{pmatrix} 1 & 0 & 0 & 0 \\ 0 & 1 & 0 & 0 \\ 0 & 0 & 1 & 0 \\ 0 & 0 & 0 & 1\end{pmatrix},\begin{pmatrix} 1 & 0 & 0 & 0 \\ 0 & 1 & 0 & 0 \\ 0 & 0 & 1 & 0 \\ 0 & 0 & 0 & 1\end{pmatrix}\rb \\
(P_3,Q_3) &= \lb\begin{pmatrix} 1 & 0 & 0 & 0 \\ 0 & 1 & 0 & 0 \\ 0 & 0 & 1 & 0 \\ 1 & 1 & 1 & 0\end{pmatrix},\begin{pmatrix} 1 & 0 & 0 & 1 \\ 0 & 1 & 0 & 1 \\ 0 & 0 & 1 & 1 \\ 0 & 0 & 0 & 0\end{pmatrix}\rb.
\end{align*}
Second, we have to show that this list of extremal rays is complete, i.e.~such that
\[
\mathcal{S}\lb \CP_2\cap\coCP_2\rb = \text{cone}\lb \bigcup^3_{i=1} \text{orb}_2\lbr (P_i,Q_i) \rbr\rb
\]
where $\text{orb}_2$ denotes the orbits under the symmetry group described in Theorem \ref{thm:OrbitStruc}. For abbreviation it will be helpful to use the terminology introduced in the statement of Theorem \ref{thm:ExtremeRaysN2}: For any $\alpha>0$ we call the elements of
\begin{itemize}
\item $\alpha \text{orb}_2\lbr (P_1,Q_1)\rbr$ boxes.
\item $\alpha \text{orb}_2\lbr (P_2,Q_2)\rbr$ diagonals.
\item $\alpha \text{orb}_2\lbr (P_3,Q_3)\rbr$ crosses.
\end{itemize}
It is easy to check that every box is of the form 
\[
\alpha \lb \proj{\lset i,j\rset}{\lset k,l\rset},\proj{\lset i,j\rset^c}{\lset k,l\rset^c}\rb
\]
for some $\alpha>0$ and $i,j,k,l\in\lset 1,2,3,4\rset$ satsifying $i<j$ and $k<l$ (the notation $\ket{\lset i,j\rset}$ was previously introduced before Theorem \ref{thm:ExtremalRaysN1}). Similarly, it is easy to check that every diagonal is of the form $\alpha \lb U_\sigma,U_\sigma\rb$ for some $\alpha>0$ and a permutation matrix $U_\sigma\in \M_4$ corresponding to a permutation $\sigma\in S_4$. Finally, by definition every cross from is of the form $\alpha\lb U_{\sigma_1}P_{3}U_{\sigma_2},U_{\sigma_1}P^T_{3}U_{\sigma_2}\rb$ or $\alpha\lb U_{\sigma_1}P^T_{3}U_{\sigma_2},U_{\sigma_1}P_{3}U_{\sigma_2}\rb$ for some $\alpha>0$ and a permutations $\sigma_1,\sigma_2\in S_4$. 

To prove Theorem \ref{thm:ExtremeRaysN2} it will be useful to formulate the orthogonality relations from Lemma \ref{lem:Orthogonality} explicitely for the boxes, diagonals and crosses giving the following lemmas. 

\begin{lem}[Box rule]
Let $(P,Q)\in \mathcal{S}\lb \CP_2\cap\coCP_2\rb$. For all $i,j,k,l\in\lset 1,2,3,4\rset$ satsifying $i<j$ and $k<l$ the following are equivalent:
\begin{enumerate}
\item $\bra{\lset i,j\rset} P \ket{\lset k,l\rset}=0$ .
\item $\bra{\lset i,j\rset^c} Q \ket{\lset k,l\rset^c}=0$. 
\end{enumerate}
\label{lem:boxRule}   
\end{lem}

\begin{lem}[Diagonal rule]
Let $(P,Q)\in \mathcal{S}\lb \CP_2\cap\coCP_2\rb$. For any permutation $\sigma\in S_4$ the following are equivalent:
\begin{enumerate}
\item $\text{Tr}\lb U_\sigma P\rb =0$.
\item $\text{Tr}\lb U_\sigma Q\rb =0$.
\end{enumerate}
\label{lem:diagRule}   
\end{lem}

\begin{lem}[Cross rule]
Let $(P,Q)\in \mathcal{S}\lb \CP_2\cap\coCP_2\rb$. For any permutations $\sigma_1,\sigma_2\in S_4$ the following are equivalent:
\begin{enumerate}
\item $\text{Tr}\lb U_{\sigma_1}P_{3}U_{\sigma_2} P\rb =0$.
\item $\text{Tr}\lb U_{\sigma_1}P^T_{3}U_{\sigma_2} Q\rb =0$.
\end{enumerate}
The same equivalence also holds with the roles of $P$ and $Q$ exchanged.
\label{lem:crossRule}   
\end{lem}

We found it helpful to visualize the box rule as
\[
\begin{pmatrix}
* & * & * & * \\ 0 & * & 0 & * \\ * & * & * & * \\ 0 & * & 0 & *
\end{pmatrix}\subseteq \mathcal{Z}\lb P\rb \Leftrightarrow \begin{pmatrix}
* & 0 & * & 0 \\ * & * & * & * \\ * & 0 & * & 0 \\ * & * & * & *
\end{pmatrix}\subseteq \mathcal{Z}\lb Q\rb,
\]
here for the special case of $i=2, j=4 ,k=1,$ and $l=3$. The diagonal rule can be visualized as  
\[
\begin{pmatrix}
* & 0 & * & * \\ * & * & 0 & * \\ * & * & * & 0 \\ 0 & * & * & *
\end{pmatrix}\subseteq \mathcal{Z}\lb P\rb \Leftrightarrow \begin{pmatrix}
* & 0 & * & * \\ * & * & 0 & * \\ * & * & * & 0 \\ 0 & * & * & *
\end{pmatrix}\subseteq \mathcal{Z}\lb Q\rb,
\]
here for the special case $\sigma=(1432)$. Finally, we can visualize the cross rule as
\[
\begin{pmatrix}
0 & * & * & * \\ * & 0 & * & * \\ * & * & 0 & * \\ 0 & 0 & 0 & *
\end{pmatrix}\subseteq \mathcal{Z}\lb P\rb \Leftrightarrow \begin{pmatrix}
0 & * & * & 0 \\ * & 0 & * & 0 \\ * & * & 0 & 0 \\ * & * & * & *
\end{pmatrix}\subseteq \mathcal{Z}\lb Q\rb,
\]
here for the special case $\sigma_1=\sigma_2=\ident_4$. 

The difficult part of the proof of Theorem \ref{thm:ExtremeRaysN2} is to show that the list of extremal rays stated above is complete. To do so, we will show that the zero pattern of any extremal ray of the cone $\mathcal{S}\lb \CP_2\cap\coCP_2\rb$ has to contain at least $8$ zeros. By classifying all zero patterns with $8$ zeros we will have a list of subpatterns that have to occur in each zero pattern of an extremal ray. Finally, we will show that every extremal ray whose zero pattern contains one of these subpatterns has to be a box, a diagonal or a cross.

\subsection{Extremality of Boxes, Diagonals and Crosses}
\label{sec:ExtrDiagCrossBox}

\begin{thm}[Extremality of boxes, diagonals and crosses]

The elements 
\[
(P_1,Q_2),(P_2,Q_2),(P_3,Q_3)\in \mathcal{S}\lb \CP_2\cap\coCP_2\rb
\]
introduced in Theorem \ref{thm:ExtremeRaysN2} generate extremal rays.
\label{thm:ExtrThreePairs}
\end{thm}

\begin{proof}
By Theorem \ref{thm:ExtremalRaysN1} for any $i,j\in\lset 1,2,3,4\rset$ satisfying $i<j$ the vector 
\[
\ket{\lset i,j\rset} = \ket{i}+\ket{j} \in \R^4
\]
generates an extremal ray of $\mathcal{S}\lb \CP_1\cap\coCP_1\rb$. By Lemma \ref{lem:TensorProdExtremal} this shows that the boxes arising as (scalings of) tensor products $\ket{\lset i,j\rset}\otimes \ket{\lset k,l\rset}$ for $i<j$ and $k<l$ generate extremal rays of $\mathcal{S}\lb \CP_2\cap\coCP_2\rb$.

By Lemma \ref{lem:Permutations} it is sufficient to show that 
\[
(P_2,Q_2) = \lb\begin{pmatrix} 1 & 0 & 0 & 0 \\ 0 & 1 & 0 & 0 \\ 0 & 0 & 1 & 0 \\ 0 & 0 & 0 & 1\end{pmatrix},\begin{pmatrix} 1 & 0 & 0 & 0 \\ 0 & 1 & 0 & 0 \\ 0 & 0 & 1 & 0 \\ 0 & 0 & 0 & 1\end{pmatrix}\rb
\]
generates an extremal ray of $\mathcal{S}\lb \CP_2\cap\coCP_2\rb$. Consider a pair $(P,Q)\in \mathcal{S}\lb \CP_2\cap\coCP_2\rb$ satisfying $(P_2,Q_2) <_Z (P,Q)$. This implies that either $P$ or $Q$ have at least one zero on their diagonal. Since every off-diagonal element of $P$ and $Q$ is zero, we can apply Lemma \ref{lem:boxRule} repeatedly to show that $P=Q=0$. By Lemma \ref{lem:extremesZP} we have shown that $(P_2,Q_2)$ generates an extremal ray. 

By Lemma \ref{lem:Permutations} it is again sufficient to show that 
\[
(P_3,Q_3) = \lb\begin{pmatrix} 1 & 0 & 0 & 0 \\ 0 & 1 & 0 & 0 \\ 0 & 0 & 1 & 0 \\ 1 & 1 & 1 & 0\end{pmatrix},\begin{pmatrix} 1 & 0 & 0 & 1 \\ 0 & 1 & 0 & 1 \\ 0 & 0 & 1 & 1 \\ 0 & 0 & 0 & 0\end{pmatrix}\rb
\]
generates an extremal ray. Again, Consider a pair $(P,Q)\in \mathcal{S}\lb \CP_2\cap\coCP_2\rb$ satisfying $(P_3,Q_3) <_Z (P,Q)$. There are two cases: Either $P$ or $Q$ could have at least two zeros on their diagonal, or either $P$ or $Q$ could have two zeros in the fourth row or column respectively. In both cases we can apply Lemma \ref{lem:boxRule} repeatedly to show that $P=Q=0$, and by Lemma \ref{lem:extremesZP} we conclude that $(P_2,Q_2)$ generates an extremal ray.    

\end{proof}

\subsection{Combinatorics of zero patterns with $8$ zeros}

It will be convenient to identify zero patterns $Z'\subset \lset 1,2,3,4\rset^2$ (see Definition \ref{defn:ZPattern}) with $(0,1)$-matrices $Z\in \M_4\lb \lset 0,1\rset\rb$ such that $Z_{ij}=0$ if and only if $(i,j)\in Z'$ for $i,j\in\lset 1,2,3 ,4\rset$. To prove Theorem \ref{thm:ExtremeRaysN2} we will classify all $(0,1)$-matrices $Z\in \M_4\lb \lset 0,1\rset\rb$ with $8$ zeros up to row and column permutations using a result by R.~A.~Brualdi~\cite{BRUALDI20063054} computing the number of $(0,1)$-matrices with prescribed row and column sums.

Consider two integer partitions $r=(r_1,r_2,r_3,r_4)\in\N_0^4$ and $s=(s_1,s_2,s_3,s_4)\in \N^4_0$ of the number $8$ into $4$ parts each smaller than $4$, i.e.~such that 
\[
8 = r_1+r_2+r_3+r_4 = s_1+s_2+s_3+s_4.
\]
and  
\[
4\geq r_1\geq r_2\geq r_3\geq r_4\geq 0 \quad\text{ and }\quad 4\geq s_1\geq s_2\geq s_3\geq s_4\geq 0.
\]  
We denote by $r\preccurlyeq s$ the usual majorization ordering, i.e.
\[
\sum^k_{i=1} r_i \leq \sum^k_{i=1} s_i
\]
for all $k\in\lset 1,2,3\rset$ and equality for $k=4$. Moreover, we denote by $r^*=(r^*_1,r^*_2,r^*_3,r^*_4)$ the conjugate partition with entries
\[
r^*_j = |\lset i\in \lset 1,2,3,4\rset ~:~r_i\geq j\rset|.
\] 
Consider now the set of all $(0,1)$-matrices in $\M_4\lb\lset 0,1\rset\rb$ with row sum vector $r=(r_1,r_2,r_3,r_4)$ and column sum vector $s = (s_1,s_2,s_3,s_4)$ denoted by  
\[
\mathcal{A}\lb r,s\rb := \lset Z\in \M_4\lb\lset 0,1\rset\rb ~:~\sum^4_{j=1} Z_{kj} = r_k \text{ and }\sum^4_{i=1} Z_{il} = s_l \text{ for any }k,l\in\lset 1,2,3,4\rset\rset.
\]
Note that $\mathcal{A}\lb R,S\rb$ may be empty. For partitions $r=(r_1,r_2,r_3,r_4)$ and $s = (s_1,s_2,s_3,s_4)$ of $8$ such that $r_1,s_1\leq 4$ we can apply a result by R.~A.~Brualdi (see~\cite[Equation 3]{BRUALDI20063054}) to determine the cardinality $\mathcal{A}\lb R,S\rb$:
\begin{equation}
|\mathcal{A}\lb R,S\rb| = \sum_{s\preccurlyeq \lambda \preccurlyeq r^*} K_{\lambda^*,r}K_{\lambda,s},
\label{equ:Brualdi}
\end{equation}
where the sum runs over integer partitions $\lambda=(\lambda_1,\lambda_2,\lambda_3,\lambda_4)$ of the number $8$ into $4$ parts, and where $K_{\lambda,\mu}$ denote the Kostka numbers, i.e.~the number of Young tableaux with shape $\lambda$ and content $\mu$ (for details on Young tableaux and Kostka numbers see~\cite{fulton1997young}). Note that the integer partitions $r^*$, $\lambda$ and $\lambda^*$ appearing in \eqref{equ:Brualdi} are partitions of the number $8$ into $4$ parts with each part bounded by $4$. There are $8$ such integer partitions and in Table \ref{table:Kostka} we have included the relevant Kostka numbers $K_{\lambda,\mu}$ from \cite{kostka1882ueber}.

Using \eqref{equ:Brualdi} and the Kostka numbers from Table \ref{table:Kostka} we can compute the cardinalities $|\mathcal{A}\lb r,s\rb|$ for integer partitions $r=(r_1,r_2,r_3,r_4)$ and $s = (s_1,s_2,s_3,s_4)$ of $8$ into $4$ parts such that $r_1,s_1\leq 4$. Table \ref{table:ARS} contains the results of this elementary computation. 

Finally, we can classify all $(0,1)$-matrices $Z\in \M_4\lb \lset 0,1\rset\rb$ with $8$ zeros up to row and column permutations. For $Z_1,Z_2\in \mathcal{A}(r,s)$ we write $Z_1\sim Z_2$ if and only if $Z_1$ can be obtained from $Z_2$ by a sequence of row and column permutations. Then, it is straightforward albeit slightly tedious to compute the representatives of the equivalence classes in $\mathcal{A}(r,s)/\sim$. Table \ref{table:ARSSymme} contains a complete\footnote{We found it easiest to check completeness of this list by generating distinct elements in $\mathcal{A}(r,s)$ from the representatives given in Table \ref{table:ARSSymme}. It is not too hard to check that the numbers of Table \ref{table:ARS} can be obtained in this way showing completeness of Table \ref{table:ARSSymme}.} list of these representatives. We close this section with a lemma summarizing the previous discussion.

\begin{table}
\begin{tabular}{r!{\vline width 1.3pt}*{9}{c|}}
 \diagbox{$\lambda$}{$\mu$} & (4,4,0,0) & (4,3,1,0) & (4,2,2,0) & (4,2,1,1) & (3,3,2,0) & (3,3,1,1) & (3,2,2,1) & (2,2,2,2) \\
\Xhline{1.3pt}
(4,4,0,0) & 1 & 1 & 1 & 1 & 1 & 2 & 2 & 3  \\
\hline
(4,3,1,0) & 0 & 1 & 1 & 2 & 2 & 4 & 5 & 7   \\
\hline
(4,2,2,0) & 0 & 0 & 1 & 1 & 1 & 1 & 3 & 6   \\
\hline
(4,2,1,1) & 0 & 0 & 0 & 1 & 0 & 1 & 2 & 3  \\
\hline
(3,3,2,0) & 0 & 0 & 0 & 0 & 1 & 1 & 2 & 3   \\
\hline
(3,3,1,1) & 0 & 0 & 0 & 0 & 0 & 1 & 1 & 2   \\
\hline 
(3,2,2,1) & 0 & 0 & 0 & 0 & 0 & 0 & 1 & 3   \\
\hline
(2,2,2,2) & 0 & 0 & 0 & 0 & 0 & 0 & 0 & 1   \\
\end{tabular}
\caption{Table of Kostka numbers $K_{\lambda,\mu}$ from \cite{kostka1882ueber}}
\label{table:Kostka}
\end{table}

\begin{lem}[Classification of zero patterns with $8$ zeros]
Every $(0,1)$-matrices $Z\in \M_4\lb \lset 0,1\rset\rb$ with $8$ zeros is equivalent to a matrix from Table \ref{table:ARSSymme} by a sequence of row and column permutations. 
\label{lem:ClassZeroPatt}
\end{lem}

\begin{table}
\begin{tabular}{r!{\vline width 1.3pt}*{9}{c|}}
 \diagbox{r}{s}  & (4,4,0,0) & (4,3,1,0) & (4,2,2,0) & (4,2,1,1) & (3,3,2,0) & (3,3,1,1) & (3,2,2,1) & (2,2,2,2) \\
\Xhline{1.3pt}
(4,4,0,0) & 0 & 0 & 0 & 0 & 0 & 0 & 0 & 1  \\
\hline
(4,3,1,0) & 0 & 0 & 0 & 0 & 0 & 0 & 1 & 4   \\
\hline
(4,2,2,0) & 0 & 0 & 0 & 0 & 0 & 1 & 2 & 6  \\
\hline
(4,2,1,1) & 0 & 0 & 0 & 1 & 0 & 2 & 5 & 12  \\
\hline
(3,3,2,0) & 0 & 0 & 0 & 0 & 1 & 2 & 5 & 12   \\
\hline
(3,3,1,1) & 0 & 0 & 1 & 2 & 2 & 4 & 12 & 28   \\
\hline 
(3,2,2,1) & 0 & 1 & 2 & 5 & 5 & 12 & 24 & 48   \\
\hline
(2,2,2,2) & 1 & 4 & 6 & 12 & 12 & 28 & 48 & 90   \\
\end{tabular}
\caption{The cardinality $|\mathcal{A}(r,s)|$ computed via \eqref{equ:Brualdi} and Table \ref{table:Kostka}. }
\label{table:ARS}
\end{table}

\begin{sidewaystable}
\fontsize{1}{7}\selectfont
\begin{tabular}{r!{\vline width 1.3pt}*{9}{c|}}
 \diagbox{r}{s}  & (4,4,0,0) & (4,3,1,0) & (4,2,2,0) & (4,2,1,1) & (3,3,2,0) & (3,3,1,1) & (3,2,2,1) & (2,2,2,2) \\
\Xhline{1.3pt}
(4,4,0,0) \vspace{0.1cm}& 0 & 0 & 0 & 0 & 0 & 0 & 0 & $\begin{pmatrix} 1 & 1 & 1 & 1 \\ 1 & 1 & 1 & 1 \\ 0 & 0 & 0 & 0 \\ 0 & 0 & 0 & 0\end{pmatrix}_{\ref{lem:ZeroSubBoxZeroSubD}}$ \\
\hline
(4,3,1,0) \vspace{0.1cm}& 0 & 0 & 0 & 0 & 0 & 0 & $\begin{pmatrix} 1 & 1 & 1 & 1 \\ 1 & 1 & 1 & 0 \\ 1 & 0 & 0 & 0 \\ 0 & 0 & 0 & 0\end{pmatrix}_{\ref{lem:ZeroSubBoxZeroSubD}}$ & $\begin{pmatrix} 1 & 1 & 1 & 1 \\ 1 & 1 & 1 & 0 \\ 0 & 0 & 0 & 1 \\ 0 & 0 & 0 & 0\end{pmatrix}_{\ref{lem:ZeroSubBoxZeroSubD}}$   \\
\hline
(4,2,2,0) \vspace{0.1cm}& 0 & 0 & 0 & 0 & 0 & $\begin{pmatrix} 1 & 1 & 1 & 1 \\ 1 & 1 & 0 & 0 \\ 1 & 1 & 0 & 0 \\ 0 & 0 & 0 & 0\end{pmatrix}_{\ref{lem:ZeroSubBoxZeroSubD}}$ & $\begin{pmatrix} 1 & 1 & 1 & 1 \\ 1 & 1 & 0 & 0 \\ 1 & 0 & 1 & 0 \\ 0 & 0 & 0 & 0\end{pmatrix}_{\ref{lem:CasesCross}}$ & $\begin{pmatrix} 1 & 1 & 1 & 1 \\ 1 & 1 & 0 & 0 \\ 0 & 0 & 1 & 1 \\ 0 & 0 & 0 & 0\end{pmatrix}_{\ref{lem:cases1}}$   \\
\hline
(4,2,1,1) \vspace{0.1cm}& 0 & 0 & 0 & $\begin{pmatrix} 1 & 1 & 1 & 1 \\ 1 & 1 & 0 & 0 \\ 1 & 0 & 0 & 0 \\ 1 & 0 & 0 & 0\end{pmatrix}_{\ref{lem:ZeroSubBoxZeroSubD}}$ & 0 & $\begin{pmatrix} 1 & 1 & 1 & 1 \\ 1 & 1 & 0 & 0 \\ 1 & 0 & 0 & 0 \\ 0 & 1 & 0 & 0\end{pmatrix}_{\ref{lem:ZeroSubBoxZeroSubD}}$ & $\begin{pmatrix} 1 & 1 & 1 & 1 \\ 1 & 1 & 0 & 0 \\ 1 & 0 & 0 & 0 \\ 0 & 0 & 1 & 0\end{pmatrix}_{\ref{lem:CasesCross}},\begin{pmatrix} 1 & 1 & 1 & 1 \\ 0 & 1 & 1 & 0 \\ 1 & 0 & 0 & 0 \\ 1 & 0 & 0 & 0\end{pmatrix}_{\ref{lem:ZeroSubBoxZeroSubD}}$ & $\begin{pmatrix} 1 & 1 & 1 & 1 \\ 1 & 1 & 0 & 0 \\ 0 & 0 & 1 & 0 \\ 0 & 0 & 0 & 1\end{pmatrix}_{\ref{lem:cases3}}$  \\
\hline
(3,3,2,0) \vspace{0.1cm}& 0 & 0 & 0 & 0 & $\begin{pmatrix} 1 & 1 & 1 & 0 \\ 1 & 1 & 1 & 0 \\ 1 & 1 & 0 & 0 \\ 0 & 0 & 0 & 0\end{pmatrix}_{\ref{lem:ZeroRowColumn}}$ & $\begin{pmatrix} 1 & 1 & 0 & 1 \\ 1 & 1 & 1 & 0 \\ 1 & 1 & 0 & 0 \\ 0 & 0 & 0 & 0\end{pmatrix}_{\ref{lem:cases3}}$ & $\begin{pmatrix} 1 & 1 & 1 & 0 \\ 1 & 1 & 1 & 0 \\ 1 & 0 & 0 & 1 \\ 0 & 0 & 0 & 0\end{pmatrix}_{\ref{lem:cases3}},\begin{pmatrix} 1 & 1 & 1 & 0 \\ 1 & 0 & 1 & 1 \\ 1 & 1 & 0 & 0 \\ 0 & 0 & 0 & 0\end{pmatrix}_{\ref{lem:CasesCross}}$ & $\begin{pmatrix} 1 & 1 & 1 & 0 \\ 0 & 1 & 1 & 1 \\ 1 & 0 & 0 & 1 \\ 0 & 0 & 0 & 0\end{pmatrix}_{\ref{lem:6erLem}}$  \\
\hline
(3,3,1,1) \vspace{0.1cm}& 0 & 0 & $\cancel{\begin{pmatrix} 1 & 1 & 1 & 0 \\ 1 & 1 & 1 & 0 \\ 1 & 0 & 0 & 0 \\ 1 & 0 & 0 & 0\end{pmatrix}}$ & $\cancel{\begin{pmatrix} 1 & 1 & 1 & 0 \\ 1 & 1 & 0 & 1 \\ 1 & 0 & 0 & 0 \\ 1 & 0 & 0 & 0\end{pmatrix}}$ & $\cancel{\begin{pmatrix} 1 & 1 & 1 & 0 \\ 1 & 1 & 1 & 0 \\ 0 & 1 & 0 & 0 \\ 1 & 0 & 0 & 0\end{pmatrix}}$ & $\begin{pmatrix} 1 & 1 & 1 & 0 \\ 1 & 1 & 0 & 1 \\ 1 & 0 & 0 & 0 \\ 0 & 1 & 0 & 0\end{pmatrix}_{\ref{lem:cases1}}$ & \makecell{$\begin{pmatrix} 1 & 1 & 1 & 0 \\ 1 & 1 & 1 & 0 \\ 1 & 0 & 0 & 0 \\ 0 & 0 & 0 & 1\end{pmatrix}_{\ref{lem:cases1}},\begin{pmatrix} 1 & 1 & 1 & 0 \\ 1 & 0 & 1 & 1 \\ 1 & 0 & 0 & 0 \\ 0 & 1 & 0 & 0\end{pmatrix}_{\ref{lem:CasesCross}}$,\\ $\begin{pmatrix} 1 & 1 & 1 & 0 \\ 0 & 1 & 1 & 1 \\ 1 & 0 & 0 & 0 \\ 1 & 0 & 0 & 0\end{pmatrix}_{\ref{lem:ZeroSubBoxZeroSubD}}$} & \makecell{$\begin{pmatrix} 1 & 1 & 1 & 0 \\ 1 & 1 & 1 & 0 \\ 0 & 0 & 0 & 1 \\ 0 & 0 & 0 & 1\end{pmatrix}_{\ref{lem:ZeroSubBoxZeroSubD}}$, \\$\begin{pmatrix} 1 & 1 & 1 & 0 \\ 0 & 1 & 1 & 1 \\ 1 & 0 & 0 & 0 \\ 0 & 0 & 0 & 1\end{pmatrix}_{\ref{lem:6erLem}}$ }  \\
\hline 
(3,2,2,1) \vspace{0.1cm}& 0 & $\cancel{\begin{pmatrix} 1 & 1 & 1 & 0 \\ 1 & 1 & 0 & 0 \\ 1 & 1 & 0 & 0 \\ 1 & 0 & 0 & 0\end{pmatrix}}$ & $\cancel{\begin{pmatrix} 1 & 1 & 1 & 0 \\ 1 & 1 & 0 & 0 \\ 1 & 0 & 1 & 0 \\ 1 & 0 & 0 & 0\end{pmatrix}}$ & \makecell{$\cancel{\begin{pmatrix} 1 & 1 & 1 & 0 \\ 1 & 1 & 0 & 0 \\ 1 & 0 & 0 & 0 \\ 0 & 0 & 1 & 0\end{pmatrix}}$, \\$\cancel{\begin{pmatrix} 1 & 0 & 1 & 1 \\ 1 & 1 & 0 & 0 \\ 1 & 1 & 0 & 0 \\ 1 & 0 & 0 & 0\end{pmatrix}}$} & \makecell{$\cancel{\begin{pmatrix} 1 & 1 & 1 & 0 \\ 1 & 1 & 0 & 0 \\ 1 & 1 & 0 & 0 \\ 0 & 0 & 1 & 0\end{pmatrix}}$, \\$\cancel{\begin{pmatrix} 1 & 1 & 1 & 0 \\ 1 & 0 & 1 & 0 \\ 1 & 1 & 0 & 0 \\ 0 & 1 & 0 & 0\end{pmatrix}}$} & \makecell{$\cancel{\begin{pmatrix} 1 & 1 & 1 & 0 \\ 1 & 1 & 0 & 0 \\ 1 & 1 & 0 & 0 \\ 0 & 0 & 0 & 1\end{pmatrix}},\cancel{\begin{pmatrix} 1 & 1 & 1 & 0 \\ 1 & 0 & 0 & 1 \\ 1 & 1 & 0 & 0 \\ 0 & 1 & 0 & 0\end{pmatrix}}$,\\ $\cancel{\begin{pmatrix} 1 & 0 & 1 & 1 \\ 1 & 1 & 0 & 0 \\ 1 & 1 & 0 & 0 \\ 0 & 1 & 0 & 0\end{pmatrix}}$} & \makecell{$\begin{pmatrix} 1 & 1 & 1 & 0 \\ 1 & 1 & 0 & 0 \\ 1 & 0 & 1 & 0 \\ 0 & 0 & 0 & 1\end{pmatrix}_{\ref{lem:cases3}},\begin{pmatrix} 0 & 1 & 1 & 1 \\ 1 & 0 & 1 & 0 \\ 1 & 1 & 0 & 0 \\ 1 & 0 & 0 & 0\end{pmatrix}_{\ref{lem:6erLem}}$,\\ $\cancel{\begin{pmatrix} 1 & 1 & 1 & 0 \\ 0 & 0 & 1 & 1 \\ 1 & 1 & 0 & 0 \\ 1 & 0 & 0 & 0\end{pmatrix}},\begin{pmatrix} 1 & 0 & 1 & 1 \\ 1 & 0 & 1 & 0 \\ 1 & 1 & 0 & 0 \\ 0 & 1 & 0 & 0\end{pmatrix}_{\ref{lem:CasesCross}}$,\\$\cancel{\begin{pmatrix} 1 & 1 & 1 & 0 \\ 0 & 1 & 1 & 0 \\ 1 & 0 & 0 & 1 \\ 1 & 0 & 0 & 0\end{pmatrix}},\begin{pmatrix} 1 & 0 & 1 & 1 \\ 1 & 1 & 0 & 0 \\ 1 & 1 & 0 & 0 \\ 0 & 0 & 1 & 0\end{pmatrix}_{\ref{lem:6erLem}}$} & \makecell{$\begin{pmatrix} 1 & 1 & 1 & 0 \\ 1 & 1 & 0 & 0 \\ 0 & 0 & 1 & 1 \\ 0 & 0 & 0 & 1\end{pmatrix}_{\ref{lem:6erLem}}$,\\ $\begin{pmatrix} 1 & 1 & 1 & 0 \\ 0 & 0 & 1 & 1 \\ 0 & 1 & 0 & 1 \\ 1 & 0 & 0 & 0\end{pmatrix}_{\ref{lem:CasesCross}}$}   \\
\hline
(2,2,2,2) \vspace{0.1cm}& $\cancel{\begin{pmatrix} 1 & 1 & 0 & 0 \\ 1 & 1 & 0 & 0 \\ 1 & 1 & 0 & 0 \\ 1 & 1 & 0 & 0\end{pmatrix}}$ & $\cancel{\begin{pmatrix} 1 & 1 & 0 & 0 \\ 1 & 1 & 0 & 0 \\ 1 & 1 & 0 & 0 \\ 1 & 0 & 1 & 0\end{pmatrix}}$ & $\cancel{\begin{pmatrix} 1 & 1 & 0 & 0 \\ 1 & 1 & 0 & 0 \\ 1 & 0 & 1 & 0 \\ 1 & 0 & 1 & 0\end{pmatrix}}$ & $\cancel{\begin{pmatrix} 1 & 1 & 0 & 0 \\ 1 & 1 & 0 & 0 \\ 1 & 0 & 1 & 0 \\ 1 & 0 & 0 & 1\end{pmatrix}}$ & $\cancel{\begin{pmatrix} 1 & 0 & 1 & 0 \\ 1 & 1 & 0 & 0 \\ 1 & 1 & 0 & 0 \\ 0 & 1 & 1 & 0\end{pmatrix}}$ & $\cancel{\begin{pmatrix} 1 & 1 & 0 & 0 \\ 1 & 1 & 0 & 0 \\ 1 & 1 & 0 & 0 \\ 0 & 0 & 1 & 1\end{pmatrix}},\cancel{\begin{pmatrix} 1 & 0 & 1 & 0 \\ 1 & 1 & 0 & 0 \\ 1 & 1 & 0 & 0 \\ 0 & 1 & 0 & 1\end{pmatrix}}$ & $\cancel{\begin{pmatrix} 1 & 1 & 0 & 0 \\ 1 & 1 & 0 & 0 \\ 1 & 0 & 1 & 0 \\ 0 & 0 & 1 & 1\end{pmatrix}},\cancel{\begin{pmatrix} 1 & 0 & 0 & 1 \\ 1 & 0 & 1 & 0 \\ 1 & 1 & 0 & 0 \\ 0 & 1 & 1 & 0\end{pmatrix}}$ & \makecell{$\begin{pmatrix} 1 & 1 & 0 & 0 \\ 1 & 1 & 0 & 0 \\ 0 & 0 & 1 & 1 \\ 0 & 0 & 1 & 1\end{pmatrix}_{\ref{lem:6erLem}}$,\\ $\begin{pmatrix} 1 & 1 & 0 & 0 \\ 0 & 1 & 1 & 0 \\ 0 & 0 & 1 & 1 \\ 1 & 0 & 0 & 1\end{pmatrix}_{\ref{lem:cases1}}$ }  \\
\end{tabular}
\caption{Complete list of the representatives of equivalence classes of $(0,1)$-matrices in $\mathcal{A}(r,s)$ up to row and column permutation. Crossed out matrices arise as transpositions of others that are not crossed out. Every matrix, that is not crossed out, has an index refering to the lemma where its case is treated. }

\label{table:ARSSymme}
\hspace{-3cm}
\end{sidewaystable}

\subsection{Completeness of extremal rays}  

Assume that $\lb P,Q\rb\in \mathcal{S}\lb \CP_2\cap\coCP_2\rb$ generates an extremal ray and that 
\[
(P,Q)\notin\text{span}\lb\bigcup^3_{i=1} \text{orb}_2\lbr (P_i,Q_i) \rbr\rb.
\]
By Corollary \ref{cor:RankBound} we have
\[
|\mathcal{Z}\lb P,Q\rb|\geq 15,
\]
and therefore either $P$ or $Q$ contains at least $8$ zeros. We can assume without loss of generality that $P$ contains at least $8$ zeros (otherwise exchange the roles of $P$ and $Q$). This implies the existence of a zero pattern $Z'\subset \lset 1,2,3,4\rset^2$ such that $|Z'|=8$ and $Z'\subseteq\mathcal{Z}(P)$ (cf.~Definition \ref{defn:ZPattern}). 

We can identify the zero pattern $Z'\subset \lset 1,2,3,4\rset^2$ with a $(0,1)$-matrix $Z\in \M_4\lb \lset 0,1\rset\rb$ such that $Z_{ij}=0$ if and only if $(i,j)\in Z'$ for $i,j\in\lset 1,2,3,4\rset$. By Lemma \ref{lem:ClassZeroPatt} the matrix $Z$ is equivalent to a matrix $\tilde{Z}$ in Table \ref{table:ARSSymme} by a sequence of row and column permutations, and specifically we assume that $\tilde{Z}=U_{\sigma_1}Z U_{\sigma_2}$ for permutation matrices $U_{\sigma_1},U_{\sigma_2}\in\M_4$ corresponding to permutations $\sigma_1,\sigma_2\in S_4$. By Lemma \ref{lem:Permutations} the pair 
\[
\lb U_{\sigma_1}P U_{\sigma_2},U_{\sigma_1}Q U_{\sigma_2}\rb\in \mathcal{S}\lb \CP_2\cap\coCP_2\rb
\]
generates an extremal ray and by assumption it satisfies 
\[
\lb U_{\sigma_1}P U_{\sigma_2},U_{\sigma_1}Q U_{\sigma_2}\rb\notin\text{span}\lb\bigcup^3_{i=1} \text{orb}_2\lbr (P_i,Q_i) \rbr\rb.
\]
Moreover, we have $\tilde{Z}\subseteq\mathcal{Z}(U_{\sigma_1}P U_{\sigma_2})$. We will finish the proof by showing: 

\begin{lem}
Assume that $(P,Q)\in \mathcal{S}\lb \CP_2\cap\coCP_2\rb$ generates an extremal ray. If $Z\subseteq\mathcal{Z}(P)$ for any zero pattern $Z\in \M_4\lb \lset 0,1\rset\rb$ contained in Table \ref{table:ARSSymme}, then $(P,Q)$ is a box, a diagonal, or a cross (cf.~Theorem \ref{thm:ExtremeRaysN2}).
\label{lem:Final}
\end{lem}

\begin{proof}

Note that it is enough to show the statement either for a zero pattern $Z\in \M_4\lb \lset 0,1\rset\rb$ or its transpose $Z^T\in \M_4\lb \lset 0,1\rset\rb$. The zero patterns that are crossed out in Table \ref{table:ARSSymme} arise by transposing a zero pattern in the same table that is not crossed out. We will therefore focus on these remaining zero patterns. Any zero pattern in Table \ref{table:ARSSymme} contains a reference to one of the Lemmas \ref{lem:6erLem}, \ref{lem:ZeroSubBoxZeroSubD}, \ref{lem:ZeroRowColumn}, \ref{lem:cases1}, \ref{lem:CasesCross}, or \ref{lem:cases3} that we will prove in the next section. For each zero pattern the respective lemma shows the statement of Lemma \ref{lem:Final} directly. This finishes the proof. 

\end{proof}

\subsection{Finishing the proof of Lemma \ref{lem:Final}} 

For $(P,Q)\in \mathcal{S}\lb \CP_2\cap\coCP_2\rb$ we have $Q=KPK$, where $K\in\M_4$ is the unitary matrix from \eqref{equ:MatK}. The entries of this equation are
\begin{equation}
Q_{ij} = \frac{1}{4}\lb P_{ij} + \sum^4_{\substack{k,l=1\\k\neq i, l\neq j}} P_{kl} - \sum^4_{\substack{k=1\\k\neq i}} P_{kj} - \sum^4_{\substack{l=1\\l\neq j}} P_{il}\rb .
\label{equ:KPKQ}
\end{equation} 
To simplify the following discussion we will say that $(P,Q)\in \mathcal{S}\lb \CP_2\cap\coCP_2\rb$ \emph{contains a box} (or a diagonal, or a cross) if there exists a box (or a diagonal, or a cross) $(P_b,Q_b)\in \mathcal{S}\lb \CP_2\cap\coCP_2\rb$ and a $\beta>0$ such that $(P,Q)-\beta (P_b,Q_b)\in \mathcal{S}\lb \CP_2\cap\coCP_2\rb$. Note that $(P,Q)$ contains a box if and only if the zero pattern $\mathcal{Z}(P,Q)$ contains the zero pattern of a box (as sets), and the same holds for diagonals and crosses. Clearly, any $(P,Q)\in \mathcal{S}\lb \CP_2\cap\coCP_2\rb$ generating an extremal ray and containing a box (or a diagonal, or a cross) has to be a box (or a diagonal, or a cross). With this we can start to prove the following lemma.

\begin{lem}[6er pattern implies box]
Let $Z\in \M_4\lb\lset 0,1\rset\rb$ denote a zero pattern containing either pattern
\[
\begin{pmatrix}
0 & 0 & 1 & 1 \\ 0 & 0 & 1 & 1 \\ 0 & 0 & 1 & 1 \\ 1 & 1 & 1 & 1
\end{pmatrix} \quad \text{ or }\quad \begin{pmatrix}
0 & 0 & 0 & 1 \\ 0 & 0 & 0 & 1 \\ 1 & 1 & 1 & 1 \\ 1 & 1 & 1 & 1
\end{pmatrix}
\]
up to row and column permutations. Any $(P,Q)\in \mathcal{S}\lb \CP_2\cap\coCP_2\rb$ generating an extremal ray and satisfying $Z\subseteq\mathcal{Z}(P)$ or $Z\subseteq\mathcal{Z}(Q)$ is a box. 
\label{lem:6erLem}
\end{lem} 

\begin{proof}
We can assume without loss of generality that $P_{ij}=0$ for any $i\in \lset 1,2,3\rset$ and any $j\in\lset 1,2\rset$. Applying Lemma \ref{lem:boxRule} three times for $P$ shows that $Q_{ij}=0$ for any $i\in\lset 1,2,3,4\rset$ and any $j\in\lset 3,4\rset$. Applying Lemma \ref{lem:boxRule} again for $Q$ shows that $P_{ij}=0$ for any $i\in \lset 1,2,3,4\rset$ and any $j\in\lset 1,2\rset$. Now let 
\[
P=\begin{pmatrix} 0 & 0 & p_1 & p_2 \\ 0 & 0 & p_3 & p_4 \\ 0 & 0 & p_5 & p_6 \\ 0 & 0 & p_7 & p_8\end{pmatrix} \quad\text{ and }\quad Q=\begin{pmatrix} q_1 & q_2 & 0 & 0  \\ q_3 & q_4 & 0 & 0 \\ q_5 & q_6 & 0 & 0 \\ q_7 & q_8 & 0 & 0\end{pmatrix}.
\]
Since $Q=KPK$ we find by \eqref{equ:KPKQ} that 
\[
q_1=\frac{1}{4}(p_3+p_4+p_5+p_6+p_7+p_8-p_1-p_2) = q_2.
\]
Similarly, we find that $q_3=q_4$, $q_5=q_6$, $q_7=q_8$, and using that $P=KQK$ we find that $p_1=p_2$, $p_3=p_4$, $p_5=p_6$ and $p_7=p_8$ by the same reasoning. In summary, we have
\[
P=\begin{pmatrix} 0 & 0 & p_1 & p_1 \\ 0 & 0 & p_3 & p_3 \\ 0 & 0 & p_5 & p_5 \\ 0 & 0 & p_7 & p_7\end{pmatrix} \quad\text{ and }\quad Q=\begin{pmatrix} q_1 & q_1 & 0 & 0  \\ q_3 & q_3 & 0 & 0 \\ q_5 & q_5 & 0 & 0 \\ q_7 & q_7 & 0 & 0\end{pmatrix}.
\]
By assumption $(P,Q)$ generates an extremal ray, and we can assume for contradiction that $(P,Q)$ is not a box. Then, we must have $0\in\lset p_1,p_3,q_5,q_7\rset$, $0\in \lset q_1,q_3,p_5,p_7\rset$, $0\in \lset p_3,p_7,q_1,q_5\rset$ and $0\in\lset p_1,p_5,q_3,q_7\rset$ since otherwise $(P,Q)$ would contain a box. It follows that either $P$ or $Q$ must contain two distinct zero rows. Since $(P,Q)\neq (0,0)$, we find that $(P,Q)$ must be a box. 
\end{proof}

\begin{lem}
Let $Z\in \M_4\lb\lset 0,1\rset\rb$ denote a zero pattern containing the pattern 
\[
\begin{pmatrix}
0 & 1 & 1 & 1 \\ 1 & 0 & 1 & 1 \\ 1 & 1 & 0 & 0 \\ 1 & 1 & 0 & 0
\end{pmatrix}
\]
up to row and column permutations. Any $(P,Q)\in \mathcal{S}\lb \CP_2\cap\coCP_2\rb$ generating an extremal ray and satisfying $Z\subseteq\mathcal{Z}(P)$ or $Z\subseteq\mathcal{Z}(Q)$ is a box or a diagonal. 
\label{lem:ZeroSubBoxZeroSubD}
\end{lem}

\begin{proof}
Consider $(P,Q)\in \mathcal{S}\lb \CP_2\cap\coCP_2\rb$ generating an extremal ray and without loss of generality we assume $Z\subseteq\mathcal{Z}(P)$. Applying Lemma \ref{lem:boxRule} once and Lemma \ref{lem:diagRule} twice (for the two zero diagonals crossing in the lower right $2\times 2$ corner) shows that $Q_{ij}=0$ for any $i,j\in \lset 1,2\rset$ and for any $i,j\in \lset 3,4\rset$. By applying Lemma \ref{lem:boxRule} again for $Q$ we find that 
\[
P=\begin{pmatrix} 0 & 0 & p_1 & p_2 \\ 0 & 0 & p_3 & p_4 \\ p_5 & p_6 & 0 & 0 \\ p_7 & p_8 & 0 & 0\end{pmatrix} \quad\text{ and }\quad Q=\begin{pmatrix} 0 & 0 & q_1 & q_2  \\ 0 & 0 & q_3 & q_4 \\ q_5 & q_6 & 0 & 0  \\ q_7 & q_8 & 0 & 0 \end{pmatrix}
\]
with non-negative entries. Assume that $(P,Q)$ is not a diagonal. By extremality of $(P,Q)$ it cannot contain any diagonal, and therefore we find the following inclusions:
\begin{align*}
0&\in \lset p_2,p_3,p_6,p_7,q_2,q_3,q_6,q_7\rset \\
0&\in \lset p_1,p_4,p_6,p_7,q_1,q_4,q_6,q_7\rset \\
0&\in \lset p_1,p_4,p_5,p_8,q_1,q_4,q_5,q_8\rset \\
0&\in \lset p_2,p_3,p_5,p_8,q_2,q_3,q_5,q_8\rset
\end{align*}
It is easy to see that these conditions require two zeros among $p_1,p_2,p_3,p_4,q_1,q_2,q_3,q_4$ or among $p_5,p_6,p_7,p_8,q_5,q_6,q_7,q_8$ and it is not sufficient to have both among $p_1,p_4,q_1,q_4$ or both among $p_2,p_3,q_2,q_3$ or both among $p_5,p_8,q_5,q_8$ or both among $p_6,p_7,q_6,q_7$. By Lemma \ref{lem:6erLem} we are done in the case where these two zeros lie in either $P$ or $Q$ together. Therefore, we assume that one of the two zeros lies in $P$ and the other in $Q$. 

After suitable row and column permutations we can assume that $p_1=q_3=0$. Since $Q=KPK$ we find by \eqref{equ:KPKQ} that 
\begin{align*}
p_3-p_4+p_2+p_5+p_6+p_7+p_8 &= 0 \\
-p_3-p_4+p_2+p_5-p_6+p_7-p_8 &= 0 \\
-p_3-p_4+p_2-p_5+p_6-p_7+p_8 &= 0.
\end{align*} 
Substracting the second and third equations from the first, and using that all variables are non-negative real numbers, shows that $p_3=p_5=p_6=p_7=p_8=0$. Finally, applying Lemma \ref{lem:boxRule} multiple times shows that $P=Q=0$ contradicting the assumption that $(P,Q)$ generates an extremal ray. This finishes the proof.  

\end{proof}

\begin{lem}[Zero row and column]
Let $Z\in \M_4\lb\lset 0,1\rset\rb$ denote a zero pattern containing the pattern 
\[
\begin{pmatrix}
0 & 1 & 1 & 1 \\ 0 & 1 & 1 & 1 \\ 0 & 0 & 1 & 1 \\ 0 & 0 & 0 & 0
\end{pmatrix}
\]
up to row and column permutations. Any $(P,Q)\in \mathcal{S}\lb \CP_2\cap\coCP_2\rb$ generating an extremal ray and satisfying $Z\subseteq\mathcal{Z}(P)$ or $Z\subseteq\mathcal{Z}(Q)$ is a box. 
\label{lem:ZeroRowColumn}
\end{lem}

\begin{proof}
Consider $(P,Q)\in \mathcal{S}\lb \CP_2\cap\coCP_2\rb$ generating an extremal ray and without loss of generality we assume $Z\subseteq\mathcal{Z}(P)$. Using Lemma \ref{lem:boxRule} we find that 
\[
P=\begin{pmatrix} 0 & p_1 & p_2 & p_3 \\ 0 & p_4 & p_5 & p_6 \\ 0 & 0 & p_7 & p_8 \\ 0 & 0 & 0 & 0\end{pmatrix} \quad\text{ and }\quad Q=\begin{pmatrix} q_1 & q_2 & 0 & 0  \\ q_3 & q_4 & 0 & 0 \\ q_5 & q_6 & q_7 & q_8  \\ q_9 & q_{10} & q_{11} & q_{12} \end{pmatrix}
\] 
with $p_i\geq 0$ and $q_j\geq 0$ for any $i\in\lset 1,\ldots ,8 \rset$ and any $j\in\lset 1,\ldots ,12 \rset$. Assume for contradiction that $(P,Q)$ is not a box. Then, by Lemma \ref{lem:6erLem} we can infer that $p_1,p_4,p_7,p_8>0$. Since $Q=KPK$ we can use \eqref{equ:KPKQ} to find that
\[
q_9 = \frac{1}{4}\sum^{8}_{i=1} p_i >0
\]
and 
\begin{align*}
q_5 &= \frac{1}{4}(-p_7-p_8 + p_1 + p_4 + p_2 + p_3 + p_5 + p_6) \\
q_6 &= \frac{1}{4}(-p_7-p_8 - p_1 - p_4 + p_2 + p_3 + p_5 + p_6) . \\
\end{align*}
Taking the difference of the last two equations shows that 
\[
q_5-q_6 = \frac{1}{2}(p_1+p_4)>0.
\]
Similarly, we find that 
\[
q_{10}-q_6 = \frac{1}{2}(p_7+p_8)>0. 
\]
The previous equations imply in particular that $q_5,q_{10}>0$. Since $(P,Q)$ generates an extremal ray and is not a box, we must have $0\in\lset p_2,p_3,p_5,p_6,q_6\rset$. 

Assume first that $0\in\lset p_2,p_3,p_5,p_6\rset$. Without loss of generality we can assume further that $p_5=0$ (otherwise we can permute the first two rows and last two columns). By Lemma \ref{lem:boxRule} and Lemma \ref{lem:diagRule} this implies that $q_2=q_6=q_8=q_1=q_{12}=0$. By Lemma \ref{lem:6erLem} we find that $p_6,p_2,q_3,q_4,q_7,q_{11}>0$. Finally, $p_3>0$ since otherwise $0=q_9>0$ by Lemma \ref{lem:diagRule}. The previous argument shows that all remaining non-zero entries in $P$ and $Q$ must be non-zero. However, $(P,Q)$ still contains a box $p_1,p_3,p_4,p_6$ and $q_5,q_7,q_9,q_{11}$ without any zeros. Therefore $(P,Q)$ cannot generate an extremal ray contradicting the assumption. 

Finally, we can assume that $q_6=0$. By the previous paragraph we know that every entry of $P$ must be non-zero. Since $(P,Q)$ generates an extremal ray we conclude that $0\in \lset q_1,q_2\rset$, $0\in \lset q_8,q_{12}\rset$, $0\in \lset q_7,q_{11}\rset$, and $0\in \lset q_3,q_{4}\rset$. By Lemma \ref{lem:6erLem} we are done when two of these zeros are in the same row or the same column. If none of these zeros are in the same row or column, then they must lie on the same diagonal. Then we can conclude by Lemma \ref{lem:diagRule} that $0\in\lset p_1,p_4,p_7,p_8\rset$ leading to a contradiction.

\end{proof}

\begin{lem}
Let $Z\in \M_4\lb\lset 0,1\rset\rb$ denote a zero pattern containing either one of the following patterns up to permutations of rows and columns: 
\begin{enumerate}
\item \[
\begin{pmatrix}
1 & 1 & 0 & 0 \\ 0 & 1 & 1 & 0 \\ 0 & 0 & 1 & 1 \\ 1 & 0 & 0 & 1
\end{pmatrix}
\]
\item \[
\begin{pmatrix}
1 & 1 & 1 & 0 \\ 1 & 1 & 0 & 1 \\ 1 & 0 & 0 & 0 \\ 0 & 1 & 0 & 0
\end{pmatrix}
\]
\item \[
\begin{pmatrix}
1 & 1 & 1 & 1 \\ 1 & 1 & 0 & 0 \\ 0 & 0 & 1 & 1 \\ 0 & 0 & 0 & 0
\end{pmatrix}
\]
\item \[
\begin{pmatrix}
1 & 1 & 1 & 0 \\ 1 & 1 & 1 & 0 \\ 1 & 0 & 0 & 0 \\ 0 & 0 & 0 & 1
\end{pmatrix}
\]
\end{enumerate}
Any $(P,Q)\in \mathcal{S}\lb \CP_2\cap\coCP_2\rb$ generating an extremal ray and satisfying $Z\subseteq\mathcal{Z}(P)$ or $Z\subseteq\mathcal{Z}(Q)$ is a box or a diagonal. 
\label{lem:cases1}
\end{lem}

\begin{proof} Assume that $(P,Q)\in \mathcal{S}\lb \CP_2\cap\coCP_2\rb$ is generating an extremal ray and satisfying $Z\subseteq\mathcal{Z}(P)$. In the following we will go through the cases of the lemma one-by-one every time assuming that $Z$ contains the corresponding zero pattern. 
\begin{enumerate}
\item Using Lemma \ref{lem:diagRule} twice we find that 
\[
P=\begin{pmatrix} p_1 & p_2 & 0 & 0 \\ 0 & p_3 & p_4 & 0 \\ 0 & 0 & p_5 & p_6 \\ p_7 & 0 & 0 & p_8\end{pmatrix} \quad\text{ and }\quad Q=\begin{pmatrix} q_1 & q_2 & 0 & 0  \\ 0 & q_3 & q_4 & 0 \\ 0 & 0 & q_5 & q_6  \\ q_7 & 0 & 0 & q_8 \end{pmatrix}
\]
where $p_i,q_j\geq 0$ for any $i,j\in\lset 1,\ldots , 8\rset$. It is easy to check that $p_i=0$ for any $i\in\lset 1,\ldots , 8\rset$ would create a zero pattern as in Lemma \ref{lem:ZeroSubBoxZeroSubD} and therefore $(P,Q)$ is either a box or a diagonal (in fact it has to be a diagonal in this case). The same argument works if $q_i=0$ for some $i\in\lset 1,\ldots , 8\rset$. Assuming that $(P,Q)$ is neither a diagonal or a box therefore implies that $p_i\neq 0$ and $q_j\neq 0$ for any $i,j\in\lset 1,\ldots, 8\rset$. However, then $(P,Q)$ would contain a diagonal (e.g.~the main diagonal) without any zeros, and would not generate an extremal ray.  
\item Using Lemma \ref{lem:diagRule} and Lemma \ref{lem:boxRule} we find that 
\[
P=\begin{pmatrix} p_1 & p_2 & p_3 & 0 \\ p_4 & p_5 & 0 & p_6 \\ p_7 & 0 & 0 & 0 \\ 0 & p_8 & 0 & 0\end{pmatrix} \quad\text{ and }\quad Q=\begin{pmatrix} 0 & 0 & q_1 & 0  \\ 0 & 0 & 0 & q_2 \\ q_3 & 0 & q_4 & q_5  \\ 0 & q_6 & q_7 & q_8 \end{pmatrix}
\] 
with $p_i\geq 0$ and $q_j\geq 0$ for any $i,j\in\lset 1,\ldots , 8\rset$. Note that by Lemma \ref{lem:6erLem} we are done if $0\in\lset p_3,p_6,p_7,p_8,q_1,q_2,q_3,q_6\rset$. But if this is not the case, then $(P,Q)$ contains a diagonal. Since $(P,Q)$ generates an extremal ray, it has to be equal to this diagonal.    
\item Using Lemma \ref{lem:boxRule} twice we find that  
\[
P=\begin{pmatrix} p_1 & p_2 & p_3 & p_4 \\ p_5 & p_6 & 0 & 0 \\ 0 & 0 & p_7 & p_8 \\ 0 & 0 & 0 & 0\end{pmatrix} \quad\text{ and }\quad Q=\begin{pmatrix} 0 & 0 & 0 & 0  \\ q_1 & q_2 & 0 & 0 \\ 0 & 0 & q_3 & q_4  \\ q_5 & q_6 & q_7 & q_8 \end{pmatrix}
\]
with $p_i\geq 0$ and $q_j\geq 0$ for any $i,j\in\lset 1,\ldots , 8\rset$. Note that by Lemma \ref{lem:6erLem} we are done if $0\in\lset p_7,p_8,q_1,q_2\rset$. We are also done by Lemma \ref{lem:ZeroSubBoxZeroSubD} if $0\in \lset p_3,p_4,q_5,q_6\rset$. If neither of the previous apply, we conclude by extremality that $(P,Q)$ is a box. 
\item Using Lemma \ref{lem:boxRule} we find that 
\[
P=\begin{pmatrix} p_1 & p_2 & p_3 & 0 \\ p_4 & p_5 & p_6 & 0 \\ p_7 & 0 & 0 & 0 \\ 0 & 0 & 0 & p_8\end{pmatrix} \quad\text{ and }\quad Q=\begin{pmatrix} 0 & q_1 & q_2 & 0  \\ 0 & q_3 & q_4 & 0 \\ q_5 & q_6 & q_7 & q_8  \\ q_9 & q_{10} & q_{11} & q_{12} \end{pmatrix}
\] 
with $p_i\geq 0$ for any $i\in\lset 1,\ldots , 8\rset$ and $q_j\geq 0$ for any $j\in\lset 1,\ldots , 12\rset$. Note that by Lemma \ref{lem:6erLem} we are done if $0\in \lset p_7,p_8\rset$ and by Lemma \ref{lem:ZeroSubBoxZeroSubD} we are done if $0\in\lset p_1,p_4\rset$. 

Consider first the case where $0\in\lset p_2,p_3,p_5,p_6\rset$. By permuting rows and columns we can assume that $p_2=0$. Then, we can use Lemma \ref{lem:boxRule} again to conclude that $q_4=q_9=q_{11}=0$ and Lemma \ref{lem:diagRule} to conclude that $q_1=q_7=0$. Now we would be done by Lemma \ref{lem:6erLem} if $0\in\lset p_3,p_5,q_2,q_3,q_5\rset$. By extremality either $(P,Q)$ is a diagonal, or $q_{12}=0$. But the latter would imply that $p_4=0$ and we would be done by Lemma \ref{lem:ZeroSubBoxZeroSubD}. 

For contradiction, we assume that $(P,Q)$ is neither a box nor a diagonal. Then, by the previous paragraph we can assume that $p_i>0$ for any $i\in\lset 1,\ldots , 8\rset$. Since $Q=KPK$ we can use \eqref{equ:KPKQ} to conclude
\[
q_{12} = \frac{1}{4}\sum^8_{i=1} p_i >0
\] 
and 
\begin{align*}
p_8 + p_1 + p_2 + p_3 &= p_7 + p_4 + p_5 + p_6 \\
p_8 + p_4 + p_5 + p_6 &= p_7 + p_1 + p_2 + p_3 .
\end{align*}
From the previous equations we conclude that $p_8 = p_7$. Therefore, we find that 
\begin{align*}
q_6 &= \frac{1}{4}(p_3 + p_6 - p_2 - p_5 + p_1 + p_4) = q_{10} \\
q_7 &= \frac{1}{4}(- p_3 - p_6 + p_2 + p_5 + p_1 + p_4) = q_{11}, 
\end{align*}
which also implies that $q_6+q_7 = \frac{1}{2}(p_1+p_4)>0$. We also find that 
\begin{align*}
q_8 &= \frac{1}{4}(-p_7-p_8+ p_1 + p_2 + p_3 + p_4 + p_5 + p_6) \\
q_9 &= \frac{1}{4}(-p_7-p_8 - p_1 + p_2 + p_3 - p_4 + p_5 + p_6).
\end{align*}
Taking the difference of the last two equations implies that 
\[
q_8 - q_9 = \frac{1}{2}(p_1+p_4)>0,
\]
which implies $q_8>0$. Summarizing the previous discussion we have
\[
Q=\begin{pmatrix} 0 & q_1 & q_2 & 0  \\ 0 & q_3 & q_4 & 0 \\ q_5 & q_6 & q_7 & q_8  \\ q_9 & q_{6} & q_{7} & q_{12} \end{pmatrix}
\] 
where $q_8,q_{12}>0$, and $q_6>0$ or $q_7>0$. Therefore, $(P,Q)$ contains a box contradicting extremality. 

\end{enumerate}
\end{proof}

\begin{lem}
Let $Z\in \M_4\lb\lset 0,1\rset\rb$ denote a zero pattern containing either one of the following patterns up to permutations of rows and columns: 
\begin{enumerate}
\item \[
\begin{pmatrix}
1 & 1 & 1 & 0 \\ 0 & 0 & 1 & 1 \\ 0 & 1 & 0 & 1 \\ 1 & 0 & 0 & 0
\end{pmatrix} \sim \begin{pmatrix}
1 & 1 & 1 & 0 \\ 0 & 1 & 0 & 1 \\ 0 & 0 & 1 & 1 \\ 1 & 0 & 0 & 0
\end{pmatrix}
\]
\item \[
\begin{pmatrix}
1 & 0 & 1 & 1 \\ 1 & 0 & 1 & 0 \\ 1 & 1 & 0 & 0 \\ 0 & 1 & 0 & 0
\end{pmatrix} \sim \begin{pmatrix}
1 & 1 & 0 & 1 \\ 1 & 1 & 0 & 0 \\ 1 & 0 & 1 & 0 \\ 0 & 0 & 1 & 0
\end{pmatrix}
\]
\item \[
\begin{pmatrix}
1 & 1 & 1 & 0 \\ 1 & 0 & 1 & 1 \\ 1 & 0 & 0 & 0 \\ 0 & 1 & 0 & 0
\end{pmatrix} \sim \begin{pmatrix}
1 & 1 & 1 & 0 \\ 1 & 1 & 0 & 1 \\ 0 & 0 & 1 & 0 \\ 1 & 0 & 0 & 0
\end{pmatrix}
\]
\item \[
\begin{pmatrix}
1 & 1 & 1 & 0 \\ 1 & 0 & 1 & 1 \\ 1 & 1 & 0 & 0 \\ 0 & 0 & 0 & 0
\end{pmatrix} \sim \begin{pmatrix}
1 & 1 & 1 & 0 \\ 1 & 1 & 0 & 1 \\ 1 & 0 & 1 & 0 \\ 0 & 0 & 0 & 0
\end{pmatrix}
\]
\item \[
\begin{pmatrix}
1 & 1 & 1 & 1 \\ 1 & 1 & 0 & 0 \\ 1 & 0 & 0 & 0 \\ 0 & 0 & 1 & 0
\end{pmatrix}
\]
\item \[
\begin{pmatrix}
1 & 1 & 1 & 1 \\ 1 & 1 & 0 & 0 \\ 1 & 0 & 1 & 0 \\ 0 & 0 & 0 & 0
\end{pmatrix}
\]
\end{enumerate}
Any $(P,Q)\in \mathcal{S}\lb \CP_2\cap\coCP_2\rb$ generating an extremal ray and satisfying $Z\subseteq\mathcal{Z}(P)$ or $Z\subseteq\mathcal{Z}(Q)$ is a cross, a diagonal or a box. 
\label{lem:CasesCross}
\end{lem}

\begin{proof} Assume that $(P,Q)\in \mathcal{S}\lb \CP_2\cap\coCP_2\rb$ is generating an extremal ray and satisfying $Z\subseteq\mathcal{Z}(P)$. In the following we will go through the cases of the lemma one-by-one every time assuming that $Z$ contains the corresponding zero pattern. 
\begin{enumerate}
\item Using Lemma \ref{lem:crossRule} and Lemma \ref{lem:diagRule} we find that 
\[
P=\begin{pmatrix} p_1 & p_2 & p_3 & 0 \\ 0 & p_4 & 0 & p_5 \\ 0 & 0 & p_6 & p_7 \\ p_8 & 0 & 0 & 0\end{pmatrix} \quad\text{ and }\quad Q=\begin{pmatrix} 0 & q_1 & q_2 & 0  \\ 0 & q_3 & 0 & q_4 \\ 0 & 0 & q_5 & q_6  \\ q_7 & 0 & 0 & q_8 \end{pmatrix}
\]
where $p_i,q_j\geq 0$ for any $i,j\in\lset 1,\ldots , 8\rset$. By Lemma \ref{lem:ZeroSubBoxZeroSubD} we would be done if 
\[
0\in\lset p_2,p_3,p_4,p_6,p_8,q_3,q_4,q_5,q_6,q_7\rset .
\]
By extremality either $(P,Q)$ is a cross, or $0\in \lset p_1,q_8\rset$. Note that $p_1=0$ if and only if $q_8=0$ by Lemma \ref{lem:diagRule}. Assuming that $p_1=q_8=0$, we can use Lemma \ref{lem:ZeroSubBoxZeroSubD} to conclude that either $(P,Q)$ is a box or a diagonal (here it would be a diagonal actually), or $0\notin \lset p_5,p_7, q_1,q_2\rset$. Assuming that $(P,Q)$ is neither a box nor a diagonal, then all remaining entries of $P$ and $Q$ would have to be non-zero. Since, $(P,Q)$ still contains a diagonal (since $p_3,p_4,p_6,p_8>0$ and $q_3,q_2,q_6,q_7>0$) this would contradict extremality. 
\item Using Lemma \ref{lem:crossRule} and Lemma \ref{lem:boxRule} we find that 
\[
P=\begin{pmatrix} p_1 & p_2 & 0 & p_3 \\ p_4 & p_5 & 0 & 0 \\ p_6 & 0 & p_7 & 0 \\ 0 & 0 & p_8 & 0\end{pmatrix} \quad\text{ and }\quad Q=\begin{pmatrix} 0 & 0 & 0 & q_1  \\ 0 & q_2 & 0 & 0 \\ q_3 & 0 & q_4 & q_5  \\ 0 & q_6 & q_7 & q_8 \end{pmatrix}
\]
where $p_i,q_j\geq 0$ for any $i,j\in\lset 1,\ldots , 8\rset$. By Lemma \ref{lem:ZeroSubBoxZeroSubD} we would be done if 
\[
0\in\lset p_1,p_3,p_4,p_5,p_6,q_1,q_4,q_5,q_6\rset .
\]
By Lemma \ref{lem:6erLem} we would also be done if $q_2=0$. By extremality either $(P,Q)$ is a cross, or $0\in \lset p_7,q_6\rset$. Note that $p_7=0$ if and only if $q_6=0$ by Lemma \ref{lem:boxRule}. Assuming that $p_7=q_6=0$, we can use Lemma \ref{lem:ZeroSubBoxZeroSubD} to conclude that either $(P,Q)$ is a box or a diagonal (here it would be a box actually), or $0\notin \lset p_2,q_5\rset$. If the latter were the case, then $(P,Q)$ would contain a box since $p_1,p_2,p_4,p_5>0$ and $q_4,q_5,q_7,q_8>0$. By extremality $(P,Q)$ is a box in this case. 
\item Using Lemma \ref{lem:crossRule} and Lemma \ref{lem:boxRule} we find that 
\[
P=\begin{pmatrix} p_1 & p_2 & p_3 & 0 \\ p_4 & p_5 & 0 & p_6 \\ 0 & 0 & p_7 & 0 \\ p_8 & 0 & 0 & 0\end{pmatrix} \quad\text{ and }\quad Q=\begin{pmatrix} 0 & q_1 & 0 & 0  \\ 0 & q_2 & 0 & q_3 \\ 0 & 0 & q_4 & q_5  \\ q_6 & 0 & q_7 & q_8 \end{pmatrix}
\]
where $p_i,q_j\geq 0$ for any $i,j\in\lset 1,\ldots , 8\rset$. By Lemma \ref{lem:ZeroSubBoxZeroSubD} we would be done if 
\[
0\in\lset p_1,p_2,p_5,p_8,q_3,q_4,q_5,q_6,q_8\rset .
\]
By Lemma \ref{lem:6erLem} we would also be done if $p_7=0$. By extremality either $(P,Q)$ is a cross, or $0\in \lset p_3,q_2\rset$. Note that $p_3=0$ if and only if $q_2=0$ by Lemma \ref{lem:boxRule}. Assuming that $p_3=q_2=0$, we can use Lemma \ref{lem:ZeroSubBoxZeroSubD} to conclude that either $(P,Q)$ is a box or a diagonal (here it would be a box actually), or $0\notin \lset p_4,q_7\rset$. If the latter were the case, then $(P,Q)$ would contain a box since $p_1,p_2,p_4,p_5>0$ and $q_4,q_5,q_7,q_8>0$. By extremality $(P,Q)$ is a box in this case. 
\item Using Lemma \ref{lem:crossRule}, Lemma \ref{lem:boxRule}, and Lemma \ref{lem:diagRule} we find that 
\[
P=\begin{pmatrix} p_1 & p_2 & p_3 & 0 \\ p_4 & p_5 & 0 & p_6 \\ p_7 & 0 & p_8 & 0 \\ 0 & 0 & 0 & 0\end{pmatrix} \quad\text{ and }\quad Q=\begin{pmatrix} 0 & q_1 & 0 & 0  \\ 0 & q_2 & 0 & q_3 \\ 0 & 0 & q_4 & q_5  \\ 0 & q_6 & q_7 & q_8 \end{pmatrix}
\]
where $p_i,q_j\geq 0$ for any $i,j\in\lset 1,\ldots , 8\rset$. By Lemma \ref{lem:ZeroSubBoxZeroSubD} we would be done if $0\in\lset p_1,p_4,q_6,q_8\rset$, and by Lemma \ref{lem:6erLem} we would be done if $0\in\lset p_2,p_7,p_8,q_4,q_7\rset$. Moreover, by Lemma \ref{lem:ZeroRowColumn} we would be done if $0\in\lset p_4,q_1\rset$. Assuming that all these variables are non-zero, $(P,Q)$ then contains a cross since $p_1,p_2,p_4,p_6,p_7,p_8>0$ and $q_1,q_3,q_4,q_6,q_7,q_8>0$. By extremality $(P,Q)$ has to be a cross in this case.
\item Using Lemma \ref{lem:crossRule} and Lemma \ref{lem:boxRule} twice we find that 
\[
P=\begin{pmatrix} p_1 & p_2 & p_3 & p_4 \\ p_5 & p_6 & 0 & 0 \\ p_7 & 0 & 0 & 0 \\ 0 & 0 & p_8 & 0\end{pmatrix} \quad\text{ and }\quad Q=\begin{pmatrix} 0 & 0 & 0 & q_1  \\ 0 & q_2 & 0 & q_3 \\ q_4 & 0 & q_5 & q_6  \\ 0 & 0 & q_7 & q_8 \end{pmatrix}
\]
where $p_i,q_j\geq 0$ for any $i,j\in\lset 1,\ldots , 8\rset$. By Lemma \ref{lem:ZeroSubBoxZeroSubD} we would be done if $0\in \lset p_1,p_2,q_6,q_8\rset$ and by Lemma \ref{lem:6erLem} we would be done if $0\in\lset p_6,p_7,p_8,q_2,q_4,q_7\rset$. By extremality either $(P,Q)$ is a cross, or $0\in \lset p_3,q_3\rset$. Note that $p_3=0$ if and only if $q_3=0$ by Lemma \ref{lem:diagRule}. Assuming that $p_3=q_3=0$, we can use Lemma \ref{lem:ZeroSubBoxZeroSubD} to conclude that either $(P,Q)$ is a box or a diagonal or $0\notin \lset p_5,q_5\rset$. If the latter were the case, then $(P,Q)$ would contain a box since $p_1,p_2,p_5,p_6>0$ and $q_5,q_6,q_7,q_8>0$. By extremality $(P,Q)$ would then be a box. 
\item Using Lemma \ref{lem:crossRule} and Lemma \ref{lem:boxRule} twice we find that 
\[
P=\begin{pmatrix} p_1 & p_2 & p_3 & p_4 \\ p_5 & p_6 & 0 & 0 \\ p_7 & 0 & p_8 & 0 \\ 0 & 0 & 0 & 0\end{pmatrix} \quad\text{ and }\quad Q=\begin{pmatrix} 0 & 0 & 0 & q_1  \\ 0 & q_2 & 0 & q_3 \\ 0 & 0 & q_4 & q_5  \\ 0 & q_6 & q_7 & q_8 \end{pmatrix}
\]
where $p_i,q_j\geq 0$ for any $i,j\in\lset 1,\ldots , 8\rset$. By Lemma \ref{lem:ZeroSubBoxZeroSubD} we would be done if $0\in\lset p_1,q_8\rset$, and by Lemma \ref{lem:6erLem} we would be done if $0\in\lset p_5,p_6,p_7,p_8,q_2,q_4,q_6,q_7\rset$. Moreover, by Lemma \ref{lem:ZeroRowColumn} we would be done if $0\in\lset p_4,q_1\rset$. Assuming that all these variables are non-zero $(P,Q)$ contains a cross since $p_1,p_4,p_5,p_6,p_7,p_8>0$ and $q_1,q_2,q_4,q_6,q_7,q_8>0$. By extremality $(P,Q)$ has to be a cross in this case.
\end{enumerate}
\end{proof}

\begin{lem}
Let $Z\in \M_4\lb\lset 0,1\rset\rb$ denote a zero pattern containing either one of the following patterns up to permutations of rows and columns: 
\begin{enumerate}
\item \[
\begin{pmatrix}
1 & 1 & 1 & 1 \\ 1 & 1 & 0 & 0 \\ 0 & 0 & 1 & 0 \\ 0 & 0 & 0 & 1
\end{pmatrix}
\]
\item \[
\begin{pmatrix}
1 & 1 & 1 & 0 \\ 1 & 1 & 1 & 0 \\ 1 & 0 & 0 & 1 \\ 0 & 0 & 0 & 0
\end{pmatrix}
\]
\item \[
\begin{pmatrix}
1 & 1 & 0 & 1 \\ 1 & 1 & 1 & 0 \\ 1 & 1 & 0 & 0 \\ 0 & 0 & 0 & 0
\end{pmatrix}
\]
\item \[
\begin{pmatrix}
1 & 1 & 1 & 0 \\ 1 & 1 & 0 & 0 \\ 1 & 0 & 1 & 0 \\ 0 & 0 & 0 & 1
\end{pmatrix}
\]
\end{enumerate}
Any $(P,Q)\in \mathcal{S}\lb \CP_2\cap\coCP_2\rb$ generating an extremal ray and satisfying $Z\subseteq\mathcal{Z}(P)$ or $Z\subseteq\mathcal{Z}(Q)$ is a cross, a diagonal or a box. 
\label{lem:cases3}
\end{lem}

\begin{proof} Assume that $(P,Q)\in \mathcal{S}\lb \CP_2\cap\coCP_2\rb$ is generating an extremal ray and satisfying $Z\subseteq\mathcal{Z}(P)$. In the following we will go through the cases of the lemma one-by-one every time assuming that $Z$ contains the corresponding zero pattern. 
\begin{enumerate}
\item Using Lemma \ref{lem:boxRule} we find that 
\[
P=\begin{pmatrix} p_1 & p_2 & p_3 & p_4 \\ p_5 & p_6 & 0 & 0 \\ 0 & 0 & p_7 & 0 \\ 0 & 0 & 0 & p_8\end{pmatrix} \quad\text{ and }\quad Q=\begin{pmatrix} q_1 & q_2 & 0 & 0  \\ q_3 & q_4 & 0 & 0 \\ q_5 & q_6 & q_{7} & q_8  \\ q_9 & q_{10} & q_{11} & q_{12} \end{pmatrix}
\]
where $p_i\geq 0$ for any $i\in\lset 1,\ldots , 8\rset$ and $q_j\geq 0$ for any $j\in\lset 1,\ldots , 12\rset$. By Lemma \ref{lem:ZeroSubBoxZeroSubD} we would be done if $0\in \lset p_3,p_4\rset$ and by Lemma \ref{lem:6erLem} we would be done if $0\in\lset p_7,p_8\rset$. Note that if $0\in\lset p_5,p_6\rset$, then the remaining zero pattern would contain the pattern
\[
\begin{pmatrix}
* & * & * & * \\ * & 0 & 0 & 0 \\ * & 0 & * & 0 \\ 0 & 0 & 0 & *
\end{pmatrix}
\]
up to a column permutation. Then, by case 5 of Lemma \ref{lem:CasesCross} we conclude that $(P,Q)$ would be a cross, a diagonal, or a box. Similarly, if $0\in\lset p_1,p_2\rset$, then the remaining zero pattern would contain the pattern
\[
\begin{pmatrix}
0 & * & * & * \\ * & * & 0 & 0 \\ * & 0 & * & 0 \\ 0 & 0 & 0 & *
\end{pmatrix}
\]
up to a column permutation. Then, by case 1 of Lemma \ref{lem:CasesCross} we conclude that $(P,Q)$ would be a cross, a diagonal, or a box. 

Finally, assume that $p_i>0$ for any $i\in\lset 1,\ldots , 8\rset$. Then, since $Q=KPK$ we can use \eqref{equ:KPKQ} to find that
\begin{align*}
q_8 &= \frac{1}{4}(-p_7-p_8-p_4+p_3+p_1+p_2+p_5+p_6) \\
q_{12} &= \frac{1}{4}(p_7 + p_8 - p_4+p_3+p_1+p_2+p_5+p_6) \\
q_7 &= \frac{1}{4}(p_7+p_8+p_4-p_3+p_1+p_2+p_5+p_6) \\
q_{11} &= \frac{1}{4}(-p_7 - p_8 + p_4-p_3+p_1+p_2+p_5+p_6) .
\end{align*}
From this we can conclude that 
\begin{align*}
q_{12}-q_8 = q_{7}-q_{11} =  \frac{p_8+p_7}{2}>0 
\end{align*}
showing that $q_{12}>0$ and $q_7>0$. By extremality, either $(P,Q)$ is a diagonal or we have $q_1=q_2=0$ or $q_3=q_4=0$ or $q_1=q_3=0$ or $q_2=q_4=0$. In the two latter cases we are done by Lemma \ref{lem:6erLem}. The two first cases can be argued in the same way, and we will only write out the first one. Assuming that $q_1=q_2=0$ we can argue using extremality that either $(P,Q)$ is a cross, or $0\in \lset q_3,q_4,q_5,q_9\rset$. If $0\in \lset q_3,q_4\rset$ we would be done by Lemma \ref{lem:6erLem}. Assuming that $0\in \lset q_5,q_9\rset$ we can again argue by extremality that either $(P,Q)$ is a box, or $0\in \lset q_8,q_{11}\rset$. In the latter case we conclude that we have at least one of the following: 
\begin{enumerate}
\item $q_5=q_8=0$
\item $q_9=q_{11}=0$
\item $q_5=q_{11}=0$
\item $q_9=q_8=0$
\end{enumerate}
For $a.)$ and $b.)$ Lemma \ref{lem:boxRule} implies that $p_6=0$, and in the cases $c.)$ and $d.)$ Lemma \ref{lem:diagRule} shows that $p_2=0$. Therefore, all of the cases lead to a contradiction. 
\item Using Lemma \ref{lem:boxRule} we find that 
\[
P=\begin{pmatrix} p_1 & p_2 & p_3 & 0 \\ p_4 & p_5 & p_6 & 0 \\ p_7 & 0 & 0 & p_8 \\ 0 & 0 & 0 & 0\end{pmatrix} \quad\text{ and }\quad Q=\begin{pmatrix} 0 & q_1 & q_2 & 0  \\ 0 & q_3 & q_4 & 0 \\ q_5 & q_6 & q_{7} & q_8  \\ q_9 & q_{10} & q_{11} & q_{12} \end{pmatrix}
\]
where $p_i\geq 0$ for any $i\in\lset 1,\ldots , 8\rset$ and $q_j\geq 0$ for any $j\in\lset 1,\ldots , 12\rset$. By Lemma \ref{lem:ZeroSubBoxZeroSubD} we would be done if $0\in \lset p_1,p_4\rset$ and by Lemma \ref{lem:6erLem} we would be done if $0\in\lset p_7,p_8\rset$. Note that if $0\in\lset p_2,p_3,p_5,p_6\rset$, we can permute rows and columns to restrict to the case where $p_6=0$. Then, after permuting the third and fourth column the remaining zero pattern would contain the pattern
\[
\begin{pmatrix}
* & * & * & * \\ * & * & 0 & 0 \\ * & 0 & * & 0 \\ 0 & 0 & 0 & 0
\end{pmatrix}.
\]
By case 6 of Lemma \ref{lem:CasesCross} we conclude that $(P,Q)$ would be a cross, a diagonal, or a box. 

Finally, assume that $p_i>0$ for any $i\in\lset 1,\ldots , 8\rset$. Since $Q=KPK$ we can use \eqref{equ:KPKQ} to find that
\begin{align*}
q_8 &= \frac{1}{4}(p_8-p_7+\sum^6_{i=1} p_i )\\
q_{12} &= \frac{1}{4}(-p_8+p_7+\sum^6_{i=1} p_i) \\
q_9 &= \frac{1}{4}(p_8-p_7-p_1-p_4+p_2+p_3+p_5+p_6) \\
q_5 &= \frac{1}{4}(-p_8+p_7-p_1-p_4+p_2+p_3+p_5+p_6).
\end{align*}
Note that $q_{8}=0$ would imply that $q_9<0$ contradicting the assumption that $(P,Q)\in \mathcal{S}\lb \CP_2\cap\coCP_2\rb$, and similarly $q_{12}=0$ would imply that $q_5<0$ leading to the same contradiction. Therefore, we have $q_8>0$ and $q_{12}>0$. Since $(P,Q)$ generates an extremal ray, we either have that $(P,Q)$ is a box, or we have $0\in\lset q_5,q_9\rset$ and $0\in\lset q_6,q_{10}\rset$ and $0\in\lset q_7,q_{11}\rset$. Assuming that $(P,Q)$ is not a box, we would be done by Lemma \ref{lem:ZeroSubBoxZeroSubD} if $0\in\lset q_6,q_{11}\rset$ or $0\in\lset q_{10},q_7\rset$. Thus, we either have $q_6=q_7=0$ or $q_{10}=q_{11}=0$. 

Since $Q=KPK$ we can use \eqref{equ:KPKQ} to find that
\begin{align*}
q_{10} = \frac{1}{4}(-p_2-p_5 + p_3 + p_6 + p_1 + p_4 + p_7 + p_8)\\
q_{11} = \frac{1}{4}(-p_3-p_6+p_2+p_5+p_1+p_4+p_7+p_8).
\end{align*} 
This implies that $q_{10}+q_{11}>0$, and therefore we cannot have $q_{10}=q_{11}=0$. By the above discussion we conclude that $q_6=q_7=0$ and we also have $0\in\lset q_5,q_9\rset$ and $q_{10},q_{11}>0$. By extremality, either $(P,Q)$ is a cross, or we have $0\in\lset q_1,q_4\rset$ and $0\in\lset q_3,q_2\rset$. Assuming that $(P,Q)$ is not a cross, we conclude that one of the following cases holds true:
\begin{enumerate}
\item $q_1=q_2=0$
\item $q_3=q_{4}=0$
\item $q_1=q_{3}=0$
\item $q_2=q_4=0$
\end{enumerate}
In cases $c.)$ and $d.)$ we are done by Lemma \ref{lem:6erLem}. In cases $a.)$ and $b.)$ we can apply Lemma \ref{lem:boxRule} to conclude that $p_4=0$ or $p_1=0$. But this contradicts the assumption that $p_4,p_1>0$. 

\item Using Lemma \ref{lem:boxRule} we find that 
\[
P=\begin{pmatrix} p_1 & p_2 & 0 & p_3 \\ p_4 & p_5 & p_6 & 0 \\ p_7 & p_8 & 0 & 0 \\ 0 & 0 & 0 & 0\end{pmatrix} \quad\text{ and }\quad Q=\begin{pmatrix} 0 & 0 & q_1 & q_2  \\ 0 & 0 & q_3 & q_4 \\ q_5 & q_6 & q_{7} & q_8  \\ q_9 & q_{10} & q_{11} & q_{12} \end{pmatrix}
\]
where $p_i\geq 0$ for any $i\in\lset 1,\ldots , 8\rset$ and $q_j\geq 0$ for any $j\in\lset 1,\ldots , 12\rset$. By Lemma \ref{lem:6erLem} we would be done if $0\in\lset p_3,p_6,p_7,p_8\rset$. Note that if $0\in\lset p_1,p_2,p_4,p_5\rset$, we can permute rows and columns to restrict to the case where $p_2=0$. Then, after permuting the second and third column the remaining zero pattern would contain the pattern
\[
\begin{pmatrix}
* & * & * & * \\ * & * & 0 & 0 \\ * & 0 & * & 0 \\ 0 & 0 & 0 & 0
\end{pmatrix}.
\]
By case 6 of Lemma \ref{lem:CasesCross} we conclude that $(P,Q)$ would be a cross, a diagonal, or a box.  

Finally, assume that $p_i>0$ for any $i\in\lset 1,\ldots , 8\rset$. Since $Q=KPK$ we can use \eqref{equ:KPKQ} to find
\begin{align*}
q_{12} &= \frac{1}{4}(-p_3+ p_7+p_8+ p_1+p_2+p_4+p_5+p_6)\\
q_{8} &= \frac{1}{4}(-p_3-p_7-p_8+p_1+p_2+p_4+p_5+p_6) \\
q_{11} &= \frac{1}{4}(-p_6+p_7+p_8+p_1+p_2+p_3+p_4+p_5) \\
q_{7} &= \frac{1}{4}(-p_6-p_7-p_8+p_1+p_2+p_3+p_4+p_5).
\end{align*}    
From this we conclude that $q_{12}>0$ since otherwise $q_8<0$, and similarly $q_{11}>0$ since otherwise $q_7<0$. We also find that 
\begin{align*}
q_3 &= \frac{1}{4}(p_6 +p_3 - p_4-p_5+p_1+p_2+p_7+p_8) \\
q_4 &= \frac{1}{4}(-p_6-p_3-p_4-p_5 + p_1 + p_2 + p_7 + p_8) \\
q_2 &= \frac{1}{4}(p_3 -p_1-p_2+p_6 + p_4+p_5+p_7+p_8 )\\
q_1 &= \frac{1}{4}(-p_3-p_1-p_2-p_6 + p_4 + p_5 + p_7 + p_8 ).  
\end{align*}
From this we conclude that $q_3>0$ since otherwise $q_4<0$, and $q_2>0$ since otherwise $q_1<0$. By extremality, either $(P,Q)$ is a box, or we have $0\in\lset q_1,q_2,q_{11},q_{12}\rset$ and $0\in\lset q_3,q_4,q_{11},q_{12}\rset$ and $0\in\lset q_7,q_8,q_{11},q_{12}\rset$. As $q_{11},q_{12},q_2,q_3>0$ we find that $q_1 = q_4 = 0$ and $0\in\lset q_7,q_8\rset$. Again, by extremality we have that $(P,Q)$ is a cross or $0\in\lset q_6,q_{10}\rset$. By Lemma \ref{lem:boxRule} the case $q_6=0$ implies that either $p_4=0$ (if $q_7=0$) or that $p_1=0$ (if $q_8=0$) both contradicting our assumptions. Therefore, we must have $q_{10}=0$. Finally, by Lemma \ref{lem:6erLem} we are done if $q_9=0$ and by Lemma \ref{lem:diagRule} we find that $p_7=0$ in the case where $q_5=0$ leading to a contradiction. Therefore, we have $q_5,q_9,q_3,q_2,q_{11},q_{12}>0$ and $p_2,p_5,p_8,p_6,p_3,p_7>0$. But this means that $(P,Q)$ contains a cross, and by extremality $(P,Q)$ must be equal to that cross. 
\item Using Lemma \ref{lem:diagRule} we find that 
\[
P=\begin{pmatrix} p_1 & p_2 & p_3 & 0 \\ p_4 & p_5 & 0 & 0 \\ p_6 & 0 & p_7 & 0 \\ 0 & 0 & 0 & p_8\end{pmatrix} \quad\text{ and }\quad Q=\begin{pmatrix} q_1 & q_2 & q_3 & 0  \\ q_4 & q_5 & 0 & q_6 \\ q_7 & 0 & q_{8} & q_9  \\ 0 & q_{10} & q_{11} & q_{12} \end{pmatrix}
\]
where $p_i\geq 0$ for any $i\in\lset 1,\ldots , 8\rset$ and $q_j\geq 0$ for any $j\in\lset 1,\ldots , 12\rset$. By Lemma \ref{lem:ZeroSubBoxZeroSubD} we would be done if $0\in \lset p_2,p_3,p_4,p_5\rset$ and by Lemma \ref{lem:ZeroRowColumn} we would be done if $p_8=0$. Note that if $0\in\lset p_5,p_7\rset$, we can permute rows and columns to restrict to the case where $p_7=0$. Then, after permuting the third and fourth columns the remaining zero pattern would contain the pattern
\[
\begin{pmatrix}
* & * & * & * \\ * & * & 0 & 0 \\ * & 0 & 0 & 0 \\ 0 & 0 & * & 0
\end{pmatrix}.
\]
By case 5 of Lemma \ref{lem:CasesCross} we conclude that $(P,Q)$ would be a cross, a diagonal, or a box. Similarly, if $p_1=0$, then after permuting the first and fourth and the second and third columns the remaining zero pattern would contain the pattern
\[
\begin{pmatrix}
* & * & * & 0 \\ 0 & 0 & * & * \\ 0 & * & 0 & * \\ * & 0 & 0 & 0
\end{pmatrix}.
\]  
By case 1 of Lemma \ref{lem:CasesCross} we conclude that $(P,Q)$ would be a cross, a diagonal, or a box.

Finally, we assume that $p_i>0$ for any $i\in\lset 1,\ldots , 8\rset$. Since $Q=KPK$ we can use \eqref{equ:KPKQ} to find that
\[
q_{12} = \frac{1}{4}\sum^8_{i=1} p_i>0.
\]
Furthermore, we find that 
\begin{align*}
p_8+p_1+p_2+p_3 &=p_4+p_5+p_7+p_6 \\
p_8+p_1+p_2+p_6 &= p_4+p_5+p_7+p_3 \\
p_8+p_1+p_4+p_6 &=p_2+p_5+p_7+p_3 \\
p_8+p_1+p_4+p_3 &= p_2+p_5+p_7+p_6.
\end{align*}
From this we conclude first that $p_3=p_6$ and then that $p_2=p_4$. Therefore, $P$ is a symmetric matrix and since $Q=KPK$ it follows that $Q$ is symmetric as well. Moreover, we also find 
\begin{align*}
q_5 = \frac{1}{4}(-p_2-p_4+p_1+p_3+p_6+p_7+p_8)= \frac{1}{4}(p_1+p_7+p_8)>0.
\end{align*}
By extremality, we conclude that $(P,Q)$ is either equal to a diagonal, or $0\in\lset q_3,q_5,q_7,q_{12}\rset$. Assuming that $(P,Q)$ is not a diagonal, we conclude (using symmetry) that both $q_3=q_7=0$. Then, we would be done by Lemma \ref{lem:ZeroSubBoxZeroSubD} if $0\in\lset q_6,q_{10}\rset$. But if $q_6,q_{10}>0$, then $(P,Q)$ contains a box by the previous discussion. As $(P,Q)$ is extremal it must then be equal to a box.  
\end{enumerate}

\end{proof}

\bibliographystyle{IEEEtran}
\bibliography{mybibliography.bib}

\begin{thebibliography}{10}
\providecommand{\url}[1]{#1}
\csname url@samestyle\endcsname
\providecommand{\newblock}{\relax}
\providecommand{\bibinfo}[2]{#2}
\providecommand{\BIBentrySTDinterwordspacing}{\spaceskip=0pt\relax}
\providecommand{\BIBentryALTinterwordstretchfactor}{4}
\providecommand{\BIBentryALTinterwordspacing}{\spaceskip=\fontdimen2\font plus
\BIBentryALTinterwordstretchfactor\fontdimen3\font minus
  \fontdimen4\font\relax}
\providecommand{\BIBforeignlanguage}[2]{{%
\expandafter\ifx\csname l@#1\endcsname\relax
\typeout{** WARNING: IEEEtran.bst: No hyphenation pattern has been}%
\typeout{** loaded for the language `#1'. Using the pattern for}%
\typeout{** the default language instead.}%
\else
\language=\csname l@#1\endcsname
\fi
#2}}
\providecommand{\BIBdecl}{\relax}
\BIBdecl

\bibitem{stormer1963positive}
E.~St{\o}rmer, ``Positive linear maps of operator algebras,'' \emph{Acta
  Mathematica}, vol. 110, pp. 233--278, 1963.

\bibitem{woronowicz1976positive}
S.~L. Woronowicz, ``{Positive maps of low dimensional matrix algebras},''
  \emph{Reports on Mathematical Physics}, vol.~10, no.~2, pp. 165--183, 1976.

\bibitem{tang1986positive}
W.-S. Tang, ``On positive linear maps between matrix algebras,'' \emph{Linear
  algebra and its applications}, vol.~79, pp. 33--44, 1986.

\bibitem{HORODECKI19961}
M.~Horodecki, P.~Horodecki, and R.~Horodecki, ``Separability of mixed states:
  necessary and sufficient conditions,'' \emph{Physics Letters A}, vol. 223,
  no.~1, pp. 1 -- 8, 1996.

\bibitem{stormer1982decomposable}
E.~St{\o}rmer, ``{Decomposable positive maps on $C^\star$-algebras},''
  \emph{Proceedings of the American Mathematical Society}, vol.~86, no.~3, pp.
  402--404, 1982.

\bibitem{horodecki1997separability}
P.~Horodecki, ``Separability criterion and inseparable mixed states with
  positive partial transposition,'' \emph{Physics Letters A}, vol. 232, no.~5,
  pp. 333--339, 1997.

\bibitem{horodecki1998mixed}
M.~Horodecki, P.~Horodecki, and R.~Horodecki, ``Mixed-state entanglement and
  distillation: is there a `bound'' entanglement in nature?'' \emph{Physical
  Review Letters}, vol.~80, no.~24, p. 5239, 1998.

\bibitem{horodecki2005secure}
K.~Horodecki, M.~Horodecki, P.~Horodecki, and J.~Oppenheim, ``Secure key from
  bound entanglement,'' \emph{Physical review letters}, vol.~94, no.~16, p.
  160502, 2005.

\bibitem{smith2008quantum}
G.~Smith and J.~Yard, ``Quantum communication with zero-capacity channels,''
  \emph{Science}, vol. 321, no. 5897, pp. 1812--1815, 2008.

\bibitem{bauml2015limitations}
S.~B{\"a}uml, M.~Christandl, K.~Horodecki, and A.~Winter, ``Limitations on
  quantum key repeaters,'' \emph{Nature communications}, vol.~6, p. 6908, 2015.

\bibitem{aubrun2015two}
G.~Aubrun and S.~J. Szarek, ``{Two proofs of St\o rmer's theorem},''
  \emph{arXiv preprint arXiv:1512.03293}, 2015.

\bibitem{gurvits2004classical}
L.~Gurvits, ``Classical complexity and quantum entanglement,'' \emph{Journal of
  Computer and System Sciences}, vol.~69, no.~3, pp. 448--484, 2004.

\bibitem{aubrun2017alice}
G.~Aubrun and S.~J. Szarek, ``{Alice and Bob meet Banach},'' \emph{Mathematical
  Surveys and Monographs}, vol. 105, 2017.

\bibitem{filippov2017positive}
S.~N. Filippov and K.~Y. Magadov, ``{Positive tensor products of maps and
  n-tensor-stable positive qubit maps},'' \emph{Journal of Physics A:
  Mathematical and Theoretical}, vol.~50, no.~5, p. 055301, 2017.

\bibitem{ruskai2002analysis}
M.~B. Ruskai, S.~Szarek, and E.~Werner, ``{An analysis of completely-positive
  trace-preserving maps on M2},'' \emph{Linear algebra and its applications},
  vol. 347, no. 1-3, pp. 159--187, 2002.

\bibitem{fujiwara1999one}
A.~Fujiwara and P.~Algoet, ``One-to-one parametrization of quantum channels,''
  \emph{Physical Review A}, vol.~59, no.~5, p. 3290, 1999.

\bibitem{choi1975completely}
M.-D. Choi, ``{Completely positive linear maps on complex matrices},''
  \emph{Linear algebra and its applications}, vol.~10, no.~3, pp. 285--290,
  1975.

\bibitem{horodecki2003entanglement}
M.~Horodecki, P.~W. Shor, and M.~B. Ruskai, ``Entanglement breaking channels,''
  \emph{Reviews in Mathematical Physics}, vol.~15, no.~06, pp. 629--641, 2003.

\bibitem{rockafellar1970convex}
R.~T. Rockafellar, \emph{Convex analysis}.\hskip 1em plus 0.5em minus
  0.4em\relax Princeton university press, 1970, no.~28.

\bibitem{rudolph2000separability}
O.~Rudolph, ``A separability criterion for density operators,'' \emph{Journal
  of Physics A: Mathematical and General}, vol.~33, no.~21, p. 3951, 2000.

\bibitem{rudolph2005further}
------, ``Further results on the cross norm criterion for separability,''
  \emph{Quantum Information Processing}, vol.~4, no.~3, pp. 219--239, 2005.

\bibitem{chen2002matrix}
K.~Chen, L.~Yang, and L.~Wu, ``A matrix realignment method for recognizing
  entanglement,'' \emph{Quantum Inf. Comput.}, vol.~3, pp. 193--202, 2002.

\bibitem{lupo2008bipartite}
C.~Lupo, P.~Aniello, and A.~Scardicchio, ``Bipartite quantum systems: on the
  realignment criterion and beyond,'' \emph{Journal of Physics A: Mathematical
  and Theoretical}, vol.~41, no.~41, p. 415301, 2008.

\bibitem{johnston2018inverse}
N.~Johnston and E.~Patterson, ``The inverse eigenvalue problem for entanglement
  witnesses,'' \emph{Linear Algebra and its Applications}, vol. 550, pp. 1--27,
  2018.

\bibitem{horodecki2000operational}
P.~Horodecki, M.~Lewenstein, G.~Vidal, and I.~Cirac, ``Operational criterion
  and constructive checks for the separability of low-rank density matrices,''
  \emph{Physical Review A}, vol.~62, no.~3, p. 032310, 2000.

\bibitem{MPT3}
M.~Herceg, M.~Kvasnica, C.~Jones, and M.~Morari, ``{Multi-Parametric Toolbox
  3.0},'' in \emph{Proc.~of the European Control Conference}, Z\"urich,
  Switzerland, July 17--19 2013, pp. 502--510,
  \url{http://control.ee.ethz.ch/~mpt}.

\bibitem{polymake1}
E.~Gawrilow and M.~Joswig, ``\texttt{polymake}: a framework for analyzing
  convex polytopes,'' in \emph{Polytopes---combinatorics and computation
  ({O}berwolfach, 1997)}, ser. DMV Sem.\hskip 1em plus 0.5em minus 0.4em\relax
  Birkh\"auser, Basel, 2000, vol.~29, pp. 43--73.

\bibitem{polymake2}
B.~Assarf, E.~Gawrilow, K.~Herr, M.~Joswig, B.~Lorenz, A.~Paffenholz, and
  T.~Rehn, ``Computing convex hulls and counting integer points with
  \texttt{polymake},'' \emph{Math. Program. Comput.}, vol.~9, no.~1, pp. 1--38,
  2017.

\bibitem{christandl2012PPT}
M.~Christandl, ``{PPT square conjecture},'' \emph{Banff International Research
  Station workshop: Operator structures in Quantum Information Theory}, 2012.

\bibitem{muller2018PPT}
M.~Christandl, A.~M{\"u}ller-Hermes, and M.~M. Wolf, ``When do composed maps
  become entanglement breaking?'' \emph{Annales Henri Poincar{\'e}}, vol.~20,
  no.~7, pp. 2295--2322, 2019.

\bibitem{kennedy2018composition}
M.~Kennedy, N.~A. Manor, and V.~I. Paulsen, ``{Composition of PPT maps},''
  \emph{Quantum Information \& Computation}, vol.~18, no. 5-6, pp. 472--480,
  2018.

\bibitem{rahaman2018eventually}
M.~Rahaman, S.~Jaques, and V.~I. Paulsen, ``Eventually entanglement breaking
  maps,'' \emph{Journal of Mathematical Physics}, vol.~59, no.~6, p. 062201,
  2018.

\bibitem{collins2018ppt}
B.~Collins, Z.~Yin, and P.~Zhong, ``{The PPT square conjecture holds
  generically for some classes of independent states},'' \emph{Journal of
  Physics A: Mathematical and Theoretical}, vol.~51, no.~42, p. 425301, 2018.

\bibitem{chen2019positive}
L.~Chen, Y.~Yang, and W.-S. Tang, ``Positive-partial-transpose square
  conjecture for n= 3,'' \emph{Physical Review A}, vol.~99, no.~1, p. 012337,
  2019.

\bibitem{hanson2020eventually}
E.~P. Hanson, C.~Rouz{\'e}, and D.~S. Fran{\c{c}}a, ``Eventually entanglement
  breaking markovian dynamics: Structure and characteristic times,''
  \emph{Annales Henri Poincar{\'e}}, pp. 1--55, 2020.

\bibitem{stormer1986extension}
E.~St{\o}rmer, ``{Extension of positive maps into $B(H)$},'' \emph{Journal of
  Functional Analysis}, vol.~66, no.~2, pp. 235--254, 1986.

\bibitem{skowronek2009cones}
{\L}.~Skowronek, E.~St{\o}rmer, and K.~{\.Z}yczkowski, ``Cones of positive maps
  and their duality relations,'' \emph{Journal of Mathematical Physics},
  vol.~50, no.~6, p. 062106, 2009.

\bibitem{muller2018decomposability}
A.~M{\"u}ller-Hermes, ``Decomposability of linear maps under tensor powers,''
  \emph{Journal of Mathematical Physics}, vol.~59, no.~10, p. 102203, 2018.

\bibitem{muller2016positivity}
A.~M{\"u}ller-Hermes, D.~Reeb, and M.~M. Wolf, ``{Positivity of linear maps
  under tensor powers},'' \emph{Journal of Mathematical Physics}, vol.~57,
  no.~1, p. 015202, 2016.

\bibitem{divincenzo2003unextendible}
D.~P. DiVincenzo, T.~Mor, P.~W. Shor, J.~A. Smolin, and B.~M. Terhal,
  ``Unextendible product bases, uncompletable product bases and bound
  entanglement,'' \emph{Communications in Mathematical Physics}, vol. 238,
  no.~3, pp. 379--410, 2003.

\bibitem{fulvio20xx}
F.~Gesmundo and A.~M{\"u}ller-Hermes, ``In preparation,'' 2020.

\bibitem{BRUALDI20063054}
R.~A. Brualdi, ``Algorithms for constructing (0,1)-matrices with prescribed row
  and column sum vectors,'' \emph{Discrete Mathematics}, vol. 306, no.~23, pp.
  3054 -- 3062, 2006.

\bibitem{fulton1997young}
W.~Fulton, \emph{Young tableaux: with applications to representation theory and
  geometry}.\hskip 1em plus 0.5em minus 0.4em\relax Cambridge University Press,
  1997, vol.~35.

\bibitem{kostka1882ueber}
C.~Kostka, ``{\"Uber den Zusammenhang zwischen einigen Formen von symmetrischen
  Functionen.}'' \emph{{Journal f{\"u}r die reine und angewandte Mathematik}},
  vol.~93, pp. 89--123, 1882.

\end{thebibliography}

\end{document}